\documentclass[prx,twocolumn,english,superscriptaddress,floatfix,longbibliography]{revtex4-2}

\usepackage{amsthm}
\usepackage{graphicx}% Include figure files
\usepackage{dcolumn}% Align table columns on decimal point
\usepackage{bm}% bold math
\usepackage[utf8]{inputenc}
\usepackage{mathptmx}
\usepackage{multirow}
\usepackage{braket}
\usepackage[english]{babel}
\usepackage{xcolor}% Include colors for document elements
\colorlet{RED}{red}
\colorlet{BLUE}{blue}
\usepackage{dcolumn}% Align table columns on decimal pointhttps://v1.overleaf.com/17481932tmwjwzytcbcq#
\usepackage{bm}% bold math
\usepackage[version=4]{mhchem} % HvD: for general chemical formula like \ce{(TiO2)n}
\usepackage{acronym}
\usepackage{adjustbox}
\usepackage{float}
\usepackage{microtype} % fixes many of the overflow hbox errors

\definecolor{background-color}{gray}{0.98}
\usepackage[margin=2.3cm,bmargin=1cm,footnotesep=1cm]{geometry}

\usepackage{tikz}
\usetikzlibrary{automata,positioning}
\usetikzlibrary{shapes.geometric, arrows}
\tikzstyle{startstop} = [rectangle, rounded corners, text centered, draw=black, fill=none]
\tikzstyle{io} = [trapezium, trapezium left angle=70, trapezium right angle=110, text centered, draw=black, fill=blue!30]
\tikzstyle{process} = [rectangle, text centered, draw=black, fill=none]
\tikzstyle{decision} = [diamond, aspect=1.5, text centered, draw=black, fill=none]
\tikzstyle{arrow} = [thick,->,>=stealth]

\usepackage{physics}
\usepackage{bm}% bold math
\usepackage{xcolor}
\usepackage{wrapfig}
\usepackage{float}

\newtheorem{theorem}{Theorem}
\newtheorem{lemma}[theorem]{Lemma}
\newtheorem{corollary}[theorem]{Corollary}

\begin{document}

% \title{Qubitization and Quantum Simulation of Boson-Related Hamiltonians: Techniques, Effective Hamiltonian Construction, and Error Analysis}
\title{Quantum Simulation of Boson-Related Hamiltonians: Techniques, Effective Hamiltonian Construction, and Error Analysis}

\author{Bo Peng} 
\email{peng398@pnnl.gov}
\affiliation{Physical Sciences and Computational Division, Pacific Northwest National Laboratory, Richland, WA 99354, United States of America}

\author{Yuan Su}
\email{yuan.su@microsoft.com}
\affiliation{Microsoft, Azure Quantum, Redmond, Washington 98052, USA}

\author{Daniel Claudino}
\affiliation{Computational Sciences and Engineering Division, Oak Ridge National\ Laboratory, Oak Ridge, TN, 37831, USA}

\author{Karol Kowalski} 
\affiliation{Physical Sciences and Computational Division, Pacific Northwest National Laboratory, Richland, WA 99354, United States of America}

\author{Guang Hao Low}
\affiliation{Microsoft, Azure Quantum, Redmond, Washington 98052, USA}

\author{Martin Roetteler}
\affiliation{Microsoft, Azure Quantum, Redmond, Washington 98052, USA}

\date{\today}

\begin{abstract}
Elementary quantum mechanics proposes that a closed physical system consistently evolves in a reversible manner. However, control and readout necessitate the coupling of the quantum system to the external environment, subjecting it to relaxation and decoherence. Consequently, system-environment interactions are indispensable for simulating physically significant theories. A broad spectrum of physical systems in condensed-matter and high-energy physics, vibrational spectroscopy, and circuit and cavity QED necessitates the incorporation of bosonic degrees of freedom, such as phonons, photons, and gluons, into optimized fermion algorithms for near-future quantum simulations. In particular, when a quantum system is surrounded by an external environment, its basic physics can usually be simplified to a spin or fermionic system interacting with bosonic modes. Nevertheless, troublesome factors such as the magnitude of the bosonic degrees of freedom typically complicate the direct quantum simulation of these interacting models, necessitating the consideration of a comprehensive plan. This strategy should specifically include a suitable fermion/boson-to-qubit mapping scheme to encode sufficiently large yet manageable bosonic modes, and a method for truncating and/or downfolding the Hamiltonian to the defined subspace for performing an approximate but highly accurate simulation, guided by rigorous error analysis. In this pedagogical tutorial review, we aim to provide such an exhaustive strategy, focusing on encoding and simulating certain bosonic-related model Hamiltonians, inclusive of their static properties and time evolutions. Specifically, we emphasize two aspects: (1) the discussion of recently developed quantum algorithms for these interacting models and the construction of effective Hamiltonians, and (2) a detailed analysis regarding a tightened error bound for truncating the bosonic modes for a class of fermion-boson interacting Hamiltonians.
\end{abstract}

\maketitle

%\tableofcontents
%%%%%%%%%%%%%%%%%%%%%%%%%%%%%%%%%%%%%%%%%%%%%%%%%%%%%%%%%%%%%

\section{Introduction}

%Introduce the topic of qubitization of boson operators and quantum simulation of boson-related Hamiltonians, explaining their significance and relevance in the field of quantum computing. 

%Briefly mention the different sections of the article and their importance, including the review of existing literature, the methodology, and the research component on error analysis.

Understanding the physics of open quantum systems, which involve crucial interactions between a quantum system and its environment, is a highly non-trivial and sometimes challenging problem~\cite{Breuer2002}. The past few decades have seen monumental advances in classical computing for simulating quantum systems in various fields, including quantum mechanics, molecular dynamics, quantum chemistry, condensed matter physics, and quantum field theory~\cite{Levine2009,Tuckerman2010,Atkins2010,Kittel2005,Peskin1995}. However, despite this progress, the rapidly increasing computational cost makes it unrealistic to fully treat quantum many-body effects, particularly when associated with strong system-environment interactions, using classical computing~\cite{Nielsen2010,Lloyd1996}.

Specifically, common approaches such as the Lindblad master equation~\cite{gorini1976completely}, Redfield equation~\cite{Breuer2002}, and quantum Monte Carlo methods~\cite{foulkes2001quantum}, which are often employed to treat weak interactions in the open quantum systems, rely on certain approximations that may not be valid in all situations or for all types of systems. For strong system-environment interactions, one would usually resort to methods like non-Markovian quantum master equations~\cite{breuer2009measure}, exact diagonalization methods~\cite{white2004real}, path integral methods~\cite{feynman1963theory}, hierarchical equations of motion~\cite{tanimura2006hierarchical}, tensor network approaches~\cite{schollwock2011density}, and variational approaches~\cite{tomasi2005quantum}. Nevertheless, these methods can be computationally demanding and may require additional approximations or simplifications. Consequently, there is a growing need for novel techniques and algorithms that can efficiently simulate quantum systems with both weak and strong system-environment interactions. 

This dilemma can be solved using computational devices that build upon the laws of quantum mechanics themselves. Indeed, the efficient simulation of quantum systems is one of the main motivations for Feynman and others to propose the idea of quantum computers~\cite{Fey82,Manin80}. Overall the years, many quantum simulation algorithms have been developed based on product formulas~\cite{Lloyd1996,BACS05} as well as more advanced techniques~\cite{berry2015simulating,Low2019hamiltonian} that can be deployed on a scalable quantum computer.
Quantum simulations have also recently emerged as a promising playground for noisy intermediate-scale quantum (NISQ) demonstrations of quantum advantage compared to classical computing~\cite{Preskill2018,Barreiro_2011,Garc_a_P_rez_2020,PhysRevB.102.125112,rost2021demonstrating,PRXQuantum.3.010320}.

Nevertheless, quantum simulations of open quantum systems with strong system-environment interaction are not straightforward and require careful consideration of various aspects (Figure \ref{fig:problem}). For example, the size and complexity of the system, as well as the number of environmental degrees of freedom, can significantly impact these considerations~\cite{Nielsen2010,schollwock2011density,georgescu2014quantum,RevModPhys.93.015008}. Furthermore, the simulations on the NISQ devices also requires robust error mitigation/correction and software infrastructure~\cite{gottesman1997stabilizer,Corcoles_2015}. In the context of these multifaceted considerations, it becomes imperative to direct concerted attention towards model Hamiltonian selection, quantum algorithms, and efficient pre-processing and post-processing steps. These areas form the crucial foundation that directly impacts the efficiency, accuracy, and practical applicability of the simulations. 

\begin{figure}
    \centering
    \includegraphics[width=\linewidth]{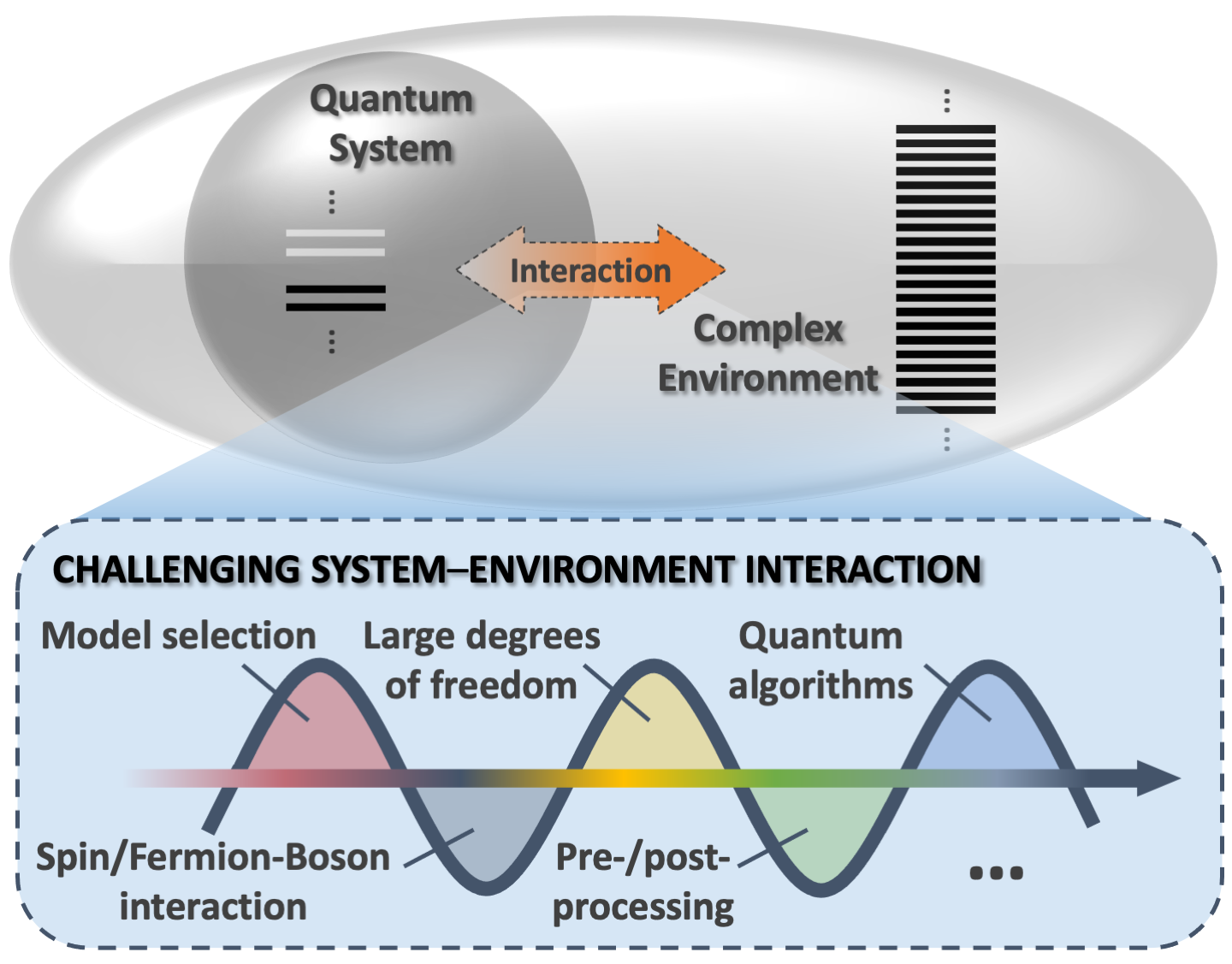}
    \caption{System–environment interactions in an open quantum system comprising a quantum system in a complex environment can pose a challenging problem. }
    \label{fig:problem}
\end{figure}

Regarding the choice of model Hamiltonian, it determines the level of abstraction and the trade-off between computational efficiency and the accuracy of the representation of the physical system~\cite{busemeyer2012quantum,RevModPhys.76.1267}. A suitable model can capture essential physics while still being amenable to efficient simulation on quantum hardware. For example, decoherence or measurement is often said to cause a system to become entangled with its environment~\cite{weiss2021quantum}. Such entanglement can often be studied through the simple spin-boson model~\cite{Leggett1987}, where the system-environment interaction is abstracted through a finite number of qubits interacting with a collection of harmonic oscillators. Here, the model Hamiltonians  are able to describe ultra-strong coupling regime that has been investigated experimentally in many scenarios, such as circuit QED~\cite{PhysRevLett.105.237001,Forn_2017,Niemczyk_2010,Braumuller_2017,Yoshihara_2017,Langford_2017}, trapped ions~\cite{PhysRevX.8.021027}, photonic~\cite{PhysRevLett.108.163601} and semiconductor quantum systems~\cite{PhysRevLett.102.186402,Gunter_2009}.
When dealing with larger quantum systems, such as many-fermion systems coupled with the environment, as often encountered in quantum chemistry and condensed matter physics, the complexity of the systems requires more convoluted model Hamiltonians featuring explicit system-environment interactions for physically important theories~\cite{giuliani2008quantum}. For example, electron-phonon interactions~\cite{Giustino2017,ALLEN19831} in the context of nonrelativistic quantum field theory are usually characterized through a fermion-boson interacting term in the model Hamiltonians. This is because phonons, the most common bosonic excitations in solids, can usually interact with electrons to significantly renormalize the electrical and transport properties of materials or lead to dramatic effects, such as superconductivity or Jahn-Teller distortions~\cite{mahan2013many,fetter2012quantum}. Similar Hamiltonians with explicit system-environment treatment can also be employed to address the interaction of electrons with other bosonic collective excitations in solids, such as spin, orbital, and charge.

In addition to model selection, the choice of quantum algorithms, together with efficient pre-/post-processing techniques, is also crucial to harness the full potential of quantum computing in simulating open quantum systems~\cite{Nielsen2010}. Efficient algorithms should lead to significant speed-ups over classical approaches, especially for problems involving strong system-environment interactions. For close quantum systems, quantum algorithms such as the Quantum Phase Estimation (QPE)~\cite{kitaev1995quantum} and the Variational Quantum Eigensolver (VQE)~\cite{Peruzzo_2014}, as well as their numerous variants, have often been applied to study quantum systems and estimate their properties, including energy levels and ground states. However, the development of algorithms specifically tailored to handle strong system-environment interactions in open quantum systems remains an active area of research 
% (see Refs. \citenum{Huh_2014,PhysRevResearch.2.023026,PhysRevLett.125.010501,Yuan2019theoryofvariational,liu2020solving,Haug_2022,Rakita_2020} for some recent developments, and see Ref. \citenum{RevModPhys.94.015004} for a recent review). 
(see Refs. \citenum{Huh_2014,PhysRevResearch.2.023026,PhysRevLett.125.010501,Yuan2019theoryofvariational,liu2020solving,Haug_2022,Rakita_2020} for some recent developments in NISQ algorithms and Ref. \citenum{RevModPhys.94.015004} for a recent review, and see Refs.~\citenum{cleve_et_al,childs2017efficient} and papers citing and cited by these work for quantum simulation algorithms for scalable quantum computers). 
Generally speaking, the algorithm being developed in this direction should take into account the unique characteristics of open quantum systems, such as their non-Markovian dynamics~\cite{breuer2009measure} and the interplay between decoherence and entanglement~\cite{RevModPhys.76.1267,RevModPhys.81.865}.

As can be seen, when considering explicit system-environment interactions, bosonic modes or harmonic oscillators in different forms are typically employed. However, encoding bosonic modes is quite different from encoding fermions~\cite{PhysRevX.6.031006}. Direct encoding of bosonic modes requires a large number of qubits, especially in intermediate and strong boson-fermion coupling regimes where encoding the Hamiltonian becomes more convoluted and resource-demanding than pure fermionic ones~\cite{PhysRevX.8.041015}. Furthermore, these challenges directly hinder a formal error analysis of quantum simulations of these model Hamiltonians~\cite{RevModPhys.87.307}. Remarkably, despite recent advances in quantum hardware technology ushering in a new era of computing science~\cite{Preskill2018,doi:10.1021/acs.jctc.4c00544,crane2024hybridoscillatorqubitquantumprocessors}, a comprehensive yet efficient methodology is still lacking in this field. This gap serves as the primary motivation for this paper, aiming to provide a relevant recipe in the rapidly evolving field of quantum computing.

This tutorial reivew is organized as follows. In the 'Qubit Mappings and Quantum Simulation Techniques' section, we describe existing methods for mapping boson operators onto qubit systems and discuss techniques customized for the quantum simulation of various types of boson-related Hamiltonians. In the 'Effective Hamiltonian Construction' section, we present a detailed explanation of constructing the effective Hamiltonian through unitary or non-unitary flows, and discuss the design of new hybrid quantum algorithms that approach the ground state and construct the effective Hamiltonian of a bosonic quantum system. In the 'Error Analysis of Truncating Bosonic Mode' section, we focus on the error analysis of truncating bosonic modes and present mathematical derivations and results for a class of fermion-boson interacting Hamiltonians. In the 'Conclusion' section, we summarize the main findings and provide an outlook.

%%%%%%%%%%%%%%%%%%%%%%%%%%%%%%%%%%%%%%%%%%%%%%%%%%%%%%%

% \section{Qubitization and Quantum Simulation Techniques}\label{sec: operators}
\section{Qubit Mappings and Quantum Simulation Techniques}\label{sec: operators}

%\textcolor{red}{Perhaps we can just say qubit mappings to avoid potential confusions. ``Qubitization'' will be used to refer to the quantum simulation algorithm of Low and Chuang.}
Consider an $N$-site open quantum system, where the system-environment interaction is simply represented by each site interacting with a bosonic mode. For each bosonic mode, since there is no limit to the number of bosons, one will need to select a finite upper limit number of bosons, $N_b$, for every site in the practical simulation, i.e., site $i$ can be occupied by $n_i$ ($0\le n_i \le N_b$) bosons.
We use $\hat{b}^\dagger_i$ and $\hat{b}_i$ to denote the bosoinc creation and annihilation operators at site $i$, and $\hat{n}_i = \hat{b}^\dagger_i \hat{b}_i$ denotes the number operator. The bosonic commutation relations (in an infinite-dimentional Hilbert space) are 
\begin{equation}
    [\hat{b}^\dagger_i,\hat{b}^\dagger_j] = 0,~~[\hat{b}_i,\hat{b}_j] = 0,~~ [\hat{b}_i,\hat{b}^\dagger_j] = \delta_{ij}. \label{boson_comm}
\end{equation}
The bosonic state with $n_i$ ($0\le n_i \le N_b$) bosons at the $i$-th site is represented as a natural tensor product structure
\begin{equation}
    | n_1,\cdots,n_i,\cdots,n_N \rangle = | n_1 \rangle \otimes \cdots \otimes | n_i \rangle \otimes \cdots \otimes | n_{N} \rangle. \label{boson_state0} 
\end{equation}
When $\hat{b}^\dagger_i$, $\hat{b}_i$, and $\hat{n}_i$ are acting on the bosonic state, we have
\begin{align}
\hat{b}^\dagger_i | n_1,\cdots,n_i,\cdots,n_N \rangle &= ~~\sqrt{n_i+1} ~~| n_1,\cdots,n_i+1,\cdots,n_N \rangle, \notag \\
\hat{b}_i | n_1,\cdots,n_i,\cdots,n_N \rangle &= ~~\sqrt{n_i} ~~| n_1,\cdots,n_i-1,\cdots,n_N \rangle, \notag \\
\hat{n}_i | n_1,\cdots,n_i,\cdots,n_N \rangle &=  ~~n_i ~~| n_1,\cdots,n_i,\cdots,n_N \rangle. \label{boson_op1}
\end{align}
Similar to (\ref{boson_state0}), we can write the tensor product representations of $\hat{b}^\dagger_i$, $\hat{b}_i$ and $\hat{n}_i$,
\begin{align}
\hat{b}^\dagger_i &= \mathbf{I}_1 \otimes \cdots \otimes \tilde{b}^\dagger_i \otimes \cdots \otimes \mathbf{I}_N, \notag \\
\hat{b}_i &= \mathbf{I}_1 \otimes \cdots \otimes \tilde{b}_i \otimes \cdots \otimes \mathbf{I}_N, \notag \\
\hat{n}_i &= \mathbf{I}_1 \otimes \cdots \otimes \tilde{n}_i \otimes \cdots \otimes \mathbf{I}_N,  \label{boson_op2}
\end{align}
where $\mathbf{I}_i$ represents an identity operation at qubit $i$, and $\tilde{b}^\dagger_i$, $\tilde{b}_i$, and $\tilde{n}_i$ satisfy
\begin{align}
    \tilde{b}^\dagger_i ~ | n_i \rangle &= \sqrt{n_i + 1} ~ | n_i + 1 \rangle, \notag \\
    \tilde{b}_i ~ | n_i \rangle &= ~~~\sqrt{n_i} ~~~| n_i - 1 \rangle, \notag \\
    \tilde{n}_i ~ | n_i \rangle &= ~~~~~~n_i ~~~| n_i \rangle. \label{boson_op3}
\end{align}
Bosonic operators are usually employed to re-express the \textbf{Harmonic Oscillators}. 
\textcolor{black}{The harmonic oscillator is a fundamental concept in physics due to its ubiquitous presence in nature and its mathematical simplicity. Its importance stems from the fact that any system near a stable equilibrium can be approximated as a harmonic oscillator, making it applicable to diverse phenomena ranging from mechanical vibrations and molecular bonds to electromagnetic fields and solid-state physics. Furthermore, the quantum harmonic oscillator has analytical solutions, making it a cornerstone for understanding quantum mechanics and quantum field theory.
When studying environmental effects on quantum systems, the environment is often modeled as a collection of harmonic oscillators. Bosonic creation and annihilation operators provide a powerful and compact way to represent these oscillators. Their simple commutation relations greatly simplify calculations, and their connection to quantum field theory allows for the application of advanced theoretical techniques. Thus, bosonic operators are the standard tool for analyzing the interaction of a quantum system with an environment modeled as a bath of harmonic oscillators.}
Take the Hamiltonian of the linear harmonic oscillator as an example
\begin{align}
    \hat{H} = \frac{\hat{p}^2}{2m} + \frac{1}{2}m\omega^2 \hat{q}^2,
\end{align}
%
%\textcolor{red}{$X,P$ should be $p,q$?}
where $\hat{p}$ and $\hat{q}$ are the coordinate and momentum Hermitian operators satisfying the canonical commutation relation
\begin{align}
    [\hat{q},\hat{p}] = i\hbar.
\end{align}
The $\hat{q}$ and $\hat{p}$ operators are mapped to bosonic operators through
\begin{align}
    \hat{b}^\dagger &= \left( \begin{array}{cc}
    \sqrt{\frac{m\omega}{2\hbar}} & \frac{-i}{\sqrt{2m\omega \hbar}}  
    \end{array} \right) \left( \begin{array}{c}
    \hat{q}  \\ \hat{p}
    \end{array} \right), \notag \\
    \hat{b} &= \left( \begin{array}{cc}
    \sqrt{\frac{m\omega}{2\hbar}} & \frac{i}{\sqrt{2m\omega \hbar}}  
    \end{array} \right) \left( \begin{array}{c}
    \hat{q}  \\ \hat{p}
    \end{array} \right), \label{eq:harmonic_oscillator_map1}
\end{align}
or equivalently
\begin{align}
    \hat{q} = \sqrt{\frac{\hbar}{2m\omega}} (\hat{b} + \hat{b}^\dagger),~~
    \hat{p} = \sqrt{\frac{m\omega \hbar}{2}} \frac{(\hat{b} - \hat{b}^\dagger)}{i}, \label{eq:harmonic_oscillator_map2}
\end{align}
from which the Hamiltonian takes a simple diagonal form
\begin{align}
    \hat{H} = \hbar \omega (\hat{b}^\dagger \hat{b} + \frac{1}{2}).
\end{align}
One can straightforwardly extend the above mapping to the case of many harmonic oscillators. Consider a system of $N_h$ identical linear harmonic oscillators of mass $M$ and frequency $\omega$ with coordinates $\hat{\textbf{Q}} = \{\hat{Q}_j, j=1,\cdots,N_h\}$ and momenta $\hat{\textbf{P}} = \{\hat{P}_j, j=1,\cdots,N_h\}$ also satisfying the commutation relations
\begin{align}
    \left\{\begin{array}{l}\left[\hat{Q}_j,\hat{Q}_k\right] = \left[\hat{P}_j,\hat{P}_k\right] = 0, \\ \left[\hat{Q}_j,\hat{P}_k\right] = i\hbar \delta_{jk}
    \end{array}\right. ,~~j,k \in [1,\cdots,N_h]. \label{eq:harmonic_oscillator_comm}
\end{align}
The Hamiltonian then reads
\begin{align}
    \hat{H} = \sum_{j=1}^{N_h} \frac{\hat{P}_j^2}{2M_j} + \frac{1}{2} \sum_{j,k=1}^{N_h} V_{jk} \hat{Q}_j \hat{Q}_k
\end{align}
with matrix $\textbf{V}$ a positive semidefinite Hermitian matrix. Construct a matrix $\mathbf{v}$ with elements
\begin{align}
    v_{jk} = V_{jk}/\sqrt{M_jM_k},
\end{align}
and two vectors
\begin{align}
    \hat{\textbf{p}} = \{\hat{p}_j, j =1,\cdots,N_h\},~~
    \hat{\textbf{q}} = \{\hat{q}_j, j =1,\cdots,N_h\} 
\end{align}
with elements
\begin{align}
    \hat{q}_j = \sqrt{M_j}\hat{Q}_j,~~\hat{p}_j = \hat{P}_j/\sqrt{M_j},
    ~~\text{s.t.}~~[\hat{q}_j,\hat{p}_k] = i\hbar \delta_{jk},
\end{align}
the Hamiltonian can then be rewritten as
\begin{align}
    \hat{H} =\frac{1}{2}\hat{\textbf{p}}^T\hat{\textbf{p}} + \frac{1}{2} \hat{\textbf{q}}^T \mathbf{v} \hat{\textbf{q}}.
\end{align}
Since the matrix $\mathbf{V}$ is a Hermitian matrix that is positive semidefinite, so is the matrix $\mathbf{v}$, and we can always find a unitary matrix to diagonalize a Hermitian matrix (via the finite-dimensional spectral theorem) with non-negative eigenvalues
\begin{align}
    \mathbf{v} = \mathbf{U} \Omega \mathbf{U}^\dagger, 
\end{align}
%
%\textcolor{red}{Need additional assumption that v is positive semidefinite.}
where $\Omega$ is a diagonal matrix with the non-negative diagonal elements $\{\omega_j^2\ge 0, j=1 ,\cdots, N_h\}$. The unitary matrix when acting on $\mathbf{q}$ and $\mathbf{p}$ generates the \textbf{normal mode} ($\mathbf{\tilde{q}}$ and $\mathbf{\tilde{p}}$)
\begin{align}
    \mathbf{\tilde{q}} = \mathbf{U} \hat{\mathbf{q}},~~    
    \mathbf{\tilde{p}} = \mathbf{U} \hat{\mathbf{p}}
\end{align}
that (i) preserves the commutation relations (\ref{eq:harmonic_oscillator_comm}), and (ii) re-expresses the Hamiltonian as
\begin{align}
    \hat{H} = \frac{1}{2} \left( \mathbf{\tilde{p}}^T\mathbf{\tilde{p}}+ \mathbf{\tilde{q}}^T\Omega\mathbf{\tilde{q}} \right).
\end{align}
Now we can map from the normal modes to bosonic operators similar to (\ref{eq:harmonic_oscillator_map1},\ref{eq:harmonic_oscillator_map2})
\begin{align}
    \hat{b}_j^\dagger &= \left( \begin{array}{cc}
    \sqrt{\frac{\omega_j}{2\hbar}} & \frac{-i}{\sqrt{2\omega_j \hbar}} 
    \end{array} \right) \left( \begin{array}{c}
    \tilde{q}_j  \\ \tilde{p}_j
    \end{array} \right), \\
    \hat{b}_j &= \left( \begin{array}{cc}
    \sqrt{\frac{\omega_j}{2\hbar}} & \frac{i}{\sqrt{2\omega_j \hbar}}  
    \end{array} \right) \left( \begin{array}{c}
    \tilde{q}_j  \\ \tilde{p}_j
    \end{array} \right), \\
    \tilde{q}_j &= \sqrt{\frac{\hbar}{2\omega_j}} (\hat{b} + \hat{b}^\dagger), \\
    \tilde{p}_j &= \sqrt{\frac{\omega_j \hbar}{2}} \frac{(\hat{b} - \hat{b}^\dagger)}{i}, 
\end{align}
and obtain the Hamiltonian in normal mode
\begin{align}
    \hat{H} = \sum_{j=1}^{N_h} \hbar \omega_j \left( \hat{b}_j^\dagger \hat{b}_j + \frac{1}{2} \right).
\end{align}

%%%%%%%%%%%%%%%%%%%%%%%%%%%%%%%%%%%%%%%%%%%%%%%%%%%%%%%%%%%%%%%

\subsection{Boson-to-Qubit mapping}\label{sec: boson_qubit}

A typical boson-to-qubit mapping is the direct \textbf{one-to-one mapping} (also known as the unary mapping) studied by previous work such as Ref.~\citenum{Somma2003Quantum}, which maps a boson number state $| n_i \rangle$ ($0\le n_i \le N_b$) to a tensor representation employing $(N_b+1)$-qubit,
\begin{align}
| n_i \rangle &\leftrightarrow | 0_0 \cdots 0_{n_i-1} 1_{n_i} 0_{n_i+1} \cdots  0_{N_b} \rangle \notag \\
&~~~~= | 0 \rangle_0 \otimes \cdots \otimes | 0 \rangle_{n_i-1} \otimes | 1 \rangle_{n_i} \otimes | 0 \rangle_{n_i+1} \otimes \cdots \otimes | 0 \rangle_{N_b} , \label{occup_rep}
\end{align}
where
\begin{align}
    |0\rangle_j = \left( \begin{array}{c} 1 \\ 0 \end{array}\right),~~
    |1\rangle_j = \left( \begin{array}{c} 0 \\ 1 \end{array}\right),~~
\end{align}
are computation basis states of qubit $j$. From (\ref{boson_op2}), (\ref{boson_op3}), and (\ref{occup_rep}), we can then represent $\mathbf{I}_i$, $\tilde{b}^\dagger_i$, $\tilde{b}_i$, and $\tilde{n}_i$ using a similar tensor structure,
\begin{align}
    \mathbf{I}_i &= I_0 \otimes I_1 \otimes \cdots \otimes I_{N_b}, \notag \\
    \tilde{b}^\dagger_i &= \sum_{n=0}^{N_b-1} \sqrt{n+1} ~I_0 \otimes \cdots \otimes (+)_n \otimes (-)_{n+1} \otimes \cdots \otimes I_{N_b}, \notag \\
    \tilde{b}_i &= \sum_{n=1}^{N_b} \sqrt{n} ~I_0 \otimes \cdots \otimes (-)_{n-1} \otimes (+)_{n} \otimes \cdots \otimes I_{N_b},  \notag \\
    \tilde{n}_i &= \tilde{b}^\dagger_i \tilde{b}_i . %\sum_{n=0}^{N_b} n ~I_0 \otimes \cdots \otimes \frac{1}{2}( I_{n} + \sigma^{z}_{n} ) \otimes \cdots \otimes I_{N_b} .
\end{align}
where $(\pm)_j$ are ladder operator (on qubit $j$) defined from the Pauli matrices
\begin{align}
    (+) = \frac{1}{2} (X + i Y) = \left( \begin{array}{cc} 0 & 1 \\ 0 & 0 \end{array} \right), \notag \\
    (-) = \frac{1}{2} (X - i Y) = \left( \begin{array}{cc} 0 & 0 \\ 1 & 0 \end{array} \right),\label{ladder_op}
\end{align}
with the Pauli matrices (on qubit $j$)
\begin{align}
    X = \left( \begin{array}{cc} 0 & 1 \\ 1 & 0 \end{array} \right), ~~
    Y = \left( \begin{array}{cc} 0 & -i \\ i & 0 \end{array} \right), ~~
    Z = \left( \begin{array}{cc} 1 & 0 \\ 0 & -1 \end{array} \right). \label{Pauli_mat}
\end{align}
The Pauli ladder operators satisfy the following relations 
\begin{align}
    &(+) |1\rangle = |0\rangle,~~ (-) |1\rangle = 0,\notag \\ 
    &(+) |0\rangle = 0,~~ (-) |0\rangle = |1\rangle.
\end{align}

\noindent It's worth mentioning that the above approach uses $N_b$ qubits to represent $N_b$ bosonic particle states in one site. Nevertheless, the computational basis of $N_b$ $(>1)$ qubits explore the subspace of up to $2^{N_b}$ dimension. In other words, to encode $N_b$ boson paricle states and one vacuum state at each site, one would in principle only need $\lceil \log_2 (N_b+1) \rceil$ qubits. To achieve this, we can choose a \textbf{binary mapping} (studied in Ref. \citenum{Veis2016Quantum}, with its bit-swapping-efficient version, e.g. Gray code, discussed in Ref. \citenum{Sawaya2020Resource}) to represent bosonic states (at a given site $i$)
\begin{align}
   | n_i \rangle &\leftrightarrow | \underbrace{011\cdots101}_{\text{binary rep. of $n_i$}} \rangle \notag \\
&~~~~= | 0 \rangle_1 \otimes | 1 \rangle_2 \otimes | 1 \rangle_3 \otimes \cdots  \notag \\
&~~~~~~~~ \otimes| 1 \rangle_{N_q-2} \otimes | 0 \rangle_{N_q-1} \otimes | 1 \rangle_{N_q} . \label{bin_rep} 
\end{align}
Using the above convention we can write
\begin{align}
| 0\rangle &\leftrightarrow |0_1\rangle \otimes \cdots |0_{N_q-2}\rangle \otimes |0_{N_q-1}\rangle \otimes |0_{N_q}\rangle \notag \\ 
| 1\rangle &\leftrightarrow |0_1\rangle \otimes \cdots |0_{N_q-2}\rangle \otimes |0_{N_q-1}\rangle \otimes |1_{N_q}\rangle \notag \\
| 2\rangle &\leftrightarrow |0_1\rangle \otimes \cdots |0_{N_q-2}\rangle \otimes |1_{N_q-1}\rangle \otimes |0_{N_q}\rangle \notag \\
| 3\rangle &\leftrightarrow |0_1\rangle \otimes \cdots |0_{N_q-2}\rangle \otimes |1_{N_q-1}\rangle \otimes |1_{N_q}\rangle \notag \\
& \vdots \notag \\
| 2^{N_q}-1\rangle &\leftrightarrow |1_1\rangle \otimes \cdots |1_{N_q-2}\rangle \otimes |1_{N_q-1}\rangle \otimes |1_{N_q}\rangle. \label{binary_rep}
\end{align}
Accordingly, $\mathbf{I}_i$, $\tilde{b}^\dagger_i$, $\tilde{b}_i$, and $\tilde{n}_i$ can be represented in a $2^{N_q}\times 2^{N_q}$ dimensional matrix form that can also be rewritten in a tensor product form, i.e.,
\begin{align}
    \mathbf{I}_i &= \left( \begin{array}{ccccc}
      1  & 0  & 0  & \cdots & 0  \\
      0  & 1  & 0  & \cdots & 0  \\
      0  & 0  & 1  & \cdots & 0  \\
      \vdots & \vdots & \vdots & \ddots & \vdots \\
      0  & 0  & 0  & \cdots & 1 \\
    \end{array} \right) \notag \\
    &= I_1 \otimes I_2 \otimes \cdots \otimes I_{N_q}, \\
    \tilde{b}^\dagger_i &= \left( \begin{array}{cccccc}
      0  & 0  & 0  & \cdots & 0  & 0 \\
      1  & 0  & 0  & \cdots & 0  & 0 \\
      0  & \sqrt{2}  & 0  & \cdots & 0  & 0 \\
      \vdots & \vdots & \vdots & \ddots & \vdots & \vdots \\
      0  & 0  & 0  & \cdots & \sqrt{2^{N_q}-1}  & 0 \\
    \end{array} \right) \notag \\
    &= \sum_{n=0}^{2^{N_q}-2} \sqrt{n+1}~ | n+1\rangle \langle n |, \\
    \tilde{b}_i &= \left( \begin{array}{ccccc}
      0  & 1  & 0  & \cdots & 0 \\
      0  & 0  & \sqrt{2} & \cdots & 0 \\
      0  & 0  & 0 & \cdots & 0 \\
      \vdots & \vdots & \vdots & \ddots & \vdots \\
      0  & 0  & 0  & \cdots & \sqrt{2^{N_q}-1}\\
      0  & 0  & 0  & \cdots & 0 \\
    \end{array} \right) \notag \\
    &= \sum_{n=1}^{2^{N_q}-1} \sqrt{n}~ | n-1\rangle \langle n |,\\
    \tilde{n}_i &= \left( \begin{array}{ccccc}
      0  & 0  & 0  & \cdots & 0  \\
      0  & 1  & 0  & \cdots & 0  \\
      0  & 0  & 2  & \cdots & 0  \\
      \vdots & \vdots & \vdots & \ddots & \vdots \\
      0  & 0  & 0  & \cdots & 2^{N_q}-1 \\
    \end{array} \right) \notag \\
    &= \sum_{n=0}^{2^{N_q}-1} n~ | n\rangle \langle n |, \label{binary_rep2}
\end{align}
where the outer products $|\cdot \rangle \langle \cdot |$ can be expressed as a tensor product form similar to (\ref{binary_rep}) with the operation on each qubit being one of the following four,
\begin{align}
    |0\rangle \langle 0| &= \frac{1}{2} (I + Z ),~
    |0\rangle \langle 1| = (+) ,\notag \\
    |1\rangle \langle 0| &= (-) ,~
    |1\rangle \langle 1| = \frac{1}{2} (I - Z ). \label{eq:op_one_qubit}    
\end{align}
\\

It's worth mentioning that new encoding schemes have been reported recently. For example, Li et al. reported a variational basis state encoder that further reduces the number of qubits and number of gates for common bosonic operators to $\mathcal{O}(1)$~\cite{PhysRevResearch.5.023046}.
The variational basis sets essentially requires the preparation of a linear combination of many bosonic states. In Section \ref{sec: initial_state} we will discuss how to efficiently prepare the superposition of many bosonic states.

\textbf{Representing Hamiltonians} Now using the boson-to-qubit mappings introduced above, we can map some model Hamiltonians that are commonly used in quantum computation to qubit systems. 

\subsubsection{Boson-Preserving Hamiltonian}

Consider the following general Bose-Hubbard model
\begin{align}
    \hat{H} &= \sum_{i=1}^{N} \bigg( -\mu\hat{n}_i + \frac{U}{2}\hat{n}_i(\hat{n}_i-1) - t \sum_{j>i}^{N} (b^\dagger_i b_j+b_i b^\dagger_j) \notag \\
    &~~~~~~~~~~~~~~ + V \sum_{j>i}^{N} \hat{n}_i\hat{n}_j \bigg), \label{boson_hamiltonian}
\end{align}
where scalars $t$, $U$, $V$ and $\mu$ are the hopping amplitude, on-site interaction, non-local interaction, and chemical potential. $b^\dagger_i$, $b_i$, and $\hat{n}_i$ have been defined above. In the following, we will discuss the representations of the Hamiltonian of a two-site system (with at most one boson per site, i.e. $N=2$ and $N_b=1$) employing the two boson-to-qubit mapping, respectively. 

\begin{center}
    \begin{tikzpicture}
        \filldraw[fill=blue!40!white, draw=none] (0,0) circle (0.5cm);
        \node[draw=none] at (0,-0.8) {\small $|n_1\rangle$};
        \filldraw[fill=blue!40!white, draw=none] (2,0) circle (0.5cm);
        \node[draw=none] at (2,-0.8) {\small $|n_2\rangle$};
        \draw[thick,-] (0.5,0)--(1.5,0);
        \draw[thick,-] (-0.2,0.1)--(0.2,0.1);
        \node[draw=none] at (0,0.25) {\tiny $|1 \rangle$};
        \draw[thick,-] (-0.2,-0.2)--(0.2,-0.2);
        \node[draw=none] at (0,-0.05) {\tiny $|0 \rangle$};
        \draw[thick,-] (1.8,0.1)--(2.2,0.1);
        \node[draw=none] at (2,0.25) {\tiny $|1 \rangle$};
        \draw[thick,-] (1.8,-0.2)--(2.2,-0.2);
        \node[draw=none] at (2,-0.05) {\tiny $|0 \rangle$};
    \end{tikzpicture}
\end{center}

\noindent The \textit{one-to-one mapping} requires $N\times(N_b+1)=2\times(1+1) = 4$ qubits in total. Specifically, we have
\begin{align}
    b^\dagger_1 &= \tilde{b}^\dagger_1\otimes \mathbf{I}_2,~~
    b_1 = \tilde{b}_1\otimes \mathbf{I}_2,\notag \\    
    b^\dagger_2 &= \mathbf{I}_1\otimes \tilde{b}^\dagger_2,~~
    b_2 = \mathbf{I}_1\otimes \tilde{b}_2,  \\
    \hat{n}_1 &= \tilde{n}_1 \otimes \mathbf{I}_2,~~
    \hat{n}_2 = \mathbf{I}_1 \otimes \tilde{n}_2 , \notag
\end{align}
where
\begin{align}
    \tilde{b}^\dagger_{1/2} &= (+) \otimes (-) = | 0 \rangle \langle 1 | \otimes | 1 \rangle \langle 0 |,\notag \\
    \tilde{b}_{1/2} &= (-) \otimes (+) = | 1 \rangle \langle 0 | \otimes | 0 \rangle \langle 1 |,  
\end{align}
and
\begin{align}
    \tilde{n}_{1/2} &= \tilde{b}^\dagger_{1/2} \tilde{b}_{1/2}
    = \frac{1}{4}(I+Z)\otimes(I-Z) \notag \\
    &= | 0 \rangle \langle 0 | \otimes | 1 \rangle \langle 1 |. 
\end{align}
Then the Hamiltonian can be re-written as
\begin{align}
    \hat{H} &= \alpha_{1}~ \tilde{n}_1 \otimes \textbf{I}_2 + \alpha_{2}~\textbf{I}_1 \otimes \tilde{n}_2 
 	+ \beta_{12}~ (\tilde{b}^\dagger_1 \otimes \tilde{b}_2 + \tilde{b}_1 \otimes \tilde{b}^\dagger_2 ) \notag \\
    &~~~~ + \gamma_{12}~\tilde{n}_1 \otimes \tilde{n}_2  \label{Ham_components}
\end{align}
where
\begin{align}
    \tilde{n}_1 \otimes \textbf{I}_2 
    &= \frac{1}{4} \Big( IIII + ZIII - IZII - ZZII \Big)\notag \\
    \textbf{I}_1 \otimes \tilde{n}_2 
    &= \frac{1}{4} \Big( IIII + IIZI - IIIZ - IIZZ \Big)\notag \\
    \tilde{n}_1 \otimes \tilde{n}_2 
    &= \frac{1}{16} \Big( IIII + ZIII - IZII - ZZII  \notag \\
    &~~~~~~  +IIZI + ZIZI - IZZI - ZZZI \notag \\
    &~~~~~~  -IIIZ - ZIIZ + IZIZ + ZZIZ \notag \\
    &~~~~~~  -IIZZ - ZIZZ + IZZZ + ZZZZ \Big) 
\end{align}
and
\begin{align}
      &\tilde{b}^\dagger_1 \otimes \tilde{b}_2 + \tilde{b}_1 \otimes \tilde{b}^\dagger_2 \notag \\
    &~~= \frac{1}{8} \Big(XXXX + XXYY + YYXX + YYYY \notag \\
    &~~~~~~~~~~+ XYXY - XYYX - YXXY + YXYX \Big)
\end{align}
with the symbols denoting the tensor products of Pauli matrices, e.g., $XXXX$ represents $X_0\otimes X_1\otimes X_2\otimes X_3$.\\

\noindent The \textit{binary mapping} instead requires 2 qubits in total. Specifically, according to Eqs. (\ref{binary_rep2}), we can re-define $\tilde{b}^\dagger$, $\tilde{b}$, and $\tilde{n}$ as
\begin{align}
    \tilde{b}^\dagger_{1/2} &= (-) = | 1 \rangle \langle 0 |,~~
    \tilde{b}_{1/2} = (+) = | 0 \rangle \langle 1 |, \notag \\
    \tilde{n}_{1/2} &= \frac{1}{2}(I-Z) = | 1 \rangle \langle 1 |. \label{boson_op_1Q}
\end{align}
Then the Hamiltonian components of (\ref{Ham_components}) can be re-written as
\begin{align}
    \tilde{n}_1 \otimes I_2 
    &= \frac{1}{2} \Big( II - ZI\Big),~~ 
    I_1 \otimes \tilde{n}_2 
    = \frac{1}{2} \Big( II- IZ \Big)\notag \\
    \tilde{n}_1 \otimes \tilde{n}_2 
    &= \frac{1}{4} \Big( II - IZ - ZI + ZZ \Big)\notag \\
    \tilde{b}^\dagger_1 \otimes \tilde{b}_2 + \tilde{b}_1 \otimes \tilde{b}^\dagger_2
    &= \frac{1}{2} \Big(XX + YY \Big)
\end{align}
Using the Hamiltonian encoding, we can start to study the boson dynamics under the Hamiltonian, e.g. the quantum walks of indistinguishable bosons on 1D optical lattice that depends on the (absolute) strength of the on-site interaction $U$ (see Figure \ref{fig:Boson_quantum_walk}).

\begin{figure}
    \centering
    \includegraphics[width=\linewidth]{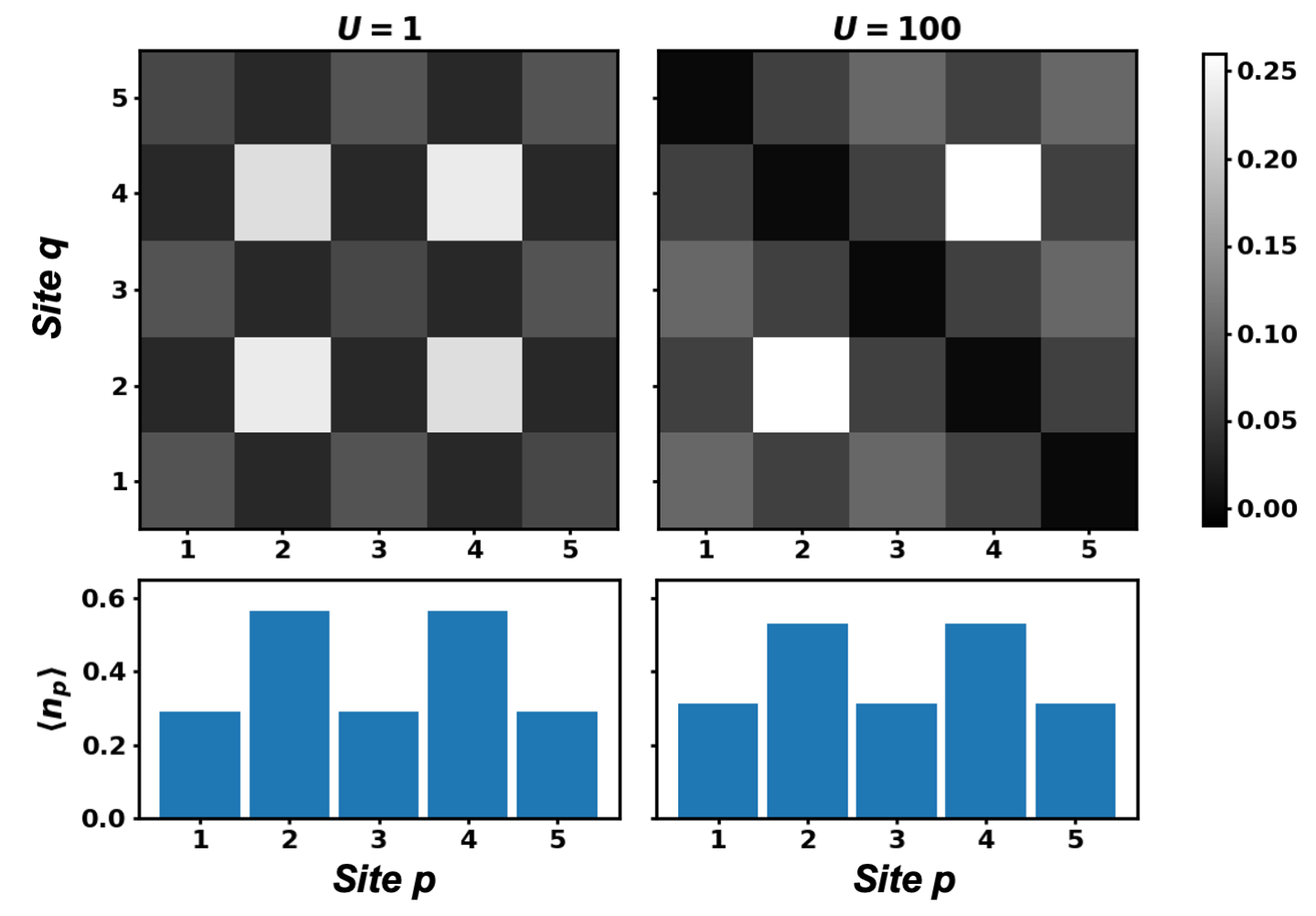}
    \caption{Two-particle correlation of the quantum walkers of two indistinguishable bosons, $\Gamma_{p,q}(t) = \langle \phi_b(t) | b_p^\dagger b_q^\dagger b_q b_p|\phi_b(t)\rangle$, on 1D optical lattice with five sites. The corresponding density distribution, $\langle n_p \rangle = \langle \phi_b(t) | b_p^\dagger b_p|\phi_b(t)\rangle$, is shown in the bottom of each plot of correlation. The Hamiltonian is given in (\ref{boson_hamiltonian}) with $N=5$, $\mu=-0.5U$, $t=1$, and $V=0$. $|\phi_b(t)\rangle$ evolves from the initial state $|\phi_b(0)\rangle = b_2^\dagger b_4^\dagger |\rm{vac}\rangle$ through the propagator $\exp(-i \hat{H} t)$ over the time $t$ = 0.003 a.u. ($\Delta t=1.0\times 10^{-5}$ a.u.), against the positions of two-boson $p$ and $q$ for different on-site interactions $U=1$ and $U=100$. }
    \label{fig:Boson_quantum_walk}
\end{figure}

%%%%%%%%%%%%%%%%%%%%%%%%%%%%%%%%%%%%%%%%%%%%%%%%%%%%%%%%%%%%%%

\subsubsection{Spin-Boson Hamiltonian}\label{sec: SB_interaction}

The simplest spin-boson model is a two-level system (TLS) coupled to a bath of harmonic oscillators referred to as a boson field. The full Hamiltonian is
\begin{align}
\label{eq:spin_boson_def}
    \hat{H} = \hat{H}_0 + \hat{H}_{SB}
\end{align}
where $\hat{H}_0$ is the zero-order Hamiltonian describing the separated subsystems
\begin{align}
    \hat{H}_0 = \Delta X + \frac{1}{2} \epsilon Z + \sum_i \omega_i \hat{n}_i
\end{align}
with $\Delta$ the bare tunneling amplitude between the two levels, $\epsilon$ the bias, and $\omega_i$ the frequencies of the oscillators. 
\begin{align}
    \hat{H}_{SB} = \frac{1}{2} X \sum_{i} g_i \omega_i (b_i + b^\dagger_i)
\end{align} 
is the coupling between the TLS and the bath of harmonic oscillators with $\lambda_i$ the coupling strength. The effect of the oscillator bath is completely determined by the spectral function $J(\omega)$, which in the ohmic case is given by
\begin{align}
    J(\omega) = \frac{\pi}{2} \sum_i c_i^2 \delta(\omega - \omega_i)
\end{align}
with $c_i \sim g_i \omega_i$. For the simplicity of the discussion, we consider simulating the spin dynamics of a spin-boson model with one spin and three bosons, and set $\omega=2$, $\epsilon=2$ and $\Delta=1$. The Hamiltonian is then simplified as%
\begin{align}
    \hat{H} = X + Z + 2 \hat{n} + gX(b^\dagger + b). \label{eq:Ham_SB}
\end{align}
If the binary mapping for the bosons is used on two qubits, then
\begin{align}
    |0\rangle &= | 00 \rangle, |1\rangle = |01\rangle, |2\rangle = |10\rangle, |3\rangle =|11\rangle,  
\end{align}
and
\begin{align}
    b^\dagger &= |1\rangle \langle 0| + \sqrt{2}|2\rangle \langle 1| + \sqrt{3}|3\rangle \langle 2|, \notag \\
    &= |01\rangle \langle 00| + \sqrt{2}|10\rangle \langle 01| + \sqrt{3}|11\rangle \langle 10|, \notag \\
    &= |0\rangle \langle 0|\otimes |1\rangle \langle 0| 
    + \sqrt{2}|1\rangle \langle 0|\otimes |0\rangle \langle 1| \notag \\
    &~~~~+ \sqrt{3}|1\rangle \langle 1|\otimes |1\rangle \langle 0|
    \notag \\
    &= \frac{1}{4} (I+Z)\otimes(X-iY) 
    + \frac{\sqrt{2}}{4}(X-iY)\otimes(X+iY) \notag \\
    &~~~~ + \frac{\sqrt{3}}{4}(I-Z)\otimes(X-iY).
\end{align}
where (\ref{eq:op_one_qubit}) is employed. Similarly,
\begin{align}
    b &= |0\rangle \langle 1| + \sqrt{2}|1\rangle \langle 2| + \sqrt{3}|2\rangle \langle 3|, \notag \\
    &= \frac{1}{4} (I+Z)\otimes(X+iY) 
    + \frac{\sqrt{2}}{4}(X+iY)\otimes(X-iY) \notag \\
    &~~~~ + \frac{\sqrt{3}}{4}(I-Z)\otimes(X+iY),   \\
    \hat{n} &= |1\rangle \langle 1| + 2 |2\rangle \langle 2| + 3 |3\rangle \langle 3| \notag \\
    &= \frac{1}{4} (I+Z)\otimes(I-Z) 
    + \frac{2}{4}(I-Z)\otimes(I+Z) \notag \\
    &~~~~ + \frac{3}{4}(I-Z)\otimes(I-Z) \notag \\
    &= \frac{1}{2} \big( 3II - IZ - 2ZI \big).
\end{align}
Now the Hamiltonian (\ref{eq:Ham_SB}) can be written as the linear combination of Pauli strings, 
\begin{align}
    \hat{H} &= XII + ZII + 3III - IIZ - 2IZI + \frac{g}{2} \bigg( (1+\sqrt{3})XIX \notag \\
    &~~~~ + (1-\sqrt{3})XZX + \sqrt{2}XXX + \sqrt{2}XYY \bigg) \label{eq:Ham_SB_qubit}
\end{align}
which can be mapped to three qubits (the first one is for the spin and the other two are for the bosons). Employing this Hamiltonian representation, we can then simulate the boson and spin dynamics of this model system. As shown in Figure \ref{fig:SB_dynamics} and Appendix \ref{app_a}, by employing a Lindblad formula, the simulation can be done in different coupling regimes with and without external conditions. 

\begin{figure*}
    \centering
    \includegraphics[width=\linewidth]{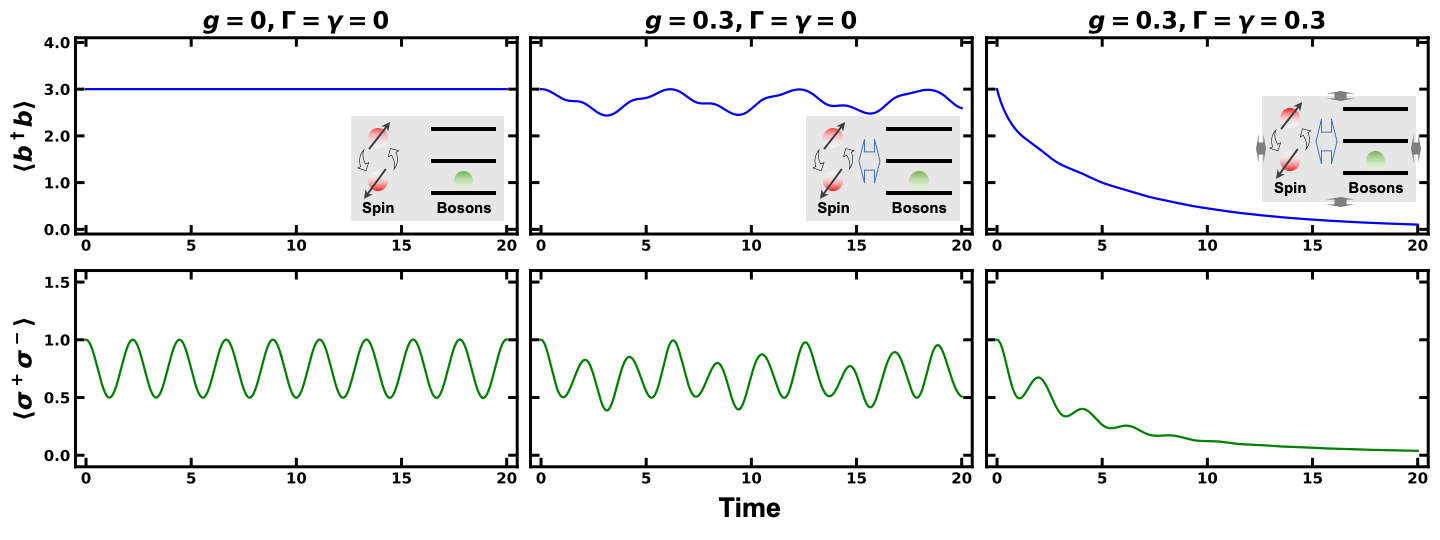}
    \caption{Boson (upper panels) and spin (lower panels) dynamics of a spin-boson model for different coupling regimes. The Hamiltonian is described in (\ref{eq:Ham_SB}) with the corresponding qubit representation in (\ref{eq:Ham_SB_qubit}). The decoherence is simulated through a Lindblad master equation as described in Appendix \ref{app_a} where the parameters $\Gamma$ and $\gamma$ account for the experimental imperfections.}
    \label{fig:SB_dynamics}
\end{figure*}

%%%%%%%%%%%%%%%%%%%%%%%%%%%%%%%%%%%%%%%%%%%%%%%%%%%%%%%%%%%%%

\subsubsection{Boson-Fermion Hamiltonian}\label{sec: BF_interaction}

%\noindent \textcolor{red}{BP: add discussion for generic hamiltonians.}

In this section, we focus on one of the simplest models that feature the boson-fermion interaction,  the Holstein model. Nevertheless, same procedure can also apply to more general Hubbard-Holstein model by including on-site interactions, or to molecules by replacing the quadratic fermionic parts with one-electron and two-electron molecular Hamiltonians. The one-dimensional version of the Holstein model reads
\begin{align}
    \hat{H} = - \sum_{\langle i,j \rangle} v f^\dagger_i f_j + \sum_i \omega b^\dagger_i b_i + \sum_i g \omega f^\dagger_i f_i (b^\dagger_i + b_i), \label{Holstein}
\end{align}
with $V$ the hopping coefficient between the nearest neighbour pair $\langle i,j \rangle$, $\omega$ the vibration frequency, and $g$ the coupling constant. For the simplicity of demonstration, we choose a three-site Holstein model with periodic boundary condition and $N_b = v = \omega = 1$ and only treat $g$ as the variable. For the binary encoding, we need $N\times \lceil \log_2 (N_b+1)\rceil = 6$ qubits. The Hamiltonian can then be re-written as
\begin{align}
    \hat{H} =& - f^\dagger_1 f_2 - f^\dagger_2 f_1 - f^\dagger_2 f_3 - f^\dagger_3 f_2 - f^\dagger_3 f_1 - f^\dagger_1 f_3 \notag \\
    & + (b^\dagger_1 b_1 + b^\dagger_2 b_2 + b^\dagger_3 b_3 ) \label{eq:3site_ham} \\
    & + g \left[f^\dagger_1 f_1 (b^\dagger_1 + b_1) + f^\dagger_2 f_2 (b^\dagger_2 + b_2) + f^\dagger_3 f_3 (b^\dagger_3 + b_3) \right] \notag
\end{align}
Assume bosons and fermions are elementary particles, then the bosonic and fermionic operators commute
\begin{align}
    [b,f] = [b,f^\dagger] = [b^\dagger, f] = [b^\dagger, f^\dagger] = 0.
\end{align}
Therefore, we can use the first three qubits for fermion encoding and the remaining three qubits for boson encoding. The fermion operators can be mapped to qubit systems through the Jordan-Wigner transformation, i.e.
\begin{align}
    f^\dagger_1 &=  (-) \otimes I \otimes I \otimes I \otimes I\otimes I , \notag \\
    f_1 &= (+) \otimes I \otimes I \otimes I \otimes I\otimes I , \notag \\
    f^\dagger_2 &= Z \otimes (-) \otimes I \otimes I \otimes I\otimes I , \notag \\
    f_2 &= Z \otimes (+) \otimes I \otimes I \otimes I\otimes I , \notag \\
    f^\dagger_3 &= Z \otimes Z \otimes (-) \otimes I \otimes I\otimes I , \notag \\
    f_3 &= Z \otimes Z \otimes (+) \otimes I \otimes I\otimes I , \notag \\
    f^\dagger_1 f_1 &= \frac{I-Z}{2} \otimes I \otimes I \otimes I \otimes I\otimes I , \notag \\
    f^\dagger_1 f_2 &= (-) \otimes (+) \otimes I \otimes I \otimes I\otimes I , \notag \\
    f^\dagger_1 f_3 &= (-) \otimes Z \otimes (+) \otimes I \otimes I\otimes I , \notag \\
    f^\dagger_2 f_2 &= I \otimes \frac{I-Z}{2} \otimes I \otimes I \otimes I\otimes I , \notag \\
    f^\dagger_2 f_1 &= (+) \otimes (-) \otimes I \otimes I \otimes I\otimes I , \notag \\
    f^\dagger_2 f_3 &= I \otimes (-) \otimes (+) \otimes I \otimes I\otimes I , \notag \\
    f^\dagger_3 f_3 &= I \otimes I \otimes \frac{I-Z}{2} \otimes I \otimes I\otimes I , \notag \\
    f^\dagger_3 f_2 &= I \otimes (+) \otimes (-) \otimes I \otimes I\otimes I , \notag \\
    f^\dagger_3 f_1 &= (+) \otimes Z \otimes (-) \otimes I \otimes I\otimes I .\label{Fermion_JW}
\end{align}
Here note the following simple Pauli relations 
\begin{align}
    Z(+) &= (+), ~~ (+) Z = - (+),\notag \\ 
    (-) Z &= (-), ~~ Z(-) = -(-), 
\end{align}
and
\begin{align}
    XZ &= -iY, ~~ ZX = iY, ~~ YZ = iX, ~~ ZY = -iX. 
\end{align}

For the boson operators, using (\ref{boson_op_1Q}), we have
\begin{align}
    b^\dagger_1 &=  I \otimes I \otimes I \otimes (-) \otimes I\otimes I ,\notag \\ 
    b_1 &= I \otimes I \otimes I \otimes (+) \otimes I\otimes I , \notag \\
    b^\dagger_2 &= I \otimes I \otimes I \otimes I \otimes (-) \otimes I ,\notag \\
    b_2 &= I \otimes I \otimes I \otimes I \otimes (+) \otimes I , \notag \\
    b^\dagger_3 &= I \otimes I \otimes I \otimes I\otimes I \otimes (-) , \notag \\
    b_3 &= I \otimes I \otimes I \otimes I\otimes I \otimes (+) , \notag \\
    b^\dagger_1 b_1 &= I \otimes I \otimes I \otimes \frac{I-Z}{2} \otimes I\otimes I , \notag \\
    b^\dagger_2 b_2 &= I \otimes I \otimes I \otimes I \otimes \frac{I-Z}{2} \otimes I , \notag \\
    b^\dagger_3 b_3 &= I \otimes I \otimes I \otimes I \otimes I\otimes \frac{I-Z}{2} , \notag \\
    b^\dagger_1 + b_1 &= I \otimes I \otimes I \otimes X \otimes I\otimes I , \notag \\
    b^\dagger_2 + b_2 &= I \otimes I \otimes I \otimes I \otimes X \otimes I , \notag \\
    b^\dagger_3 + b_3 &= I \otimes I \otimes I \otimes I \otimes I\otimes X . \label{Boson_binary}
\end{align}

\subsubsection{More Complex Hamiltonians}

\textcolor{black}{More complex Hamiltonians can usually be viewed as extensions or generalizations of the simple models discussed in the preceding sections. For instance, introducing anharmonic terms to the harmonic oscillator moves beyond the simple picture of independent oscillations, capturing more realistic dynamics where energy levels are no longer equally spaced. Similarly, the Fr\"{o}hlich Hamiltonian and more general electron-phonon or light-matter polariton Hamiltonians can be seen as generalizations of the spin-boson model, incorporating multiple electronic and bosonic modes, momentum dependence, and more intricate interaction terms. This progression from simpler to more complex models reflects the increasing realism and richness of the physical phenomena being described but also introduces significant challenges for quantum simulation. In the following, we briefly discuss the specific challenges and strategies for encoding these more complex Hamiltonians, focusing on qubit mapping strategies, truncation and approximations, and resource requirements, while leaving the discussion of quantum algorithms to later sections.}

\textcolor{black}{Encoding \textbf{anharmonicity} presents a significant challenge due to the presence of higher-order bosonic terms. Beside truncating the bosonic Hilbert space, limiting the maximum number of bosons per mode ($n_{max}$), another approach is to employ more sophisticated mappings, such as the mapping based on the discrete variable representation~\cite{DVR92}, that can be more efficient for specific anharmonic potentials. However, these mappings often lead to more complex qubit interactions. The resource requirements scale with the number of modes and the chosen $n_{max}$. The complexity of the resulting Hamiltonian also increases with the order of the anharmonic terms, leading to deeper quantum circuits for simulation.}

\textcolor{black}{The \textbf{Fr\"{o}hlich Hamiltonian}, describing electron-phonon interactions, involves both fermionic and bosonic degrees of freedom. For the bosonic part, similar strategies as in the anharmonic case are used, involving truncation of the phonon Hilbert space. The fermionic part is typically mapped using the Jordan-Wigner or Bravyi-Kitaev transform~\cite{jordan1928paulische,bravyi2002fermionic}. The major challenge comes from the momentum-dependent interaction term, which leads to all-to-all interactions in real space. This can be mitigated by working in momentum space, but it still leads to complex multi-qubit interactions. The qubit requirements scale with the number of electronic and phononic modes. If there are $N$ electronic modes and $M$ phononic modes, and each phononic mode is truncated to $n_{max}$, the total qubit requirement is $N + M\lceil\log(n_{max} + 1)\rceil$. The interaction terms also introduce significant overhead in terms of gate counts during the simulation. Techniques like interaction picture and perturbative gadgets can be used to reduce the simulation cost~\cite{PerturbationGaget}.}

\textcolor{black}{The more general \textbf{electron-phonon and light-matter polariton Hamiltonians} inherit the complexities of the Fr\"{o}hlich Hamiltonian but often include additional challenges. Long-range interactions require more careful consideration of qubit mappings to minimize non-local interactions. Transformations like the Bravyi-Kitaev transform or more advanced techniques like mapping to tree tensor networks can be beneficial~\cite{Huggins_2019}. Multiple phonon branches or photonic modes increase the number of bosonic degrees of freedom, leading to a linear increase in qubit requirements. The light-matter interactions, especially in the dipole gauge, can lead to quadratic bosonic terms, which must be treated with techniques similar to those used for anharmonic oscillators. Truncation of the bosonic Hilbert space remains a common approximation, with the same trade-off between accuracy and qubit requirements. In light-matter systems, approximations based on rotating wave approximation are also often employed, which can significantly simplify the Hamiltonian by neglecting fast oscillating terms. The resource requirements for these general models scale with the number of electronic, phononic, and photonic modes, as well as the complexity of the interactions and the level of truncation applied.}

\subsection{Prepare Initial Bosonic States}\label{sec: initial_state}

For an $N$-site quantum system with $n_i$ ($0\le n_i \le N_b$) bosons at each site, its bosonic state can be written as
\begin{align}
    | \phi_b \rangle = \mathcal{N} (b^\dagger_1)^{n_1} (b^\dagger_2)^{n_2} \cdots (b^\dagger_N)^{n_N} | \text{vac} \rangle, \label{boson_state1}
\end{align}
where $\mathcal{N}$ is a normalization constant, and 
\begin{align}
    |\text{vac}\rangle = \underbrace{|0\rangle \otimes \cdots \otimes |0 \rangle \otimes |0\rangle}_{\text{$N$-tensor}}
\end{align}
with $|0\rangle$ at site $i$ being mapped to a $n_i$-tensor product according to the one-to-one mapping (\ref{occup_rep}) or the binary mapping (\ref{bin_rep}). Using the mapping (\ref{occup_rep}), only $N$ qubits need to be flipped to prepare $| \phi_b\rangle$. For the mapping (\ref{bin_rep}), at most $N\times \lceil \log_2 ( N_b + 1 ) \rceil $ qubits need to be flipped instead.

If the initial bosonic state assumes the superposition form 
\begin{align}
    |\psi \rangle = \sum_{k=1}^K c_k |\phi_{b,k}\rangle, \label{boson_state2}
\end{align}
where
\begin{align}
    \sum_{k=1}^K |c_k|^2 = 1,~~ \langle \phi_{b,k} | \phi_{b,l} \rangle = \delta_{k,l},~~ |\phi_{b,k}\rangle = U_k |\text{vac}\rangle, \label{cond}
\end{align}
with $U_k$ representing a unitary operation, similar to (\ref{boson_state1}),
then we can use $K$ ancillas to do the following state preparation procedure~\cite{Somma2003Quantum}

\begin{itemize}
\item[\#]1. Initialize the $K$ ancillas in state $|0\rangle$;
\item[\#]2. Generate $\sum_{k=1}^K c_k |k\rangle$. Here the state $| k \rangle$ is an ancilla state with only the $k$-th qubit being $|1\rangle$;
\item[\#]3. Loop over the $K$ ancillas, perform the controlled unitary operations, i.e., conditional on the $k$-th ancilla being $|1\rangle$, apply $U_k$ on the $|\text{vac}\rangle$. This will generate $\sum_{k=1}^K c_k |k\rangle \otimes |\phi_{b,k}\rangle$;
\item[\#]4. Generate $\frac{1}{\sqrt{K}}\sum_{k=1}^K c_k |0\rangle \otimes |\phi_{b,k}\rangle$.
\end{itemize}

The corresponding circuit structure is given in Figure \ref{fig: state_prep}. As can be seen, step \#2 and \#4 in the above procedure are essentially identical, and only differing by the single qubit rotations and the selection of the subspace to project out. The single qubit rotations can be figured out from $K$ and $c_k$'s. To see this, take step \#2 as an example, the state over $|b_1\rangle$ and the $K$ ancillas evolves as follows
\begin{widetext}
\begin{align}
&|0\rangle \otimes |\underbrace{0\cdots0}_{K}\rangle 
\left[ \xrightarrow{|0\rangle \langle 0 | \otimes R_k + |1\rangle \langle 1 | \otimes X}\right]_{k=1\cdots K} 
 x_1 \cdots x_K |0\rangle \otimes |\underbrace{0\cdots0}_{K}\rangle + \sum_{k=1}^K x_1 \cdots x_{k-1} y_k |1\rangle \otimes |\underbrace{0\cdots0}_{k-1}1_k \underbrace{0\cdots0}_{K-k}\rangle \notag \\
& \xrightarrow[\text{on $K$ ancillas}]{\text{if $|b_1\rangle = |1 \rangle$}}
~~\sum_{k=1}^K x_1 \cdots x_{k-1} y_k |\underbrace{0\cdots0}_{k-1}1_k \underbrace{0\cdots0}_{K-k}\rangle ,
\end{align}
\end{widetext}
where $x_1 \cdots x_{k-1} y_k = c_k$ for $k=1\cdots K$.
The probability of successfully preparing the state in step \#4 is $1/K$. We can boost the success probability to close to $1$ by classically repeating this preparation $\mathcal{O}(K)$ times, or doing $\mathcal{O}(\sqrt{K})$ steps of amplitude amplification.

\begin{figure}
    \centering
    \includegraphics[width=\linewidth]{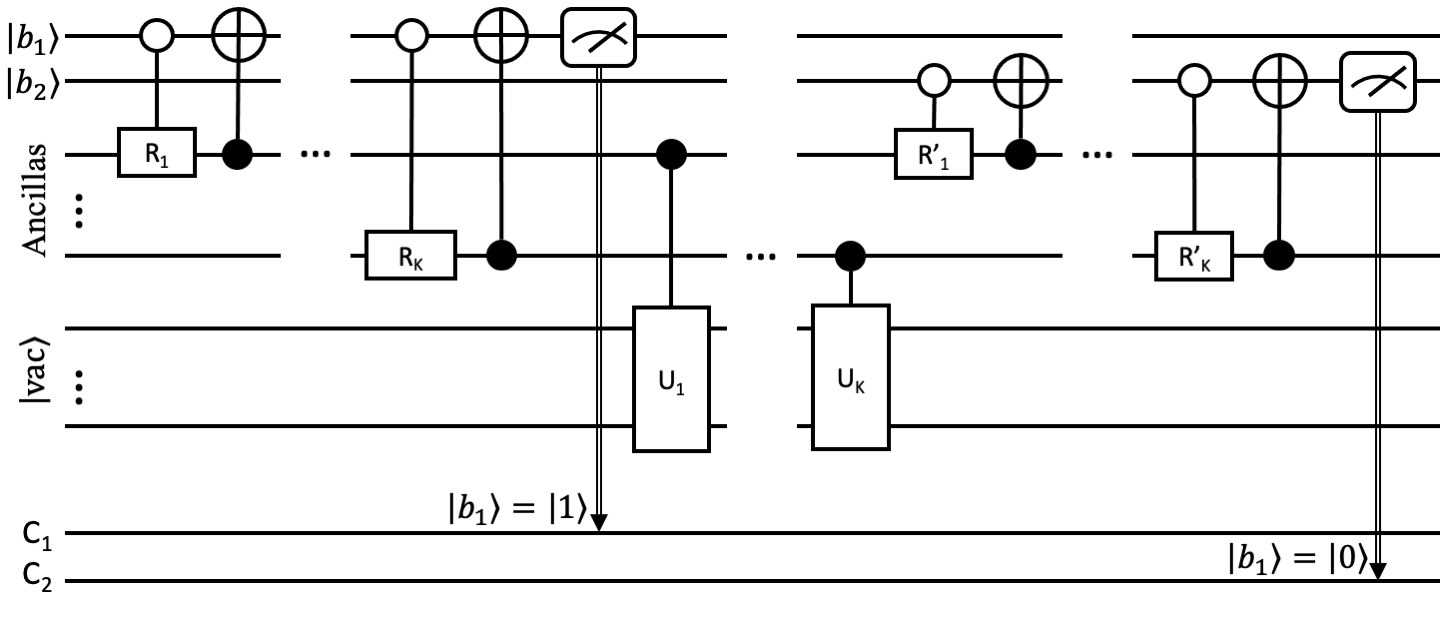}
    \caption{Circuit to prepare the state $\frac{1}{\sqrt{K}}\sum_{k=1}^K c_k |\phi_{b,k}\rangle$ with conditions (\ref{cond}). $R_k$ and $R'_k$ ($k=1,\cdots, K$) gates rotate the single qubit state $|0\rangle$ or $|1\rangle$ to the state $x|0\rangle + y|1\rangle$. The unitary gate $U_k$ ($k=1,\cdots, K$) acting on $|\text{vac}\rangle$ to generate $|\phi_{b,k}\rangle$. $C_{1/2}$ are two classical registers to take the projection of the ancillas $|b_{1/2}\rangle$ on the compuational basis $|0\rangle$ and $|1\rangle$. In the controlled gates, `$\circ$' denotes turning on the conditional operation when the controlled qubit is $|0\rangle$, while `$\bullet$' denotes turning on the operation when the controlled qubit is $|1\rangle$. Only when $|b_1\rangle = |1\rangle$ and $|b_2\rangle = |0\rangle$, the desired state will be prepared.}
    \label{fig: state_prep}
\end{figure}

We can construct a more efficient state preparation by generating the state
\begin{equation}
    \frac{1}{\sqrt{\sum_{\kappa=1}^{K}\abs{c_\kappa}}}\sum_{k=1}^K\sqrt{\abs{c_k}}\ket{k}
\end{equation}
in step \#2 and measure the same state in \#4, while introducing phases of $c_k$ in step \#3. The probability of success then becomes $1/(\sum_{k=1}^{K}\abs{c_k})^2$, and can be boosted to unity with $\mathcal{O}(\sum_{k=1}^{K}\abs{c_k})$ steps of amplitude amplification. Note that this approach has a lower complexity than the previous approach due to the fact that the inequality
\begin{equation}
    \sum_{k=1}^{K}\abs{c_k}\leq\sqrt{K}\sum_{k=1}^K |c_k|^2=\sqrt{K}
\end{equation}
always holds but is not always saturated.

It is possible to further lower the complexity by synthesizing a permutation unitary. Specifically, we first generate $\sum_{k=1}^K c_k |k\rangle$, and our goal is then to implement the transformation
\begin{equation}
    \ket{k}\mapsto\ket{\phi_{b,k}},
\end{equation}
where both $k$ and $\phi_{b,k}$ are encoded in binary. This is a permutation of $K$ basis states. It can be decomposed into cyclic permutations of total length $\mathcal{O}(K)$, and further into $\mathcal{O}(K)$ transpositions (transformations between two basis states with all remaining basis states fixed). A transposition can be represented as a two-level unitary, and can be synthesized using the Gray code as in~\cite[Section 4.5.2]{Nielsen2010}.

\subsection{Capture Ground States}\label{sec:ground}

Once the Hamiltonian and the initial states are encoded, we can start to compute the ground states. Take the boson-fermion Hamiltonian as an example, bounds to the lower portion of the eigenspectrum of $\hat{H}$ can be obtained through many approaches. On the NISQ devices, typical routines include hybrid quantum-classical Variational Quantum algorithm (VQA)~\cite{Peruzzo_2014,mcclean2016theory,romero2018strategies,shen2017quantum,Kandala2017hardware,kandala2018extending,colless2018computation, huggins2020non},
quantum approximate optimization algorithm (QAOA)~\cite{farhi2014quantum},
quantum annealing \cite{RevModPhys.94.015004,albash2018adiabatic},
gaussian boson sampling~\cite{Aaronson2011},
analog quantum simulation~\cite{trabesinger2012quantum,georgescu2014quantum},
iterative quantum assisted eigensolver~\cite{mcardle2019variational,motta2020determining,parrish2019quantum,kyriienko2020quantum}, imaginary time evolution (ITE)~\cite{PhysRevA.99.062304,mcardle2019variational,Yuan2019theoryofvariational,motta2020determining,PRXQuantum.2.010333,yeter2020practical,PRXQuantum.2.010317,PhysRevA.105.012412,liu2021probabilistic,Yeter-Aydeniz_2021,PhysRevResearch.4.033121},
and many others. Particularly, the ITE approach has a long history of being a robust computational approach to solve the ground state of a many-body quantum system. The development and application of ITE approach targeting the ground state wave function and energy dates back to 1970s, when the similar random-walk imaginary-time technique were developed for diffusion Monte-Carlo methods~\cite{davies1980application, anderson1975random, anderson1979quantum, anderson1980quantum}. 

Combining the ITE with the variational expansion, 
there had been another interesting yet less known moment approach proposed by Peeters and 
Devreese \cite{peeters1984upper}, and further analyzed by 
Soldatov \cite{soldatov1995generalized} in the 1990s, which we will refer to as the Peeters-Devreese-Soldatov (PDS) approach. We recently reviewed this method and proposed its potential use for accurate quantum computations~\cite{Kowalski2020, claudino2021improving, Peng2021variationalquantum,peng2022State}. 
Specifically, in this approach the energy functional depends on the moments of the Hamiltonian $\langle \phi | H^n | \phi \rangle =  \langle H^n \rangle$, for some trial state $|\phi \rangle$. For a PDS approximation of order $K$, that is, PDS($K$), one needs to estimate $\langle \hat{H} \rangle,\langle \hat{H}^2 \rangle, \cdots \langle \hat{H}^{2K-1} \rangle$. The moments serve as elements of the matrix $\mathbf{M}$ ($M_{ij} = \langle \hat{H}^{2K - i - j} \rangle$) and vector $\mathbf{Y}$ ($Y_i = \langle \hat{H}^{2K - i} \rangle$), which are related by $\mathbf{MX} =- \mathbf{Y}$. The solution of this linear system comprises the coefficients of the following polynomial
\begin{equation}
P_K(\mathcal{E}) = \mathcal{E}^{K} + \sum_{i=1}^K X_i \mathcal{E}^{K-i}=0,
\end{equation}
whose roots are the bounds to the $K$ lowest eigenvalues in the eigenspectrum of $\hat{H}$.

One requirement for PDS to work is that the trial state $|\phi \rangle$ finds support in the ground state of $\hat{H}$, that is, it has non-zero overlap with the ground state. This is a crucial requirement because variation in the magnitude of the coupling constant $g$ in (\ref{eq:3site_ham}) leads to changes in the relative position of the energy levels corresponding to states in different fermionic number sectors. This change is due to the number of fermions increasing as $g$ becomes larger until all fermionic sites are occupied. In the case of a three-site Holstein model, for $g \lesssim 1.27$, the total number of fermions is 1, and for $g \gtrsim 1.27$ it is 3. Consequently, an attempt to estimate the ground state energy with methods such as PDS in the entire domain of $g$ implies a trial state that is a superposition of states with all allowed particle numbers, which in this case can be as simple as
\begin{eqnarray}
    \label{eq:holstein_trial}
    |\phi\rangle = \frac{1}{\sqrt{2}}\underbrace{(|001\rangle + |111\rangle)}_{\text{fermions}} \otimes \underbrace{|000\rangle}_{\text{bosons}}.
\end{eqnarray}
The results of employing the trial state in (\ref{eq:holstein_trial}) to evaluate the bound to the ground state energy provided by various orders of the PDS energy functional as function of the parameter $g$ are reported in Figure \ref{fig:En_Holstein}.

%Then we can re-write the Hamiltonian in terms of the Pauli strings. Figure \ref{fig:En_Holstein} shows the change of the exact ground state energy of the three-site Holstein model as a function of fermion-boson coupling strength $g$.

\begin{figure}
    \centering
    \includegraphics[width=\linewidth]{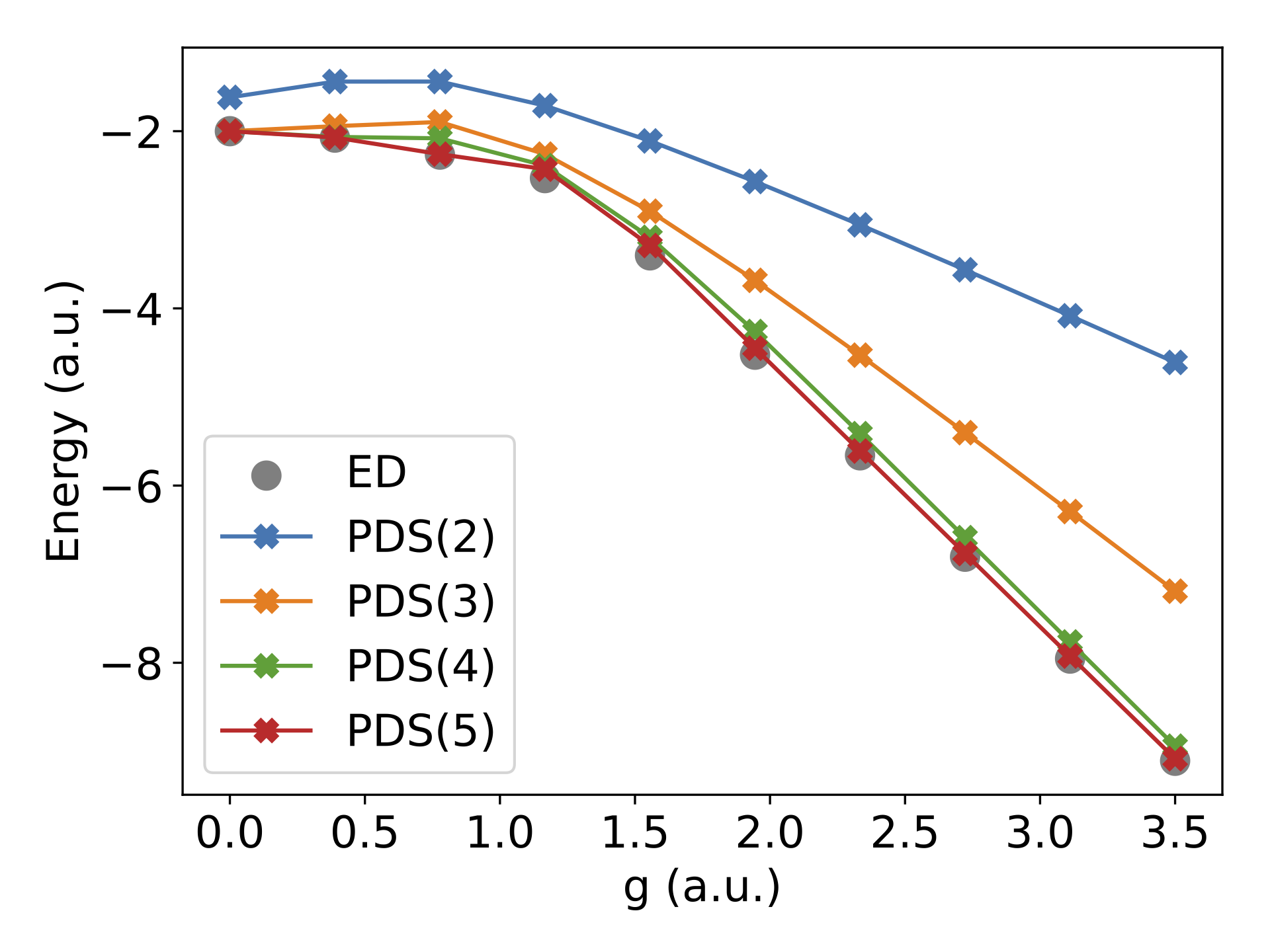}%{En_Holstein.png}
    
    \caption{The change in the ground state energy from exact diagonalization (ED) and in various orders of the PDS energy functional of a three-site Holstein model as a function of fermion-boson coupling strength.}
    %\caption{The change of the exact ground state energy of a three-site Holstein model as a function of fermion-boson coupling strength.}
    \label{fig:En_Holstein}
\end{figure}

The simple trial state put forth in (\ref{eq:holstein_trial}) is adequate to capture the different particle number sectors that are covered in the domain of the fermion-boson coupling $g$. This is evidenced by Figure \ref{fig:En_Holstein} in demonstrating that the bound to the ground state energy furnished by the PDS energy functional rapidly converges toward the value obtained by exact diagonalization (ED) with increasing order parameter $K$. While PDS(4) displays small, but still noticeable disagreement with exact values, PDS(5) energy estimates are visually indistinguishable from the results from ED. This lends credence to the suitability of the PDS approach to be invaluable tool in studying open quantum systems.

%%%%%%%%%%%%%%%%%%%%%%%%%%%%%%%%%%%%%%%%%%%%%%%%%%%%%%%%%%%%%
\subsection{Product Formulas}\label{sec:trotter}

Regarding the simulation of bosonic evolution, the idea is to approximately implement the unitary operator $\hat{U}(t) = \exp(-i\hat{H}t)$ where $H$ is the Hamiltonian of system with its representations exemplified in Section \ref{sec: boson_qubit}. For a given Hamiltonian that has a polynomial large number of terms (with respect to $N$), its unitary evolution $\hat{U}(t)$ can be efficiently performed on a quantum computer by applying a polynomial number of single and two-qubit gates.

There are many quantum algorithmic solutions designed for quantum dynamics (we refer the readers to Ref. \citenum{Miessen23_quantum} for a recent perspective). 
\textcolor{black}{One category of quantum algorithms for simulating the evolution $\hat{U}(t)$ is the product formula~\cite{Lloyd1996,BACS05,suzuki1991}, which includes the Lie-Trotter product formula and higher-order Suzuki formulas. The product formula is straightforward and efficient, while also preserving certain symmetries of the dynamics, making it a natural choice for quantum simulations.}
% , which increases the number of operations but does not destroy the efficiency of the simulation. 
A product formula approximates the target evolution by a product of exponentials of terms in the Hamiltonian.
Take the Trotter decomposition as an example, if $\hat{H}=\hat{K}+\hat{V}$, then
\begin{align}
    \hat{U}(t) = e^{-i\hat{H}t} \left\{ \begin{array}{lc}
    = e^{-i\hat{K}t}e^{-i\hat{V}t}    &  \text{if}~ [\hat{K},\hat{V}] = 0\\
    \approx  [e^{-i\hat{K}t/n}e^{-i\hat{V}t/n}]^n    & \text{if}~ [\hat{K},\hat{V}] \neq 0
    \end{array}. \right.
\end{align}
Specifically, the approximation error of each step is bounded by
\begin{equation}
    \norm{\hat{U}\left(\frac{t}{n}\right)-e^{-i\hat{K}t/n}e^{-i\hat{V}t/n}}\leq\frac{\norm{[\hat{K},\hat{V}]}t^2}{2n^2},
\end{equation}
which implies an overall error of at most
\begin{equation}
    \norm{\hat{U}\left(t\right)-[e^{-i\hat{K}t/n}e^{-i\hat{V}t/n}]^n}\leq\frac{\norm{[\hat{K},\hat{V}]}t^2}{2n}.
\end{equation}
To achieve an accuracy $\epsilon$, it thus suffices to take
\begin{equation}
    n=\left\lceil\frac{\norm{[\hat{K},\hat{V}]}t^2}{2\epsilon}\right\rceil.
\end{equation}
This analysis can be extended to Hamiltonians containing multiple terms, and to other Trotter decompositions with a higher-order accuracy~\cite{CSTWZ21}.

As discussed in Section \ref{sec: boson_qubit}, we can represent bosonic operators using either a unary or a binary encoding.
For the unary encoding, bosonic operators are represented by linear combinations of Pauli operators which can then be split using the Trotter decomposition.
To exponentiate a Pauli operator, take $\exp(-i\frac{\delta}{2} XXYY)$ as an example. Since 
\begin{align}
    HZH = X,~~ S X  S^\dagger = Y, 
\end{align}
with Hadamard gate $H = \frac{1}{\sqrt{2}}\left(\begin{array}{cc} 1 & 1 \\ 1 & -1 \end{array}\right)$ and phase gate $S = \left(\begin{array}{cc} 1 & 0 \\ 0 & i \end{array}\right)$, we then have
\begin{align}
    XXYY &= IISS\cdot XXXX \cdot IIS^\dagger S^\dagger \notag \\
    &= IISS\cdot HHHH \cdot ZZZZ \cdot HHHH \cdot IIS^\dagger S^\dagger
\end{align}
and
\begin{align}
    &\exp(-i\frac{\delta}{2} XXYY)  \\
    &= IISS\cdot HHHH \cdot \exp(-i\frac{\delta}{2} ZZZZ) \cdot HHHH \cdot IIS^\dagger S^\dagger. \notag
\end{align}
The entire circuit representing $\exp(-i\frac{\delta}{2} XXYY)$ is shown in Figure \ref{fig:ladder_op}, where the middle ``CNOT-staircase" circuit represents the operator $\exp(-i\frac{\delta}{2} ZZZZ)$.

\begin{figure}
    \centering
    \includegraphics[width=\linewidth]{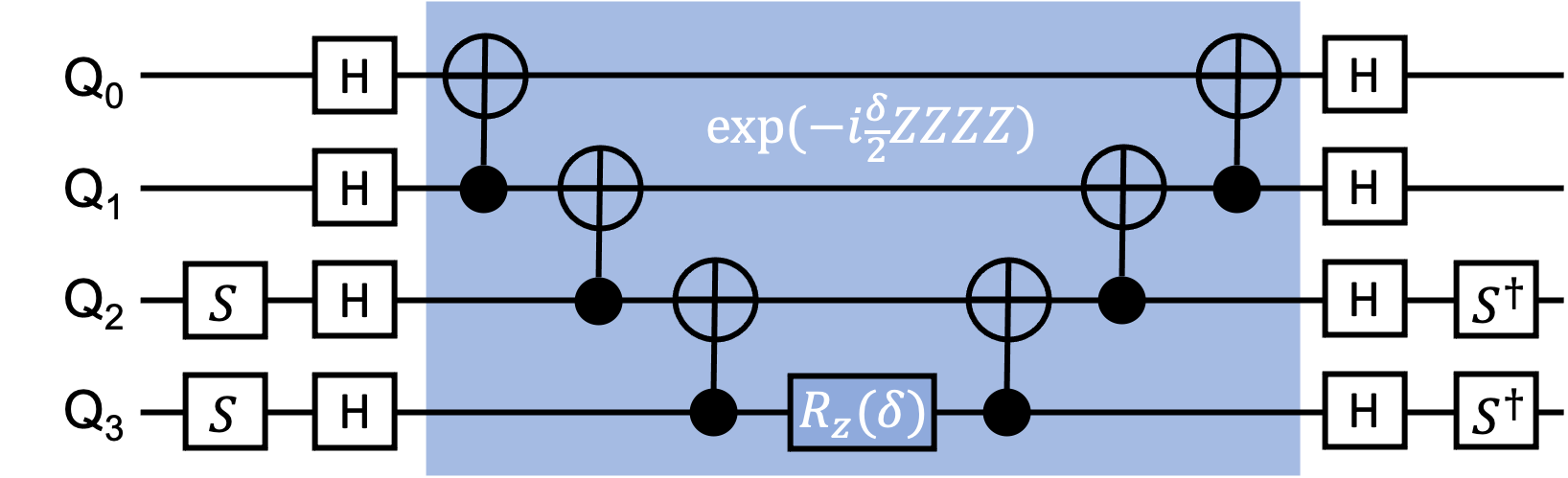}
    \caption{Circuit to generate the operator $\exp(-i\frac{\delta}{2} XXYY)$. The middle ``CNOT-staircase'' (shaded) corresponds to the operator $\exp(-i\frac{\delta}{2} ZZZZ)$.}
    \label{fig:ladder_op}
\end{figure}

For the binary encoding, we can implement bosonic operators by representing them as linear combinations of position and momentum operators, which are diagonalized in the position and momentum basis respectively. When two operators act on different modes, or when operators of the same type act on a single mode, they can be simultaneously diagonalized and the circuit implementation is straightforward. However, the implementation becomes more challenging for products of position and momentum operators on the same mode (such terms arise in squeezing Hamiltonians~\cite{gerry2005introductory}). 
Take $e^{-it (pq + qp)}$ as an example. Since neither $|p\rangle$ nor $|q\rangle$ are eigenvectors of $(pq + qp)$, the diagonalization of $(pq + qp)$ needs to be performed numerically on the classical side followed by encoding the eigenvectors in terms of the boson states~\cite{Macridin18}, i.e.
\begin{align}
    (pq + pq)|\mu\rangle = \epsilon_{\mu} | \mu \rangle,~~
    |\mu\rangle = \sum_{i=0}^{N_b-1} U_{\mu,i} |n_i\rangle.
\end{align}
where the matrix $\mathbf{U}$ with $U_{\mu,i}$ being the element is unitary. 
% Nevertheless, the number of gates employed to implement the unitary $\mathbf{U}$ has an exponential scaling (we refer the readers to the Section IV E in Ref. \citenum{Macridin18} for a detailed discussion).
The complexity of this approach has a polynomial scaling with the bosonic cutoff instead of a logarithmic scaling as in the diagonalizable case (we refer the readers to the Section IV E in Ref. \citenum{Macridin18} for a detailed discussion).
Alternatively, we can think about approximating $e^{-it (pq + qp)}$ using product formulas~\cite{Lloyd1997,Park2017}. Note that
\begin{align}
    pq + qp = \{p,q\}
\end{align}
is essentially an anticommutator, then employing the strategy introduced in Ref. \citenum{Childs2013PF_Commutator}, one can create analogs of $pq$ and $qp$ on an enlarged Hilbert space
\begin{align}
    p' := p \otimes Y,~~ q':= q \otimes X 
\end{align}
where $X$ and $Y$ are Pauli matrices. Then for any quantum state $|\phi\rangle$ prepared in the same Hilbert space as $p$ and $q$, we have
\begin{align}
    &[p',q']|\phi\rangle \otimes |0\rangle = -i \{p,q\}|\phi\rangle \otimes |0\rangle \notag \\
    \Rightarrow &~~ e^{-it\{p,q\}}|\phi\rangle \otimes |0\rangle
    = e^{t[p',q']} |\phi\rangle \otimes |0\rangle.
\end{align}
Then a product formula approximation can be constructed for $e^{t[p',q']}$ with a number of exponentials that scales almost linearly with the evolution time (a similar formula can be used to implement $\hat{H}^n$ for the PDS($K$) approach discussed in Section \ref{sec:ground}). Commutator-type error bounds can also be derived for such product formulas. For example, following Ref. \citenum{gluza2022doublebracket} (Lemma 6), we have
%
% \begin{align}
%     \| e^p e^q e^{p^\dagger} e^{q^\dagger} - e^{[p,q]/2} \|
%     \le \| p, [p,q] \| + \|q, [q,p] \|.
% \end{align}
\begin{align}
    &\| e^{-itp} e^{-itq} e^{itp} e^{itq} - e^{-t^2[p,q]/2} \| \notag\\
    &~~~~~~~~~~~~\le t^3\left(\| p, [p,q] \| + \|q, [q,p] \|\right).
\end{align}

%\textcolor{red}{Macridin's method has complexity $N^2$ where $N$ is the cutoff. But maybe we can do better? Suppose we want to implement $e^{-it(XP+PX)}$, where $X$ is the position operator and $P$ is the momentum operator. Then we can use the commutator/anticommutator product formula from [arXiv:1211.4945]? Check [arXiv:2206.11772, Lemma 6] for error analysis.}\textcolor{black}{sounds great! I have been working on citations on my local version, after I merge the changes this week, I will look into the paper you mentioned.}

% \textcolor{red}{For unary encoding, Hamiltonian terms are Pauli strings. For binary encoding, the implementation basically follows from the construction in ref. \citenum{Macridin18} with the need of performing quantum simulation in the position/momentum basis.}

\subsection{Connection to the Experimental Simulations of Open Quantum Systems}

\textcolor{black}{Quantum information processing is inherently susceptible to environmental noise, a crucial factor in simulating phenomena like non-equilibrium phase transitions~\cite{Torre10}, driven-dissipative phase transitions~\cite{Schindler13, PhysRevLett.123.173601}, and non-Hermitian topological phenomena~\cite{PhysRevX.9.041015} observed in quantum many-body systems. Real-world quantum systems deviate from idealized models, exhibiting non-unitary dynamics due to environmental interactions~\cite{Koch_2016, Nielsen2010}. To simulate the continuous evolution of open quantum systems, such as those found in circuit quantum electrodynamics (cQED) architectures, the Trotterization method has been successfully applied~\cite{PhysRevLett.127.020504}. In cQED, the first two Fock states of a microwave cavity, $|0\rangle$ and $|1\rangle$, represent the data qubit, while a dispersively coupled transmon qubit acts as an ancilla. By tuning the intensity of environmental noise, we can precisely control the open system dynamics, providing valuable insights into the behavior of quantum systems in the absence of noise.}

\textcolor{black}{The generalization of Trotterization to open systems follows the same principles as for closed systems. Higher-order Trotter schemes, such as the second-order Trotter, offer improved precision. For general Liouvillians, $\mathcal{L}_j (j \in {1, ..., m})$, encompassing both coherent and incoherent components, the second-order Trotter achieves an error of $\mathcal{O}(\Delta t_3)$:
\begin{align}
\exp(\sum_{j=1}^m\mathcal{L}_j\Delta t) &= \prod_{j=1}^m \exp(\mathcal{L}_j\Delta t/2)\prod_{j=m}^1 \exp(\mathcal{L}_j\Delta t/2) \notag \\
&~~+ \mathcal{O}(\Delta t^3),
\end{align}
compared to the $\mathcal{O}(\Delta t_2)$ error of the first-order Trotter:
\begin{align}
\exp(\sum_{j=1}^m\mathcal{L}_j\Delta t) = \prod_{j=1}^m \exp(\mathcal{L}_j\Delta t) + \mathcal{O}(\Delta t^2).
\end{align}
This approach can be readily extended to larger systems with multiple qubits or higher-dimensional qudits. Consequently, Trotterization provides a versatile framework for simulating the dynamics of open quantum systems with high dimensionality, offering the unique advantage of emulating quantum systems in environments with tunable noise intensities over a wide range.}

%%%%%%%%%%%%%%%%%%%%%%%%%%%%%%%%%%%%%%%%%%%%%%%%%%%%%%%%%%%%%
\subsection{Alternative Simulation Algorithms}\label{sec:qubitization}
We now discuss how bosonic Hamiltonians can be simulated using alternative algorithms such as qubitization and quantum signal processing. Specifically, we will focus on demonstrating qubitization-based Hamiltonian simulation. At the end of the section, we will also briefly review the application of quantum signal processing and its generalized version for the same task.

The basic component of this algorithm is a probabilistic encoding of the target operator using unitary operations. For instance, suppose the target operator has the decomposition $H=\sum_{\ell=1}^L\beta_\ell U_\ell$ as a linear combination of unitaries, where the coefficients $\beta_\ell$ are all positive and operators $U_\ell$ are unitaries. Then, we define
\begin{equation}
    \mathrm{PREP}|0\rangle=\frac{1}{\sqrt{\sum_{\ell=1}^L\beta_\ell}}\sum_{\ell=1}^L\sqrt{\beta_\ell}|\ell\rangle,~~
    \mathrm{SEL}=\sum_{\ell=1}^L|\ell\rangle\!\langle\ell|\otimes U_\ell,
\end{equation}
so that
\begin{equation}
    \left(\langle0|\mathrm{PREP}^\dagger\otimes I\right)\mathrm{SEL}\left(\mathrm{PREP}|0\rangle\otimes I\right)=\frac{H}{\sum_{\ell=1}^L\beta_\ell}.
\end{equation}
That is, the target Hamiltonian $H$ is encoded by the unitaries $\mathrm{PREP}$ and $\mathrm{SEL}$ with the normalization factor $1/\sum_\ell\beta_\ell$. With this encoding, one can use the so-called qubitization algorithm~\cite{Low2019hamiltonian} to approximate the time evolution $e^{-itH}$ for an accuracy $\epsilon$ by making
\begin{equation}
    \mathcal{O}\left(\sum_\ell\beta_\ell t+\log\left(\frac{1}{\epsilon}\right)\right)
\end{equation}
queries to $\mathrm{PREP}$ and $\mathrm{SEL}$.
The action of $\mathrm{PREP}$ is to prepare an $L$-dimensional state which takes $\mathcal{O}(L)$ gates~\cite{ShendeBullockMarkov06}. For $\mathrm{SEL}$, we need to cycle through the binary representation of all $\ell=1,2,\ldots,L$, which has cost $\mathcal{O}(L)$~\cite[Appendix G.4]{CMNRS18}.

Besides the linear-combination-of-unitary model, it is also possible to use the qubitization algorithm when the target operator $H$ is sparse, meaning that each row or column of $H$ has at most a constant number of nonzero elements. The underlying idea is to view the target Hamiltonian as the weighted adjacency matrix of some graph and then perform quantum walk~\cite{Childs2010}.

Note that the Hamiltonians introduced in Section \ref{sec: boson_qubit} can all be expressed as linear combinations of elementary products of spin, fermionic, bosonic operators. Assuming encodings to the elementary operators are available, it is straightforward to perform linear combinations and multiplications to encode the target Hamiltonian, and thereby simulating the Hamiltonian via qubitization. Efficient encodings of spin and fermionic operators are known from previous work such as Ref. \citenum{PhysRevX.8.041015}. Here, we focus on the encoding of bosonic Hamiltonians. 

Specifically, we consider the truncated bosonic operator
\begin{align}
    b^\dagger &=
    \begin{bmatrix}
    0 & \cdots & \cdots & \cdots & 0\\
    1 & 0 & \ddots & \ddots & \vdots\\
    0 & \sqrt{2} & \ddots & \ddots & \vdots\\
    \vdots & \ddots & \ddots & \ddots & \vdots\\
    0 & \cdots & \cdots & \sqrt{\Lambda-1} & 0
    \end{bmatrix} \notag \\
    &=\sum_{\lambda=0}^{\Lambda-1}\sqrt{(\lambda+1)\bmod\Lambda}|(\lambda+1)\bmod\Lambda\rangle\!\langle\lambda|.
\end{align}
This operator is sparse and can thus be encoded using quantum walk as discussed above. However, the resulting circuit is complicated and the square root function requires a large amount of arithmetics to implement. In the following we describe an alternative method based on a linear combination of uniatries, which can be easier to implement with a minimum amount of arithmetics.

Suppose we want to implement some nonnegative function $f(\lambda)$ for $\lambda\in[0,\Lambda]$. The core idea behind our implementation is to consider the integral representation~\cite[Appendix C]{babbush2018quantum}
\begin{equation}
    f(\lambda)=\int_{0}^{\|f\|_{\max}}{\rm d}x\ (-1)^{2x>\|f\|_{\max}+f(\lambda)}.
\end{equation}
Here, $\|f\|_{\max}=\max_{\lambda\in[0,\Lambda]}|f(\lambda)|$ is the maximum value of $f$ over the interval $[0,\Lambda]$, and $2x>\|f\|_{\max}+f(\lambda)$ is a Boolean expression that has value $1$ if the condition is true and $0$ otherwise. To prove this equality, note that the integral on the right-hand side simplifies to
\begin{equation}
    (+1)\frac{\norm{f}_{\max}+f(\lambda)}{2}
    +(-1)\left(\norm{f}_{\max}-\frac{\norm{f}_{\max}+f(\lambda)}{2}\right)
\end{equation}
which gives $f(\lambda)$ on the left-hand side.
Applying this integral representation to the bosonic operator, we have
\begin{align}
    b^\dagger&=\int_{0}^{\sqrt{\Lambda-1}}{\rm d}x\sum_{\lambda=0}^{\Lambda-1}(-1)^{2x>\sqrt{\Lambda-1}+\sqrt{(\lambda+1)\bmod\Lambda}} \notag \\
    &~~~~~~~~~~~~~~~~~~~~~~~~~~~~ \times |(\lambda+1)\bmod\Lambda\rangle\!\langle\lambda|.
\end{align}
This suggests a circuit implementation as follows. We first prepare a uniform superposition state $\mathrm{PREP}\ket{0}=\frac{1}{\sqrt{\Xi}}\sum_{\xi=0}^{\Xi-1}\ket{\xi}$ for some large value of $\Xi$ to be determined later. We then implement a cyclic shift on the register $\ket{\lambda}$ modulo $\Lambda$. We now test the inequality
\begin{equation}
    2\frac{\xi}{\Xi}\sqrt{\Lambda-1}>
    \sqrt{\Lambda-1}+\sqrt{(\lambda+1)\bmod\Lambda},
\end{equation}
and use the outcome to flip the minus sign. These two operations together define $\mathrm{SEL}$, and $\left(\langle0|\mathrm{PREP}^\dagger\otimes I\right)\mathrm{SEL}\left(\mathrm{PREP}|0\rangle\otimes I\right)$ encodes the Riemann sum
\begin{align}
    &\frac{1}{\Xi}\sum_{\xi=0}^{\Xi-1}
    \Bigg(\sum_{\lambda=0}^{\Lambda-1}(-1)^{2\frac{\xi}{\Xi}\sqrt{\Lambda-1}>\sqrt{\Lambda-1}+\sqrt{(\lambda+1)\bmod\Lambda}} \notag \\
    &~~~~~~~~~~~~~~~~~~~~\times |(\lambda+1)\bmod\Lambda\rangle\!\langle\lambda|\Bigg)
\end{align}
which approximates the normalized bosonic operator
\begin{equation}
    \frac{b^\dagger}{\sqrt{\Lambda-1}}
\end{equation}
when $\Xi$ is sufficiently large.

We now consider the gate complexity of implementing this encoding. To simplify the discussion, we assume that both $\Lambda$ and $\Xi$ are powers of $2$. Then the preparation of the uniform superposition state takes only $\log(\Xi)$ Hadamard gates. The cyclic shifting on $\ket{\lambda}$ can be realized as a binary addition on $\log(\Lambda)$ bits and thus has complexity $\mathcal{O}(\log(\Lambda))$. For the next step, we can equivalently test the following system of inequalities
\begin{equation}
\begin{aligned}
\begin{cases}
2\xi>\Xi,\\
(2\xi-\Xi)^2(\Lambda-1)>\Xi^2\left((\lambda+1)\bmod\Lambda\right),
\end{cases}
\end{aligned}
\end{equation}
which has a cost of
\begin{equation}
    \mathcal{O}(\log^2(\Xi)\log(\Lambda)).
\end{equation}
This is also the asymptotic gate complexity for the entire encoding.

We now consider choosing $\Xi$ so that the Riemann sum well approximates the integral. For a fixed value of $\lambda$, note that the Riemann sum
\begin{equation}
    \frac{1}{\Xi}\sum_{\xi=0}^{\Xi-1}
    (-1)^{2\frac{\xi}{\Xi}\sqrt{\Lambda-1}>\sqrt{\Lambda-1}+\sqrt{(\lambda+1)\bmod\Lambda}}
\end{equation}
and the integral
\begin{equation}
    \frac{1}{\sqrt{\Lambda-1}}\int_{0}^{\sqrt{\Lambda-1}}{\rm d}x\ (-1)^{2x>\sqrt{\Lambda-1}+\sqrt{(\lambda+1)\bmod\Lambda}}
\end{equation}
differ over an interval of length at most $1/\Xi$. Thus, their difference is bounded by $2/\Xi$. Since $b^\dagger$ is $1$-sparse, the error of approximating $b^\dagger$ is given by the maximum error of approximating its entries. Therefore, we choose $\Xi=\mathcal{O}(1/\delta)$ so that
\begin{equation}
    \norm{\left(\langle0|\mathrm{PREP}^\dagger\otimes I\right)\mathrm{SEL}\left(\mathrm{PREP}|0\rangle\otimes I\right)
    -\frac{b^\dagger}{\sqrt{\Lambda-1}}}
    \leq\delta
\end{equation}
and this encoding can thus be implemented with cost $\mathcal{O}(\log(\Lambda)\log^2(1/\delta))$.

\textcolor{black}{Beyond the bosonic context, general Hamiltonian simulation can also be performed using quantum signal processing (QSP)~\cite{low2017optimal,PRXQuantum.3.040305,PhysRevA.103.042419,Rossi2022multivariable,rossi2023quantumsignalprocessingcontinuous,PhysRevA.108.062413,PRXQuantum.5.020368}. The core idea of QSP is to construct a polynomial approximation of the target unitary operator $U$, such as the time propagator in Hamiltonian simulation, $U=\exp(-i\hat{H}t)$, assuming oracular access to a unitary $U$ that encodes the polynomial. More recently, generalized quantum signal processing (GQSP) has been introduced~\cite{PRXQuantum.5.020368}, offering fewer restrictions and a more economical number of operations. In GQSP, the block encoding of the polynomial of the target function is achieved by interleaving general $SU(2)$ rotation, $R(\theta,\phi,\lambda)$, on an ancillary qubit with a 0-controlled application of the target function, $A = |0\rangle \langle 0|\otimes U + |1\rangle \langle 1 | \otimes \mathbf{I}$. Notably, the application of GQSP in Hamiltonian simulation achieves the same query complexity as qubitization, requiring $\mathcal{O}(t+\log(1/\epsilon)/\log\log(1/\epsilon))$ controlled-$U$ operations and two-qubit gates, while using only a single ancillary qubit.}

We have so far constructed a quantum circuit that encodes the bosonic operator $b^\dagger$. The conjugate transpose of this circuit encodes $b$. On the other hand, the encoding of spin and fermionic operators is well studied in the context of quantum chemistry simulation~\cite{PhysRevX.8.041015}.
Then the qubitization and other algorithms allow us to implement linear combinations of products of these elementary operators, which is sufficient to simulate bosonic models, such as the spin-boson Hamiltonian \eqref{eq:spin_boson_def} and the boson-fermion Hamiltonian \eqref{Holstein}.

%%%%%%%%%%%%%%%%%%%%%%%%%%%%%%%%%%%%%%%%%%%%%%%%%%%%%

\section{Effective Hamiltonian Construction Through Coupled-Cluster Approach}

To further optimize quantum computation by reducing the number of gates and operations, a standard approach involves constructing an effective Hamiltonian for the quantum system. This Hamiltonian accurately encapsulates particle interactions while conserving resources. The Hamiltonian can follow a unitary and/or non-unitary path to transform into an effective form.

The unitary transformation has been extensively discussed in the context of understanding many-body localization phenomena~\cite{Heidbrink02_Renormalization,Kehrein94_flow,basko06_metal,Nandkishore15_many,PhysRevLett.119.075701,PhysRevLett.116.010404,PhysRevB.94.104202}. For completeness, we provide a concise overview of the unitary path in Appendix \ref{App:U_path}. Subsequently, our focus shifts to exploring some coupled cluster (CC) formulations, which are inherently non-unitary, and designing quantum algorithms specifically for bosonic systems, \textcolor{black}{complementing the typical discussion of the unitary coupled cluster theory in the quantum computing context (see Ref.~\cite{D1CS00932J} for a recent perspective)}. It is worth noting that the interplay between unitary and non-unitary features and operations becomes apparent during the Hamiltonian diagonalization and open system-solving processes. Specifically, traditional flow equations and corresponding generators can be generalized to handle non-Hermitian matrices and open quantum systems governed by, for example, Lindbladians (see, e.g. Ref. \citenum{Schmiedinghoff22_Efficient}). On the other hand, recent progress in encoding general non-unitary operators on quantum computers allows for highly accurate quantum simulations of correlated fermionic systems~\cite{PhysRevResearch.4.043172,Childs2012Quantum,brassard2002quantum,berry2015simulating,low2017optimal,Low2019hamiltonian}. Here, we extend the CC downfolding technique to bosonic systems. As will be evident in the ensuing discussion, the unitary analog of traditionally non-unitary CC formulations for bosonic systems can be meticulously designed to ensure the exactness of the ans\"{a}tz, which then guarantees that the effective Hamiltonian can be accurately and efficiently constructed via appropriate quantum algorithms.

The single reference coupled-cluster formalism for a mixture of bosons localized at different sites has been intensively studied in the context of Bose-Einstein condensates (BECs) and trapped bosonic systems~\cite{PhysRevA.73.043609,ALON2006151}. 
Let us focus attention on the system of $N$ identical  bosons 
that can occupy $M$ ``one-particle" states. In such case the dimensionality of the bosonic FCI space is 
\begin{equation}
{\rm dim}_{\rm FCI} = {M+N-1 \choose N}
\label{cc1}
\end{equation}
which grows much faster than analogous dimension of the Fermionic FCI space. 
The CC parametrization of the bosonic ground-state wave function $|\Psi\rangle$ takes analogous form as in the fermionic case
\begin{equation}
|\Psi\rangle = e^T |\phi_0\rangle \;,
\label{cc2}
\end{equation}
where the normalized reference function $|\phi_0\rangle$ for bosons is defined as 
\begin{equation}
|\phi_0\rangle = \frac{1}{\sqrt{N!}} (b_1^{\dagger})^N |\rm{vac}\rangle 
\end{equation}
and $|0\rangle$ denotes physical vacuum. The cluster operator $T$ (in the class of standard CC approximations) is represented as a sum of its many-body components $T_k$
\begin{equation}
T_k = \sum_{a_1,\ldots,a_k = 2}^{N} t_{a_1\ldots a_k} b_{a_1}^{\dagger} \ldots b_{a_k}^{\dagger} (b_1)^{k}
\label{cc3}
\end{equation}
where $[T_k,T_l]=0$. As for the non-Hermitian CC downfolding let us partition one-particle space into the lowest $M_{\rm act}$ active one-particle functions and remaining (inactive). Our goal is to  build an effective representation of the Hamiltonian that furnished the same ground-state energy in the active space the  dimension 
\begin{equation}
{\rm dim}_{\rm act} = {M_{\rm act} +N-1 \choose N}
\label{cc4}
\end{equation}
To this end, let us partition the $T$ operator into internal ($T_{\rm int}$) and external ($T_{\rm ext}$) parts 
\begin{equation}
T = T_{\rm int} + T_{\rm ext}
\label{cc4}
\end{equation}
where $T_{\rm int}$ and $T_{\rm ext}$ produce excitation within and outside of model space, respectively, when acting on the reference function $|\Phi\rangle$. For example, it means that cluster amplitudes defining $T_{\rm int}$ carry active orbital indices only, whereas external amplitudes must include at least one inactive orbital index. 

To derive bosonic variant of CC downfolding we start from the energy-independent form of the CC equations
\begin{equation}
(P+Q) He^T |\Phi\rangle = E (P+Q) e^T |\Phi\rangle \;,
\label{cc5}
\end{equation}
which at the solution is equivalent to the standard, connected, form  bosonic CC equations:
\begin{eqnarray}
Qe^{-T} H e^T |\Phi\rangle &=& 0 \;, \label{cc6} \\
\langle\Phi|e^{-T} H e^T |\Phi\rangle &=& E \;, \label{cc7}
\end{eqnarray}
In Eqs. (\ref{cc5})-(\ref{cc7}), $Q$ is the projection operator onto excited configurations generated by action of $T$ on the reference function and $P$ designates the projection operator onto the reference function. Projecting 
Eq. (\ref{cc5}) onto the active space configurations described by projection operator 
$(P+Q_{\rm int})$ (where $Q_{\rm int}$ is a projection operator onto excited configuration 
in the active space), and assuming that $e^{T_{\rm int}}|\Phi\rangle$ generates FCI-type  expansion in the active space 
(i.e., $T_{\rm int}|\Phi\rangle$ generates all excited active-space configuration)
, in analogy to fermionic case~\cite{bauman2019downfolding,downfolding2020t,bauman2022coupled,doublec2022}, one can show that CC energy can be calculated as eigenvalue of the active-sapce effective Hamiltonian $H^{\rm eff}$:
\begin{equation}
H^{\rm eff} e^{T_{\rm int}}|\Phi\rangle = E e^{T_{\rm int}}|\Phi\rangle \;,
\label{cc8}
\end{equation}
where 
\begin{equation}
H^{\rm eff}= (P+Q_{\rm int}) e^{-T_{\rm ext}} H e^{T_{\rm ext}} (P+Q_{\rm int})\;.
\label{cc9}
\end{equation}
If external cluster amplitudes are known or can be efficiently approximated, then the effective Hamiltonian can be viewed as a reduced-dimensionality representation of bosonic problem. 

As a specific example how the dimensionality of the problem can be compressed let us consider the simplest example where active space contains two orbitals ($M_{\rm act}=2$). In this case, the  general form of the internal $k$-tuply excited cluster operators take a simple  form 
\begin{equation}
T_k = c_{2\ldots 2} (b_2^{\dagger})^k (b_1)^k\;.
\label{cc10}
\end{equation}
For the case when $N=10$, $M_{\rm act}=2$, and $M=10$,
${\rm dim}_{\rm act}=11$ and 
the dimension compression defined as 
${\rm dim}_{\rm FCI}/{\rm dim}_{\rm act}$ amounts to 
$\simeq 10^5$.

%{\color{red} A sentence on unitary downfolding}
Similar to systems defined by interacting fermions, the Hermitian form of downfolding, based on double unitary Coupled Cluster (CC) ansatz, can be extended to bosonic systems. Figure \ref{fig:boson_cc} demonstrates a quantum-classical workflow targeting the ground state and an effective Hamiltonian of a three-site, two-boson model as described by Hamiltonian (\ref{boson_hamiltonian}). In this workflow, a Trotterized double unitary CC with singles and doubles (D-UCCSD) ansatz is employed to ensure exactness (see Appendix \ref{DUCC_boson} for the proof). However, unlike conventional VQE methods using the same ansatz, the excitation operators are partitioned into two subsets. Each subset corresponds to different excitation sub-manifolds and is treated separately.

As illustrated in Figure \ref{fig:boson_cc}b, part of the free parameters corresponding to the sub-manifold $Q_2$ is managed by the regular VQE module (the micro update), and the remaining portion (corresponding to the excitations between site 0 and site 1) is obtained through constructing and diagonalizing an effective Hamiltonian $H_{\rm eff}$ (the macro update). This partition significantly reduces the computational cost of direct optimization of the entire free-parameter space as required in conventional VQE routines. As depicted in Figure \ref{fig:boson_cc}c, the macro update in the proposed workflow for the studied model converges more effectively than the conventional VQE. It is worth noting that the diagonalization of $H_{\rm eff}$ can be done either classically or quantumly, depending on its size. The converged $H_{\rm eff}$ can also be preserved for later use such as quantum dynamics simulations and Green's function calculations.  

\begin{figure*}
    \centering
    \includegraphics[width=\linewidth]{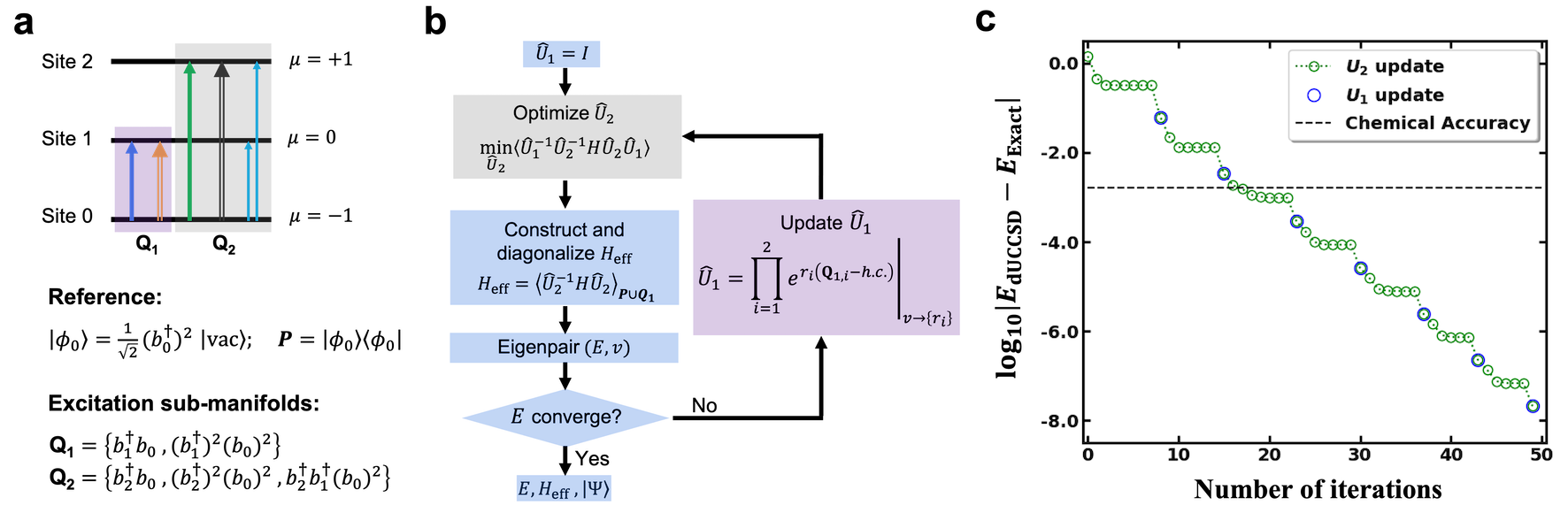}
    \caption{Proposed quantum-classical workflow and its performance for searching the ground state and constructing the effective Hamiltonian for a three-site two-boson model. The model is described by Hamiltonian (\ref{boson_hamiltonian}) with $\mu=\{-1,0,1\}$, $t=V=1$, and $U=0.5$. (\textbf{a})  In a three-site two-boson Bose-Hubbard model, the reference state is chosen to have two bosons located at site 0.  The five excitations of the model, based on this reference, can then be partitioned into two sub-manifolds, sub-manifold $\mathbf{Q}_1$ corresponding to excitations only between site 0 and site 1, and sub-manifold $\mathbf{Q}_2$, which includes all remaining excitations. (\textbf{b}) A disentangled unitary coupled cluster with singles and doubles (dUCCSD) ansatz is employed to prepare a target state within a nested optimization loop, where the outer unitary is held constant during the inner optimization. (\textbf{c}) The absolute difference between $E_{\text{dUCCSD}}$ and the exact energy as a function of the number of iterations employing the nested optimization loop.}
    \label{fig:boson_cc}
\end{figure*}

%It is worth noting that the methods of moments of coupled cluster (MMCC) equations and their implementation through a compact unitary basis and the corresponding quantum algorithms on a quantum computer (see, e.g., Ref. \citenum{PhysRevResearch.4.043172}) provide another way of constructing the effective Hamiltonian.
%
The bosonic variant of the MMCC method can be readily derived using steps similar to those used for the fermionic cases~\cite{bookmmcc,mmcc1,mmcc2,kowalski2005extensive,creom,cu2o2,piecuch2006single,crccbb,crccopen,crccobb,crccrev,kowalski2018regularized,deustua2017converging,deustua2018communication,deustua2019accurate,eriksen2020ground,gururangan2021high}. Briefly, assuming that the bosonic system is described by Hamiltonians $H$, the asymmetric energy functional can be expressed as
\begin{equation}
    E_{\rm MMCC}[\Psi_T]=\frac{\langle\Psi_T|He^{T^{(A)}}|\Phi\rangle}
    {\langle\Psi_T|e^{T^{(A)}}|\Phi\rangle}
    \label{cc11}
\end{equation}
where $|\Psi_T\rangle$ is the so-called trial wave function,  $T^{(A)}$ is an arbitrary approximation (the parent approach) to the exact cluster operator $T$, and $H$ is a many-body  Hamiltonian defined by one- and two-body interactions. When $|\Psi_T\rangle$  is replaced by the exact bosonic ground-state wave function $|\Psi\rangle$, then the value of the B-MMCC functional Eq. (\ref{cc11}) is equal to the exact ground-state energy $E$;
\begin{equation}
    E_{\rm MMCC}[\Psi]=E \;.
    \label{cc12}
\end{equation}
Assuming that low-rank moments are used to calculate cluster amplitudes, i.e.,
\begin{equation}
    Q_A M^{(A)}|\Phi\rangle = Q_A e^{-T^{(A)}}He^{-T^{(A)}}|\Phi\rangle = 0 \;. 
    \label{cc13}
\end{equation}
the many-body form of functional (\ref{cc11}) can be rewritten as 
\begin{equation}
    E_{\rm MMCC}[\Psi_T] =  E^{(A)}+\frac{\langle\Psi|e^{T^{(A)}} Q_R M^{(A)}|\Phi\rangle}
    {\langle\Psi|e^{T^{(A)}}|\Phi\rangle} \;,
    \label{eq6}
\end{equation}
where $E^{(A)}$ is the approximate CC energy and  $Q_R$ is a projection operator onto excited configurations not included in the $T^{(A)}|\Phi\rangle$ expansion. 

%%%%%%%%%%%%%%%%%%%%%%%%%%%%%%%%%%%%%%%%%%%%%%%%%%%%%%%%%%%%%
%Introduce the research component of your article, focusing on the error analysis of the truncating bosonic mode.
%Present the mathematical derivations, results, and their implications for the qubitization and quantum simulation techniques.
%Discuss the potential impact of these findings on the field of quantum computing and the development of new algorithms or hardware designs.

\section{Error Analysis of Truncating Bosonic Mode}\label{sec:trunc_err}

As mentioned earlier, bosonic modes have infinite degrees of freedom. Thus Hamiltonians that involve bosonic interactions need to be truncated before they can be simulated on a finite-memory quantum computer. In this section, we discuss how to rigorously bound the truncation error for a class of Hamiltonians with boson-fermion interactions. 
\textcolor{black}{Performing such an error analysis prior to the quantum simulation ensures that a sufficiently large basis is employed to accurately represent the system's state at each time step, while minimizing the number of qubits required.}

Specifically, given a bosonic Hamiltonian $H$, our goal is to find a truncated Hamiltonian $\widetilde{H}$, such that the time evolution $e^{-itH}\approx e^{-it\widetilde{H}}$ is well approximated when acting on a given initial state. The truncated Hamiltonian can then be simulated using a common quantum simulation algorithm such as Trotterization and Qubitization, as discussed in Sections \ref{sec:trotter} and \ref{sec:qubitization} respectively. For concreteness, we take the boson-fermion Hamiltonian \eqref{Holstein} as an example, but a similar analysis applies to the spin-boson Hamiltonian as well. 

Our setting is similar to that of a recent work~\cite{Tong2022provablyaccurate}. However, we have streamlined their analysis and improved a polylogarithmic factor over their bound, while also extending the result to time-dependent Hamiltonians. We will focus on analyzing the asymptotic scaling of the truncation cutoff, but one can keep track of all the constant factors so as to estimate the concrete resources required for a specific simulation task. Interestingly, the truncation bound we derive for the time evolution of bosonic Hamiltonians has also implications to studying their ground state properties~\cite{arealaw22}.

The remainder of this section is organized as follows.
In Section \ref{sec:ham_cond}, we present several technical conditions which are required for our analysis to hold. In Section \ref{sec:state_trunc}, we give a bound on the growth of bosonic number under the ideal Hamiltonian evolution when the initial state is restricted to a low-particle subspace. We then apply this in Section \ref{sec:ham_trunc} to bound the error in the time evolution when the Hamiltonian is truncated to finite-dimensional Hilbert spaces. We discuss generalizations to multiple bosonic modes in Section \ref{sec:multi_boson}, establishing the main result Theorem \ref{thm:multi_boson_trunc}. We show how the analysis can be adapted to Hamiltonians with explicit time dependence in Section \ref{sec:time_dependent}.

%%%%%%%%%%%%%%%%%%%%%%%%%%%%%%%%%%%%%%%%%%%%%%%%%%%%%%%%%%%%%
\subsection{Technical conditions}\label{sec:ham_cond}
Following Ref.~\citenum{Tong2022provablyaccurate}, we use $\lambda$ to denote the number of bosons in a specific bosonic mode, and write the corresponding state as $|\lambda\rangle$. Then, we define $\Pi_{{S}}$ as the projector onto the subspace of all states with bosonic number belonging to set ${S}$ and let $\overline{\Pi}_{{S}}=I-\Pi_{{S}}$ be the complementary projector. It follows from the completeness requirement that $\Pi_{[0,\infty]}=\sum_{\lambda=0}^\infty\Pi_{\lambda}=I$. However, to implement simulation algorithms on a finite-memory quantum computer, we need to choose a maximum cutoff $\Lambda$, which results in a partial projection $\Pi_{[0,\Lambda]}$ and introduces a truncation error. We will rigorously analyze this error below. 

Fixing a specific bosonic mode, our main goal is to upper bound the leakage of bosonic number under the ideal Hamiltonian evolution when the initial state has at most $\Lambda_0$ particles. 
For technical reasons, we will impose additional conditions on the Hamiltonian. Specifically, we assume that the Hamiltonian can be decomposed as
\begin{equation}
    H=H_w+H_r,
\end{equation}
where
\begin{equation}
\label{eq:ham_cond}
\begin{aligned}
    \Pi_{\lambda}H_w\Pi_{\lambda'}&=0 ~~(\text{if }|\lambda-\lambda'|>1),\\
    \left\|H_w\Pi_{[0,\Lambda]}\right\|&\leq\chi(\Lambda+1)^r,\\
    \left[H_r,\Pi_{\lambda}\right]&=0
\end{aligned}
\end{equation}
for some $\chi>0$ and $0\leq r<1$. In words, the last condition asserts that $H_r$ does not change the number of bosons in that specific mode. While the bosonic number does change under $H_w$, the amount and magnitude of this change are upper bounded by the first two conditions. Note that the conditions \eqref{eq:ham_cond} can be alternatively understood as requirements on the block structure of the Hamiltonian. Indeed, a bosonic Hamiltonian $H=H_w+H_r$ satisfies \eqref{eq:ham_cond} if and only if
\begin{align}
    H_w&=\sum_{\lambda=0}^\infty\left(\Pi_{\lambda+1}H\Pi_{\lambda}+\Pi_{\lambda}H\Pi_{\lambda+1}\right),\notag \\
    H_r&=\sum_{\lambda=0}^\infty\Pi_{\lambda}H\Pi_{\lambda}.
\end{align}

We verify these conditions for the boson-fermion Hamiltonian (\ref{Holstein}).
%\begin{align}
%    H = - \sum_{\langle i,j \rangle} v f^\dagger_i f_j + \sum_i \omega b^\dagger_i b_i + \sum_i g \omega f^\dagger_i f_i (b^\dagger_i + b_i).
%\end{align}
Without loss of generality, consider the first bosonic mode. Then, we have
\begin{align}
    H_w&=g\omega f_1^\dagger f_1(b_1^\dagger+b_1),\notag \\
    H_r&=- \sum_{\langle i,j \rangle} v f^\dagger_i f_j + \sum_i \omega b^\dagger_i b_i + \sum_{i\neq 1} g \omega f^\dagger_i f_i (b^\dagger_i + b_i). \notag
\end{align}
One can check that $H_r$ does not change the bosonic number in the first mode, whereas $H_w$ changes the number by at most $\pm 1$. Furthermore,
\begin{equation}
    \|H_w\Pi_{[0,\Lambda]}\|\leq g\omega\sqrt{4\Lambda+2}<g\omega2\sqrt{\Lambda+1}.
\end{equation}
This is because
\begin{equation}
    \left(b^\dagger+b\right)^2\leq\left(b^\dagger+b\right)^2+i^2\left(b^\dagger-b\right)^2=4b^\dagger b+2I,
\end{equation}
which implies
\begin{align}
    \|\left(b^\dagger+b\right)\Pi_{[0,\Lambda]}\|&=\sqrt{\left\|\Pi_{[0,\Lambda]}\left(b^\dagger+b\right)^2\Pi_{[0,\Lambda]}\right\|} \notag \\
    &\leq\sqrt{\left\|\Pi_{[0,\Lambda]}\left(4b^\dagger b+2I\right)\Pi_{[0,\Lambda]}\right\|} \notag \\
    &=\sqrt{4\Lambda+2}.
\end{align}

%%%%%%%%%%%%%%%%%%%%%%%%%%%%%%%%%%%%%%%%%%%%%%%%%%%%%%%%%%%%%
% \begin{equation}
% \hspace{-2.cm}
% \begin{aligned}
%      \overline{\Pi}_{[0,\Lambda]}e^{-itH}\Pi_{[0,\Lambda_0]}
%      &=\overline{\Pi}_{[0,\Lambda_r]}e^{-i\Delta t_rH}\left(\Pi_{[0,\Lambda_{r-1}]}+\overline{\Pi}_{[0,\Lambda_{r-1}]}\right)e^{-i(\Delta t_1+\cdots \Delta t_{r-1})H}\Pi_{[0,\Lambda_0]}\\
%      &=\overline{\Pi}_{[0,\Lambda_r]}e^{-i\Delta t_rH}\Pi_{[0,\Lambda_{r-1}]}e^{-i(\Delta t_1+\cdots \Delta t_{r-1})H}\Pi_{[0,\Lambda_0]}\\
%      &\quad+\overline{\Pi}_{[0,\Lambda_r]}e^{-i\Delta t_rH}\overline{\Pi}_{[0,\Lambda_{r-1}]}e^{-i\Delta t_{r-1}H}\left(\Pi_{[0,\Lambda_{r-2}]}+\overline{\Pi}_{[0,\Lambda_{r-2}]}\right)e^{-i(\Delta t_1+\cdots \Delta t_{r-2})H}\Pi_{[0,\Lambda_0]}\\
%      &=\cdots\\
% \end{aligned}
% \end{equation}

\subsection{State truncation}\label{sec:state_trunc}
To realize the truncation, we will introduce a sequence of time durations $\Delta t_1+\cdots+\Delta t_s=t$ with corresponding bosonic numbers $\Lambda_0<\Lambda_1<\cdots<\Lambda_s$.
Assuming this is done, we then proceed to bound the approximation error
\begin{align}
    e^{-itH}\Pi_{[0,\Lambda_0]}
    \approx & \Pi_{[0,\Lambda_{s}]}e^{-i\Delta t_sH}
    \Pi_{[0,\Lambda_{s-1}]}\cdots \notag \\
    &\Pi_{[0,\Lambda_2]}
    e^{-i\Delta t_2H}\Pi_{[0,\Lambda_1]}
    e^{-i\Delta t_1H}\Pi_{[0,\Lambda_0]}.
\end{align}
Specifically, we telescope the summation to get
\begin{widetext}
\begin{align}
    &\ e^{-itH}\Pi_{[0,\Lambda_0]} - \Pi_{[0,\Lambda_{s}]}e^{-i\Delta t_sH} \Pi_{[0,\Lambda_{s-1}]}\cdots \Pi_{[0,\Lambda_2]}
    e^{-i\Delta t_2H}\Pi_{[0,\Lambda_1]} e^{-i\Delta t_1H}\Pi_{[0,\Lambda_0]} \notag \\
    &= e^{-i(\Delta t_2+\cdots+\Delta t_s)H}\overline{\Pi}_{[0,\Lambda_1]}e^{-i\Delta t_1 H}\Pi_{[0,\Lambda_0]} 
    +e^{-i(\Delta t_3+\cdots+\Delta t_s)H}\overline{\Pi}_{[0,\Lambda_2]}e^{-i\Delta t_2 H}\Pi_{[0,\Lambda_1]}e^{-i\Delta t_1 H}\Pi_{[0,\Lambda_0]} \notag \\
    &~~~~ +\cdots \notag \\
    &~~~~ +\overline{\Pi}_{[0,\Lambda_{s}]}e^{-i\Delta t_sH}
    \Pi_{[0,\Lambda_{s-1}]}\cdots \Pi_{[0,\Lambda_2]} e^{-i\Delta t_2H}\Pi_{[0,\Lambda_1]} e^{-i\Delta t_1H}\Pi_{[0,\Lambda_0]}.
\end{align}
\end{widetext}
This means we can upper bound the long-time leakage by the sum of short-time leakages, and the error adds up at most linearly
\begin{widetext}
\begin{align}
    &\left\|e^{-itH}\Pi_{[0,\Lambda_0]} - \Pi_{[0,\Lambda_{s}]}e^{-i\Delta t_sH} \Pi_{[0,\Lambda_{s-1}]}\cdots \Pi_{[0,\Lambda_2]}
    e^{-i\Delta t_2H}\Pi_{[0,\Lambda_1]} e^{-i\Delta t_1H}\Pi_{[0,\Lambda_0]}\right\| 
    \leq\sum_{j=0}^{s-1}\left\|\overline{\Pi}_{[0,\Lambda_{j+1}]}e^{-i\Delta t_{j+1}H}\Pi_{[0,\Lambda_{j}]}\right\|. \label{eq:error_tot}
\end{align}
\end{widetext}

To bound the leakage $\overline{\Pi}_{[0,\Lambda']}e^{-i\Delta t H}\Pi_{[0,\Lambda]}$ for each short time $\Delta t\geq0$, we will use the interaction picture. Specifically, we have
\begin{equation}
    e^{-i\Delta tH}=e^{-i\Delta t H_r}\mathcal{T}\left\{e^{-i\int_0^{\Delta t}\mathrm{d}\tau\ e^{i\tau H_r}H_we^{-i\tau H_r}}\right\},
\end{equation}
where $\mathcal{T}$ is the time-ordering operator for time-dependent Hamiltonian evolution. Since $H_r$ does not change the bosonic number, this implies
\begin{align}
    &\overline{\Pi}_{[0,\Lambda']}e^{-i\Delta t H}\Pi_{[0,\Lambda]} \notag \\
    &=e^{-i\Delta t H_r}\overline{\Pi}_{[0,\Lambda']} \mathcal{T}\left\{e^{-i\int_0^{\Delta t}\mathrm{d}\tau\ e^{i\tau H_r}H_we^{-i\tau H_r}}\right\}\Pi_{[0,\Lambda]}.
\end{align}
We now rewrite the right-hand side using the Dyson series:
\begin{widetext}
\begin{align}
    &\overline{\Pi}_{[0,\Lambda']}\cdot \mathcal{T}\exp\left(-i\int_0^{\Delta t}\mathrm{d}\tau\ e^{i\tau H_r}H_we^{-i\tau H_r}\right)\cdot\Pi_{[0,\Lambda]} \notag \\
    & = \sum_{k=0}^\infty(-i)^k\int_{0}^{\Delta t}\mathrm{d}\tau_k\cdots\int_{0}^{\tau_3}\mathrm{d}\tau_2\int_{0}^{\tau_2}\mathrm{d}\tau_1 \overline{\Pi}_{[0,\Lambda']}e^{i\tau_k H_r}H_we^{-i\tau_k H_r}\cdots  e^{i\tau_2 H_r}H_we^{-i\tau_2 H_r}e^{i\tau_k H_1}H_we^{-i\tau_k H_1}\Pi_{[0,\Lambda]}  \\
    & = \sum_{k=\Lambda'-\Lambda}^\infty(-i)^k\int_{0}^{\Delta t}\mathrm{d}\tau_k\cdots\int_{0}^{\tau_3}\mathrm{d}\tau_2\int_{0}^{\tau_2}\mathrm{d}\tau_1
    \overline{\Pi}_{[0,\Lambda']}e^{i\tau_k H_r}H_w\Pi_{[0,\Lambda+k-1]}e^{-i\tau_k H_r}\cdots e^{i\tau_2 H_r}H_w\Pi_{[0,{\Lambda+1}]}e^{-i\tau_2 H_r}e^{i\tau_k H_1}H_w\Pi_{[0,\Lambda]}e^{-i\tau_k H_1}, \notag
\end{align}
\end{widetext}
where the last equality follows since $H_r$ does not change the bosonic number and $H_w$ only changes the number by $\pm 1$. Putting these altogether,
\begin{widetext}
\begin{align}
    \left\|\overline{\Pi}_{[0,\Lambda']}e^{-i\Delta t H}\Pi_{[0,\Lambda]}\right\|
    =&\left\|\overline{\Pi}_{[0,\Lambda']}\cdot \mathcal{T}\exp\left(-i\int_0^{\Delta t}\mathrm{d}\tau\ e^{i\tau H_r}H_we^{-i\tau H_r}\right)\cdot\Pi_{[0,\Lambda]}\right\| \notag \\
    \leq&\sum_{k=\Lambda'-\Lambda}^\infty\frac{\Delta t^k}{k!}\|H_w\Pi_{[0,\Lambda+k-1]}\|\cdots\|H_w\Pi_{[0,{\Lambda+1}]}\|\|H_w\Pi_{[0,\Lambda]}\| \notag \\
    \leq&\sum_{k=\Lambda'-\Lambda}^\infty\frac{(\chi\Delta t)^k}{k!}\sqrt{\Lambda+k}\cdots\sqrt{\Lambda+2}\sqrt{\Lambda+1}
    \leq\sum_{k=\Lambda'-\Lambda}^\infty\frac{(\chi\Delta t)^k}{k!}\sqrt{(\Lambda+k)^k} \notag \\
    \leq&\sum_{k=\Lambda'-\Lambda}^\infty\frac{(e\chi\Delta t)^k}{ek^k}2^{\frac{k}{2}-1}\left(\Lambda^{\frac{k}{2}}+k^{\frac{k}{2}}\right)
    =\sum_{k=\Lambda'-\Lambda}^\infty\left(\frac{\sqrt{2}e\chi\Delta t}{k}\right)^k\frac{\Lambda^{\frac{k}{2}}+k^{\frac{k}{2}}}{2e},
\end{align}
\end{widetext}
where we have used the inequality between the $1$- and $k/2$-norm
\begin{align}
    \Lambda+k\leq2^{1-\frac{2}{k}}\left(\Lambda^{\frac{k}{2}}+k^{\frac{k}{2}}\right)^{\frac{2}{k}}
\end{align}
and Stirling's approximation
\begin{equation}
    e\left(\frac{k}{e}\right)^k\leq k!.
\end{equation}

To proceed, we assume that the short evolution time satisfies
\begin{equation}
    0\leq\Delta t\leq\frac{1}{\chi\sqrt{\Lambda}}\leq\frac{1}{\chi}.
\end{equation}
Then, the error bound simplifies to
\begin{align}
    &\left\|\overline{\Pi}_{[0,\Lambda']}e^{-i\Delta t H}\Pi_{[0,\Lambda]}\right\| \notag \\
    &\leq\sum_{k=\Lambda'-\Lambda}^\infty\left(\frac{\sqrt{2}e\chi\Delta t}{k}\right)^k\frac{\Lambda^{\frac{k}{2}}+k^{\frac{k}{2}}}{2e} \notag \\
    &\leq\sum_{k=\Lambda'-\Lambda}^\infty\frac{1}{2e}\left(\frac{\sqrt{2}e}{k}\right)^k+\sum_{k=\Lambda'-\Lambda}^\infty\frac{1}{2e}\left(\frac{\sqrt{2}e}{\sqrt{k}}\right)^k \notag \\
    &\leq\sum_{k=\Lambda'-\Lambda}^\infty\frac{1}{e}\left(\frac{\sqrt{2}e}{\sqrt{\Lambda'-\Lambda}}\right)^k.
\end{align}
Assuming that
\begin{equation}
    \sqrt{\Lambda'-\Lambda}\geq2\sqrt{2}e,
\end{equation}
we have
\begin{align}
    &\left\|\overline{\Pi}_{[0,\Lambda']}e^{-i\Delta t H}\Pi_{[0,\Lambda]}\right\| \notag \\
    &\leq\sum_{k=\Lambda'-\Lambda}^\infty\frac{1}{e}\left(\frac{\sqrt{2}e}{\sqrt{\Lambda'-\Lambda}}\right)^k \notag \\
    &=\frac{1}{e}\left(\frac{\sqrt{2}e}{\sqrt{\Lambda'-\Lambda}}\right)^{\Lambda'-\Lambda}\sum_{k=0}^\infty\left(\frac{\sqrt{2}e}{\sqrt{\Lambda'-\Lambda}}\right)^k \notag \\
    &\leq\frac{1}{e}\left(\frac{\sqrt{2}e}{\sqrt{\Lambda'-\Lambda}}\right)^{\Lambda'-\Lambda}\sum_{k=0}^\infty\frac{1}{2^k} \notag \\
    &\leq\left(\frac{\sqrt{2}e}{\sqrt{\Lambda'-\Lambda}}\right)^{\Lambda'-\Lambda}.
\end{align}
We summarize this bound as follows.

\begin{lemma}[Short-time state truncation]
Given bosonic Hamiltonian $H=H_w+H_r$ satisfying \eqref{eq:ham_cond} with parameter $\chi>0$, we have
\begin{equation}
    \left\|\overline{\Pi}_{[0,\Lambda']}e^{-i\Delta t H}\Pi_{[0,\Lambda]}\right\|
    \leq\left(\frac{\sqrt{2}e}{\sqrt{\Lambda'-\Lambda}}\right)^{\Lambda'-\Lambda}
\end{equation}
for any $0\leq\Delta t\leq1/\chi\sqrt{\Lambda}$ and integers $0\leq\Lambda<\Lambda'$ such that $\Lambda'-\Lambda\geq8e^2$.
\end{lemma}

To extend this analysis to a long time evolution, we divide the evolution into short steps with durations $\Delta t_1+\cdots+\Delta t_s= t$ and corresponding bosonic numbers $\Lambda_0<\Lambda_1<\cdots<\Lambda_s$. \textcolor{black}{Since our primary goal is to determine $\Lambda_s$, for the simplicity of discussion, we let the bosonic number increases linearly from $0$ to $t$,} i.e.,
\begin{equation}
    \Delta\Lambda_j\equiv\Delta\Lambda ~~ \Rightarrow ~~ \Lambda_j=\Lambda_0+j\Delta\Lambda
\end{equation}
and choose the duration upper bounds
\begin{equation}
    \Delta \tau_j=\frac{1}{\chi\sqrt{\Lambda_{j-1}}}~~\Rightarrow~~ \tau_s=\sum_{j=1}^s\frac{1}{\chi\sqrt{\Lambda_0+(j-1)\Delta\Lambda}}.
\end{equation}
The total time during which this analysis applies is then at least
\begin{equation}
\begin{aligned}
    \tau_s&=\sum_{j=1}^s\frac{1}{\chi\sqrt{\Lambda_0+(j-1)\Delta\Lambda}}\\
    &\geq\int_{0}^s\mathrm{d}x\ \frac{1}{\chi\sqrt{\Lambda_0+x\Delta\Lambda}}\\
    &=\frac{2}{\chi\Delta\Lambda}\left(\sqrt{\Lambda_0+s\Delta\Lambda}-\sqrt{\Lambda_0}\right).
\end{aligned}
\end{equation}
Note that $\lim_{s\rightarrow\infty}\tau_s=\infty$, so we can choose the first integer $s$ such that $\tau_s\geq t$. Explicitly,
\begin{equation}
    s=\left\lceil\frac{1}{\Delta\Lambda}\left(\left(\sqrt{\Lambda_0}+\frac{\chi t\Delta\Lambda}{2}\right)^2-\Lambda_0\right)\right\rceil.
\end{equation}
With these time upper bounds, we simply let
\begin{align}
    \Delta t_1&=\Delta\tau_1, \notag \\
    \Delta t_2&=\Delta\tau_2, \notag \\
    &\vdots \notag \\
    \Delta t_s&=t-\left(\Delta t_1+\cdots+\Delta t_{s-1}\right)\leq \Delta\tau_s. \notag
\end{align}

Applying \eqref{eq:error_tot} and the short-time bound, we have
\begin{widetext}
\begin{align}
    & \left\|e^{-itH}\Pi_{[0,\Lambda_0]} - \Pi_{[0,\Lambda_{s}]}e^{-i\Delta t_sH} \Pi_{[0,\Lambda_{s-1}]}\cdots\Pi_{[0,\Lambda_2]}
    e^{-i\Delta t_2H}\Pi_{[0,\Lambda_1]} e^{-i\Delta t_1H}\Pi_{[0,\Lambda_0]}\right\| \notag \\
    &~~\leq \sum_{j=0}^{s-1}\left\|\overline{\Pi}_{[0,\Lambda_{j+1}]}e^{-i\Delta t_{j+1}H}\Pi_{[0,\Lambda_{j}]}\right\|
    \leq s\left(\frac{\sqrt{2}e}{\sqrt{\Delta\Lambda}}\right)^{\Delta\Lambda}.
\end{align}
\end{widetext}
For asymptotic analysis, we have
$s=\mathcal{O}\left(\sqrt{\Lambda_0}\chi^2t^2\Delta\Lambda\right)$,
which gives
\begin{widetext}
\begin{align}
    \left\|e^{-itH}\Pi_{[0,\Lambda_0]} - \Pi_{[0,\Lambda_{s}]}e^{-i\Delta t_sH} \Pi_{[0,\Lambda_{s-1}]}\cdots\Pi_{[0,\Lambda_2]}
    e^{-i\Delta t_2H}\Pi_{[0,\Lambda_1]} e^{-i\Delta t_1H}\Pi_{[0,\Lambda_0]}\right\|
    = \mathcal{O}\left(\sqrt{\Lambda_0}\chi^2t^2\Delta\Lambda\left(\frac{\sqrt{2}e}{\sqrt{\Delta\Lambda}}\right)^{\Delta\Lambda}\right).
\end{align}
\end{widetext}
We want to choose $\Delta\Lambda$ so that the leakage error is at most $\epsilon$. The scaling of $\Delta\Lambda$ can be understood via the Lambert-W function (see Ref.~\citenum{hoorfar2008inequalities} as well as Lemma \ref{lem:lambert_w} below) as
\begin{equation}
    \Delta\Lambda=\mathcal{O}\left(1+\log\left(\frac{\Lambda_0\chi t}{\epsilon}\right)\right).
\end{equation}
This then gives
\begin{align}
    \sqrt{\Lambda_s}&=\sqrt{\Lambda_0+s\Delta\Lambda}
    % =\mathcal{O}\left(\sqrt{\Lambda_0}+\chi t\Delta\Lambda\right)
    =\sqrt{\left(\sqrt{\Lambda_0}+\frac{\chi t\Delta\Lambda}{2}\right)^2+\mathcal{O}\left(\Delta\Lambda\right)} \notag \\
    &=\sqrt{\Lambda_0}+\mathcal{O}\left(\chi t\log\left(\frac{\Lambda_0\chi t}{\epsilon}\right)\right).
\end{align}

\begin{lemma}[Bounding bosonic number by the Lambert-W function]
\label{lem:lambert_w}
    For constant $b>0$ and sufficiently large $a>0$, the function
    \begin{equation}
        f(y):=a\left(\frac{b}{\sqrt{y}}\right)^y
    \end{equation}
    is monotonically decreasing for $y\geq b^2$. Furthermore, for sufficiently small $0<\epsilon=\mathcal{O}(1)$, $f(y)=\epsilon$ has a unique solution which scales like
    \begin{equation}
        y=f^{-1}(\epsilon)=\mathcal{O}\left(1+\frac{\log (a/\epsilon)}{\log\log (a/\epsilon)}\right).
    \end{equation}
\end{lemma}
\begin{proof}
The monotonicity of $f$ follows from the fact that $\sqrt{y}^y>b^y$ for all $y\geq b^2$. Since $f(b^2)=\Theta(a)$ and $f(\infty)=0$, this implies the existence and uniqueness of the solution $y=f^{-1}(\epsilon)$. Then,
\begin{align}
    a\left(\frac{b}{\sqrt{y}}\right)^y=\epsilon
    &\quad\Rightarrow\quad
    \left(\frac{\sqrt{y}}{b}\right)^y=\frac{a}{\epsilon} \notag \\
    &\quad\Rightarrow\quad
    y\log\left(\frac{\sqrt{y}}{b}\right)=\log\left(\frac{a}{\epsilon}\right).
\end{align}
Letting $x=\log(\sqrt{y}/b)$, we have $y=b^2e^{2x}$, which implies
\begin{equation}
    b^2xe^{2x}=\log\left(\frac{a}{\epsilon}\right)
    \quad\Rightarrow\quad
    2xe^{2x}=\frac{2}{b^2}\log\left(\frac{a}{\epsilon}\right).
\end{equation}
The claimed scaling now follows by solving the last equation using the Lambert-W function.
\end{proof}

\begin{corollary}[Long-time state truncation]
Given bosonic Hamiltonian $H=H_w+H_r$ satisfying \eqref{eq:ham_cond} with parameter $\chi>0$, for any $t>0$, $\epsilon>0$ and integer $\Lambda_0>0$, there exist $s$ time durations $\Delta t_1+\cdots+\Delta t_s= t$ and corresponding bosonic numbers $\Lambda_0<\Lambda_1<\cdots<\Lambda_s$, such that
\begin{widetext}
\begin{equation}
    \left\|e^{-itH}\Pi_{[0,\Lambda_0]}
    - \Pi_{[0,\Lambda_{s}]}e^{-i\Delta t_sH}
    \Pi_{[0,\Lambda_{s-1}]}\cdots\Pi_{[0,\Lambda_2]}
    e^{-i\Delta t_2H}\Pi_{[0,\Lambda_1]}
    e^{-i\Delta t_1H}\Pi_{[0,\Lambda_0]}\right\|\leq\epsilon.
\end{equation}
\end{widetext}
The final cutoff $\Lambda_s$ has the asymptotic scaling
\begin{widetext}
\begin{equation}
    \sqrt{\Lambda_s}=\sqrt{\Lambda_0+s\Delta\Lambda}
    =\sqrt{\Lambda_0}+\mathcal{O}\left(\chi t\log\left(\frac{\Lambda_0\chi t}{\epsilon}\right)\right).
\end{equation}
\end{widetext}
\end{corollary}

\subsection{Hamiltonian truncation}\label{sec:ham_trunc}
In the previous subsection, we show that when restricting the initial state to a low-particle subspace, sufficiently large cutoff values $\Lambda_0<\Lambda_1<\cdots<\Lambda_s$ can be chosen so that the following approximation holds with an arbitrarily high accuracy
\begin{align}
    e^{-itH}\Pi_{[0,\Lambda_0]}
    &\approx \Pi_{[0,\Lambda_{s}]}e^{-i\Delta t_sH}
    \Pi_{[0,\Lambda_{s-1}]}\cdots \notag \\
    &~~~~~\Pi_{[0,\Lambda_2]}
    e^{-i\Delta t_2H}\Pi_{[0,\Lambda_1]}
    e^{-i\Delta t_1H}\Pi_{[0,\Lambda_0]}.
\end{align}
We now leverage this result to truncate the Hamiltonian at some cutoff $\widetilde{\Lambda}$ so that
\begin{equation}
    e^{-itH}\Pi_{[0,\Lambda_0]}
    \approx e^{-it\Pi_{[0,\widetilde{\Lambda}]}H\Pi_{[0,\widetilde{\Lambda}]}}\Pi_{[0,\Lambda_0]}.
\end{equation}
After that, we can use a quantum algorithm to simulate $\Pi_{[0,\widetilde{\Lambda}]}H\Pi_{[0,\widetilde{\Lambda}]}$ and the outcome is provably accurate as long as the initial state has particle number at most $\Lambda_0$. For notational convenience, we define the truncated Hamiltonians
\begin{align}
    \widetilde{H} &:=\Pi_{[0,\widetilde{\Lambda}]}H\Pi_{[0,\widetilde{\Lambda}]},\quad \notag \\
    \widetilde{H}_w &:=\Pi_{[0,\widetilde{\Lambda}]}H_w\Pi_{[0,\widetilde{\Lambda}]},\quad \notag \\
    \widetilde{H}_r &:=\Pi_{[0,\widetilde{\Lambda}]}H_r\Pi_{[0,\widetilde{\Lambda}]}.
\end{align}

Our analysis proceeds in a similar way as in Subsection \ref{sec:state_trunc}. Specifically, we first consider the Hamiltonian truncation error for a short-time evolution, i.e.,
\begin{equation}
    \Pi_{[0,\Lambda']}e^{-i\Delta tH}\Pi_{[0,\Lambda]}
    \approx\Pi_{[0,\Lambda']}e^{-i\Delta t\Pi_{[0,\widetilde{\Lambda}]}H\Pi_{[0,\widetilde{\Lambda}]}}\Pi_{[0,\Lambda]}.
\end{equation}
We switch to the interaction picture for both sides of the above equation:
\begin{align}
    &\Pi_{[0,\Lambda']}e^{-i\Delta tH}\Pi_{[0,\Lambda]} \notag \\
    &~~=\Pi_{[0,\Lambda']}e^{-i\Delta t H_r}\mathcal{T}\left\{e^{-i\int_0^{\Delta t}\mathrm{d}\tau\ e^{i\tau H_r}H_we^{-i\tau H_r}}\right\}\Pi_{[0,\Lambda]} \notag \\
    &~~=e^{-i\Delta t\sum_{\lambda=0}^{\Lambda'}\Pi_{\lambda}H\Pi_{\lambda}}\Pi_{[0,\Lambda']}\mathcal{T}\left\{e^{-i\int_0^{\Delta t}\mathrm{d}\tau\ e^{i\tau H_r}H_we^{-i\tau H_r}}\right\}\Pi_{[0,\Lambda]},\notag \\
    &\Pi_{[0,\Lambda']}e^{-i\Delta t\widetilde{H}}\Pi_{[0,\Lambda]} \notag \\
    &~~=\Pi_{[0,\Lambda']}e^{-i\Delta t \widetilde{H}_r}\mathcal{T}\left\{e^{-i\int_0^{\Delta t}\mathrm{d}\tau\ e^{i\tau \widetilde{H}_r}\widetilde{H}_we^{-i\tau \widetilde{H}_r}}\right\}\Pi_{[0,\Lambda]} \notag \\
    &~~=e^{-i\Delta t\sum_{\lambda=0}^{\Lambda'}\Pi_{\lambda}H\Pi_{\lambda}}\Pi_{[0,\Lambda']}\mathcal{T}\left\{e^{-i\int_0^{\Delta t}\mathrm{d}\tau\ e^{i\tau \widetilde{H}_r}\widetilde{H}_we^{-i\tau \widetilde{H}_r}}\right\}\Pi_{[0,\Lambda]},\notag
\end{align}
where the last equality holds if we assume a sufficiently large cutoff
\begin{equation}
    \widetilde{\Lambda}\geq\Lambda'.
\end{equation}
We then expand the time-dependent Hamiltonian evolutions using Dyson series 
\begin{widetext}
% \hspace{-5cm}
\begin{align}
    & \Pi_{[0,\Lambda']}\mathcal{T}\exp\left(-i\int_0^{\Delta t}\mathrm{d}\tau\ e^{i\tau H_r}H_we^{-i\tau H_r}\right)\Pi_{[0,\Lambda]} \notag \\
    &~~~~= \sum_{k=0}^\infty(-i)^k\int_{0}^{\Delta t}\mathrm{d}\tau_k\cdots\int_{0}^{\tau_3}\mathrm{d}\tau_2\int_{0}^{\tau_2}\mathrm{d}\tau_1
    \Pi_{[0,\Lambda']}e^{i\tau_k H_r}H_we^{-i\tau_k H_r}\cdots e^{i\tau_2 H_r}H_we^{-i\tau_2 H_r}e^{i\tau_k H_1}H_we^{-i\tau_k H_1}\Pi_{[0,\Lambda]},\\
    & \Pi_{[0,\Lambda']}\mathcal{T}\exp\left(-i\int_0^{\Delta t}\mathrm{d}\tau\ e^{i\tau \widetilde{H}_r}\widetilde{H}_we^{-i\tau \widetilde{H}_r}\right)\Pi_{[0,\Lambda]} \notag \\
    &~~~~= \sum_{k=0}^\infty(-i)^k\int_{0}^{\Delta t}\mathrm{d}\tau_k\cdots\int_{0}^{\tau_3}\mathrm{d}\tau_2\int_{0}^{\tau_2}\mathrm{d}\tau_1
    \Pi_{[0,\Lambda']}e^{i\tau_k \widetilde{H}_r}\widetilde{H}_we^{-i\tau_k \widetilde{H}_r}\cdots e^{i\tau_2 \widetilde{H}_r}\widetilde{H}_we^{-i\tau_2 \widetilde{H}_r}e^{i\tau_k \widetilde{H}_1}\widetilde{H}_we^{-i\tau_k \widetilde{H}_1}\Pi_{[0,\Lambda]}.
\end{align}
\end{widetext}
Comparing the two expansions, we see that the terms agree at order $k=0,1,\ldots,\Lambda'-\Lambda-1$; hence the error only comes from higher order terms in both summations
\begin{widetext}
\begin{align}
    \left\|\Pi_{[0,\Lambda']}e^{-i\Delta tH}\Pi_{[0,\Lambda]}
    -\Pi_{[0,\Lambda']}e^{-i\Delta t\Pi_{[0,\widetilde{\Lambda}]}H\Pi_{[0,\widetilde{\Lambda}]}}\Pi_{[0,\Lambda]}\right\|
    \leq 2\sum_{k=\Lambda'-\Lambda}^\infty\frac{\Delta t^k}{k!}\|H_w\Pi_{[0,\Lambda+k-1]}\|\cdots\|H_w\Pi_{[0,{\Lambda+1}]}\|\|H_w\Pi_{[0,\Lambda]}\|.
\end{align}
\end{widetext}
Proceeding as in Subsection \ref{sec:state_trunc}, we obtain:
\begin{lemma}[Short-time Hamiltonian truncation]
Given bosonic Hamiltonian $H=H_w+H_r$ satisfying \eqref{eq:ham_cond} with parameter $\chi>0$, we have
\begin{widetext}
\begin{align}
    &\left\|\Pi_{[0,\Lambda']}\left(e^{-i\Delta t H}-e^{-i\Delta t\Pi_{[0,\widetilde{\Lambda}]}H\Pi_{[0,\widetilde{\Lambda}]}}\right)\Pi_{[0,\Lambda]}\right\| 
    \leq2\left(\frac{\sqrt{2}e}{\sqrt{\Lambda'-\Lambda}}\right)^{\Lambda'-\Lambda}
\end{align}
\end{widetext}
for any $0\leq\Delta t\leq1/\chi\sqrt{\Lambda}$ and integers $0\leq\Lambda<\Lambda'\leq\widetilde{\Lambda}$ such that $\Lambda'-\Lambda\geq8e^2$.
\end{lemma}

This bound may be generalized using the triangle inequality to truncate a long-time Hamiltonian evolution as follows:
\begin{corollary}[Long-time Hamiltonian truncation]
Given bosonic Hamiltonian $H=H_w+H_r$ satisfying \eqref{eq:ham_cond} with parameter $\chi>0$, for any $t>0$, $\epsilon>0$ and integer $\Lambda_0>0$, there exist $s$ time durations $\Delta t_1+\cdots+\Delta t_s= t$ and corresponding bosonic numbers $\Lambda_0<\Lambda_1<\cdots<\Lambda_s\leq\widetilde{\Lambda}$, such that
\begin{widetext}
\begin{align}
    &\Big\|\Pi_{[0,\Lambda_{s}]}e^{-i\Delta t_sH} \Pi_{[0,\Lambda_{s-1}]}\cdots\Pi_{[0,\Lambda_2]}
    e^{-i\Delta t_2H}\Pi_{[0,\Lambda_1]} e^{-i\Delta t_1H}\Pi_{[0,\Lambda_0]} \notag \\
    &~~~~~~~~~~~~~~~~~~~~~~~~ - \Pi_{[0,\Lambda_{s}]}e^{-i\Delta t_s\widetilde{H}}
    \Pi_{[0,\Lambda_{s-1}]}\cdots\Pi_{[0,\Lambda_2]} e^{-i\Delta t_2\widetilde{H}}\Pi_{[0,\Lambda_1]}
    e^{-i\Delta t_1\widetilde{H}}\Pi_{[0,\Lambda_0]}\Big\|\leq\epsilon.
\end{align}
\end{widetext}
The final cutoff $\Lambda_s$ has the asymptotic scaling
\begin{align}
    \sqrt{\Lambda_s}&=\sqrt{\Lambda_0+s\Delta\Lambda} \notag \\
    &=\sqrt{\Lambda_0}+\mathcal{O}\left(\chi t\log\left(\frac{\Lambda_0\chi t}{\epsilon}\right)\right).
\end{align}
\end{corollary}

It remains to revert the state truncation:
\begin{align}
    &\Pi_{[0,\Lambda_{s}]}e^{-i\Delta t_s\widetilde{H}}
    \Pi_{[0,\Lambda_{s-1}]}\cdots \notag \\
    &~~~~~~~~\Pi_{[0,\Lambda_2]}
    e^{-i\Delta t_2\widetilde{H}}\Pi_{[0,\Lambda_1]}
    e^{-i\Delta t_1\widetilde{H}}\Pi_{[0,\Lambda_0]}
    \approx e^{-it\widetilde{H}}\Pi_{[0,\Lambda_0]}, \notag
\end{align}
which proceeds similarly as in Subsection \ref{sec:state_trunc}. Setting the accuracy parameter to be $\epsilon/3$ for each of the following three approximations,
\begin{widetext}
\begin{align}
    e^{-itH}\Pi_{[0,\Lambda_0]} & \stackrel{\epsilon/3}{\approx} \Pi_{[0,\Lambda_{s}]}e^{-i\Delta t_sH} \Pi_{[0,\Lambda_{s-1}]}\cdots\Pi_{[0,\Lambda_2]}
    e^{-i\Delta t_2H}\Pi_{[0,\Lambda_1]} e^{-i\Delta t_1H}\Pi_{[0,\Lambda_0]} \notag \\
    & \stackrel{\epsilon/3}{\approx} \Pi_{[0,\Lambda_{s}]}e^{-i\Delta t_s\widetilde{H}} \Pi_{[0,\Lambda_{s-1}]}\cdots\Pi_{[0,\Lambda_2]}
    e^{-i\Delta t_2\widetilde{H}}\Pi_{[0,\Lambda_1]} e^{-i\Delta t_1\widetilde{H}}\Pi_{[0,\Lambda_0]} \notag \\
    & \stackrel{\epsilon/3}{\approx} e^{-it\widetilde{H}}\Pi_{[0,\Lambda_0]}
\end{align}
\end{widetext}
we finally obtain:
\begin{theorem}[Bosonic Hamiltonian truncation]
\label{thm:boson_trunc}
Given bosonic Hamiltonian $H=H_w+H_r$ satisfying \eqref{eq:ham_cond} with parameter $\chi>0$, for any $t>0$, $\epsilon>0$ and integer $\Lambda_0>0$, there exists an integer $\widetilde{\Lambda}>0$ such that
\begin{equation}
    \left\|e^{-itH}\Pi_{[0,\Lambda_0]}
    - e^{-it\Pi_{[0,\widetilde{\Lambda}]}H\Pi_{[0,\widetilde{\Lambda}]}}\Pi_{[0,\Lambda_0]}\right\|\leq\epsilon.
\end{equation}
The cutoff $\widetilde{\Lambda}$ has the asymptotic scaling
\begin{equation}
    \sqrt{\widetilde{\Lambda}}
    =\sqrt{\Lambda_0}+\mathcal{O}\left(\chi t\log\left(\frac{\Lambda_0\chi t}{\epsilon}\right)\right). \label{eq:truncation_threshold}
\end{equation}
\end{theorem}

%%%%%%%%%%%%%%%%%%%%%%%%%%%%%%%%%%%%%%%%%%%%%%%%%%%%%%%%%%%%%
\subsection{Multiple bosonic modes}\label{sec:multi_boson}
In our above analysis, we have fixed a specific bosonic mode and bounded the growth of occupation number under the time evolution. We now discuss how this result can be generalized to Hamiltonians with $N$ bosonic modes.

Specifically, we use $\Pi_{{S}}^{(j)}$ to denote the projector onto states of the $j$th mode with occupation number from set ${S}$. These projectors commute with each other for different bosonic modes. Then, the initial state subspace and the Hamiltonian truncation are respectively determined by the projectors
\begin{equation}
    \Pi_{[0,\Lambda_0]}^{(all)}:=\prod_{j=1}^N\Pi_{[0,\Lambda_0]}^{(j)},~~
    \Pi_{[0,\widetilde{\Lambda}]}^{(all)}:=\prod_{j=1}^N\Pi_{[0,\widetilde{\Lambda}]}^{(j)}.
\end{equation}
Our goal is to analyze the scaling of the maximum cutoff $\widetilde{\Lambda}$ for all bosonic modes such that
\begin{equation}
    e^{-itH}\Pi_{[0,\Lambda_0]}^{(all)}\approx e^{-it\Pi_{[0,\widetilde{\Lambda}]}^{(all)}H\Pi_{[0,\widetilde{\Lambda}]}^{(all)}}\Pi_{[0,\Lambda_0]}^{(all)}.
\end{equation}
This can be achieved by setting the accuracy to be $\epsilon/N$ for each of the following approximations
\begin{align}
    e^{-itH}\Pi_{[0,\Lambda_0]}^{(all)}
    &\stackrel{\epsilon/N}{\approx}e^{-it\Pi_{[0,\widetilde{\Lambda}]}^{(1)}H\Pi_{[0,\widetilde{\Lambda}]}^{(1)}}\Pi_{[0,\Lambda_0]}^{(all)} \notag \\
    &\stackrel{\epsilon/N}{\approx}e^{-it\Pi_{[0,\widetilde{\Lambda}]}^{(2)}\Pi_{[0,\widetilde{\Lambda}]}^{(1)}H\Pi_{[0,\widetilde{\Lambda}]}^{(2)}\Pi_{[0,\widetilde{\Lambda}]}^{(1)}}\Pi_{[0,\Lambda_0]}^{(all)} \notag \\
    &\stackrel{\epsilon/N}{\approx}\cdots \notag \\
    &\stackrel{\epsilon/N}{\approx}e^{-it\Pi_{[0,\widetilde{\Lambda}]}^{(all)}H\Pi_{[0,\widetilde{\Lambda}]}^{(all)}}\Pi_{[0,\Lambda_0]}^{(all)}.
\end{align}
Note that any intermediate truncated Hamiltonian
\begin{equation}
    \prod_{u=1}^j\Pi_{[0,\widetilde{\Lambda}]}^{(u)}H\prod_{u=1}^j\Pi_{[0,\widetilde{\Lambda}]}^{(u)}
\end{equation}
is a bosonic model that satisfies conditions \eqref{eq:ham_cond} with the same parameter $\chi$ as in the original Hamiltonian. It thus follows from Theorem \ref{thm:boson_trunc} that the maximum cutoff has the scaling
\begin{equation}
    \sqrt{\widetilde{\Lambda}}
    =\sqrt{\Lambda_0}+\mathcal{O}\left(\chi t\log\left(\frac{N\Lambda_0\chi t}{\epsilon}\right)\right). \label{eq:truncation_threshold_N}
\end{equation}

\begin{theorem}[$N$-mode bosonic Hamiltonian truncation]
\label{thm:multi_boson_trunc}
Given an $N$-mode bosonic Hamiltonian $H$, suppose that $H=H_w+H_r$ satisfying \eqref{eq:ham_cond} with parameter $\chi>0$ for all the bosonic modes. For any $t>0$, $\epsilon>0$ and integer $\Lambda_0>0$, there exist an integer $\widetilde{\Lambda}>0$ such that
\begin{equation}
    \left\|e^{-itH}\Pi_{[0,\Lambda_0]}^{(all)}
    - e^{-it\Pi_{[0,\widetilde{\Lambda}]}^{(all)}H\Pi_{[0,\widetilde{\Lambda}]}^{(all)}}\Pi_{[0,\Lambda_0]}^{(all)}\right\|\leq\epsilon.
\end{equation}
The cutoff $\widetilde{\Lambda}$ has the asymptotic scaling
\begin{equation}
    \sqrt{\widetilde{\Lambda}}
    =\sqrt{\Lambda_0}+\mathcal{O}\left(\chi t\log\left(\frac{N\Lambda_0\chi t}{\epsilon}\right)\right). \label{eq:multi_truncation_threshold}
\end{equation}
\end{theorem}

In Figure \ref{fig: truncation_threshold} we plot the truncation threshold $\tilde{\Lambda}(t)$ required to ensure the time propagation error due to the Hamiltonian truncation is below a predefined error $\epsilon$ for the Holsten model (\ref{Holstein}), which is a special case of the Hubbard-Holstein model without on-site interaction. We assume the initial state is a tensor product between the fermionic ground state and a quantum state of the bosonic modes that has at most $\Lambda_0$ = 1 particles in each mode. We compare the scaling of $\tilde{\Lambda}(t)$ in this work with the ones obtained in previous works,\cite{Tong2022provablyaccurate} 
observing a lower truncation threshold when the system size becomes larger or when the precision requirement is higher.
% We clearly see that when the system size becomes larger or when the precision requirement is higher, our method yields a lower truncation threshold than the two methods discussed in Ref. \citenum{Tong2022provablyaccurate}.

%\textcolor{red}{BP: Guang-Hao comments on cross-over between linear scaling and quadratical scaling, the disscussion of the qubitization of (57), and the condition of the Hamiltonian coefficients in (57) under which (197) is in favour of efficient quantum simulation.}

\begin{figure*}
    \centering
    \includegraphics[width=\linewidth]{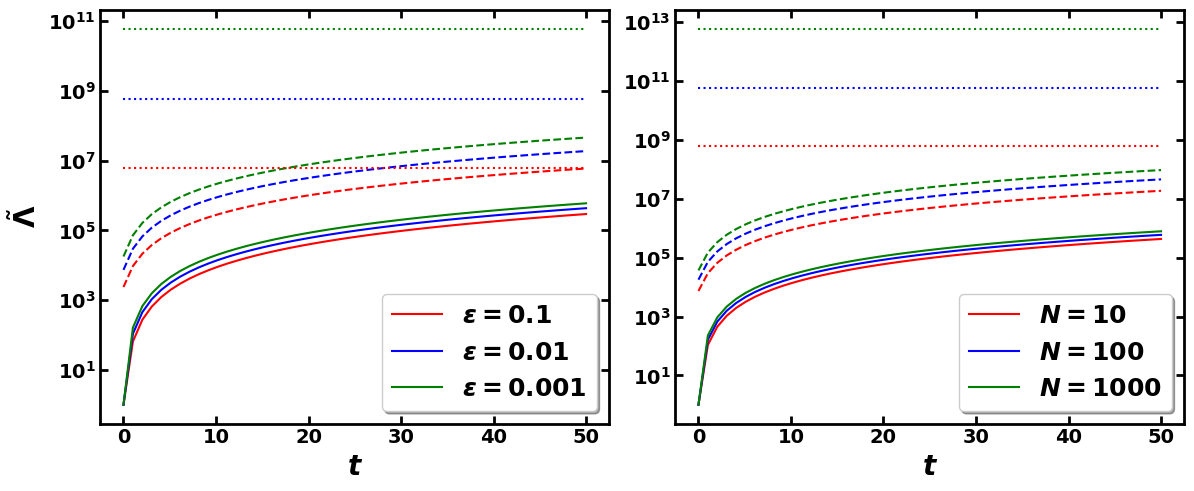}
    \caption{The truncation threshold $\tilde{\Lambda}(t)$ required to keep the error below different $\epsilon$ (left, $N$ is fixed to 100) for the time propagator of boson-fermion Hamiltonian of different size (right, $\epsilon$ is fixed to 0.001). The horizontal dotted lines show the energy-based truncation threshold, i.e., $\tilde{\Lambda} = \mathcal{O}\left( N(N+\omega E)/\omega^2\epsilon^2 \right)$ with $E \sim \mathcal{O}(N)$ (see Appendix K.2 in Ref. \citenum{Tong2022provablyaccurate} for detailed derivation). The dashed curves are the truncation threshold using the upper bound reported in Ref. \citenum{Tong2022provablyaccurate}. $\tilde{\Lambda}(t)$ in this work (solid curves) are obtained employing (\ref{eq:truncation_threshold_N}) for the Holstein model (\ref{Holstein}) with $\Lambda_0=1$ and $g\omega=1$.}
    \label{fig: truncation_threshold}
\end{figure*}

%%%%%%%%%%%%%%%%%%%%%%%%%%%%%%%%%%%%%%%%%%%%%%%%%%%%%%%%%%%%%
\subsection{Time-dependent Hamiltonians}\label{sec:time_dependent}
We now extend our error analysis to dynamics generated by time-dependent boson-fermion Hamiltonians $H(\tau)$. The resulting evolution
\begin{equation}
    \frac{d}{d\tau}U(\tau)=-iU(\tau)\quad (0\leq\tau\leq t),~~
    U(0)=I
\end{equation}
is described by a time-ordered exponential, which can be further expressed as a Dyson series
\begin{widetext}
\begin{align}
    U(t)&=\mathcal{T}\left\{e^{-i\int_0^td\tau\ H(\tau)}\right\} \notag \\
    &=\sum_{k=0}^\infty(-i)^k\int_0^td\tau_k\cdots\int_0^{\tau_3}d\tau_2\int_0^{\tau_2}d\tau_1\ H(\tau_k)\cdots H(\tau_2)H(\tau_1).
\end{align}
\end{widetext}

We first fix a specific bosonic mode and let the initial state have at most $\Lambda_0$ particles. Like in the above discussion, we assume that the Hamiltonian can be decomposed as
\begin{equation}
    H(\tau)=H_w(\tau)+H_r(\tau),
\end{equation}
where
\begin{align}
    \Pi_{\lambda}H_w(\tau)\Pi_{\lambda'}&=0,~~(\text{if }|\lambda-\lambda'|>1), \notag \\
    \left\|H_w(\tau)\Pi_{[0,\Lambda]}\right\|&\leq\chi(\tau)(\Lambda+1)^r, \notag \\
    \left[H_r(\tau),\Pi_{\lambda}\right]&=0,
\end{align}
for some $\chi(\tau)>0$ and $0\leq r<1$. These conditions are similar to the ones for time-independent Hamiltonians \eqref{eq:ham_cond}, except we require them to hold for all instantaneous times $\tau$.

Our analysis for the time-dependent case will be similar to that for the time-independent case, so we will only highlight the key steps. For the state truncation, the error still adds up at most linearly, except we need to modify the error bound to include the time dependence 
\begin{widetext}
\begin{align}
    &\bigg\|\mathcal{T}\left\{e^{-i\int_0^td\tau\ H(\tau)}\right\}\Pi_{[0,\Lambda_0]}
     - \Pi_{[0,\Lambda_{s}]}\mathcal{T}\left\{e^{-i\int_{t_{s-1}}^{t_{s}}d\tau\ H(\tau)}\right\}
    \Pi_{[0,\Lambda_{s-1}]}\cdots\Pi_{[0,\Lambda_2]}
    e^{-i\int_{t_{1}}^{t_{2}}d\tau\ H(\tau)}\Pi_{[0,\Lambda_1]}
    e^{-i\int_{t_{0}}^{t_{1}}d\tau\ H(\tau)}\Pi_{[0,\Lambda_0]}\bigg\| \notag \\
    &\leq\sum_{j=0}^{s-1}\left\|\overline{\Pi}_{[0,\Lambda_{j+1}]}\mathcal{T}\left\{e^{-i\int_{t_{j}}^{t_{j+1}}d\tau\ H(\tau)}\right\}\Pi_{[0,\Lambda_{j}]}\right\|,
\end{align}
\end{widetext}
where
\begin{equation}
    0=t_0\leq t_1\leq t_2\leq\ldots\leq t_{s-1}\leq t_s=t.
\end{equation}

For each short time leakage, we may without loss of generality shift the time interval to $0\leq\tau\leq\Delta t$. Then we have the interaction picture evolution~\cite[Lemma A.2]{CSTWZ21}
\begin{widetext}
\begin{align}
    \mathcal{T}\left\{e^{-i\int_0^{\Delta t}d\tau\ H(\tau)}\right\}
    &=\mathcal{T}\left\{e^{-i\int_0^{\Delta t}d\tau\ H_r(\tau)}\right\} 
    \cdot\mathcal{T}\left\{e^{
    -i\int_0^{\Delta t}d\tau\ \mathcal{T}\left\{e^{i\int_0^\tau ds H_r(s)}\right\}\cdot
    H_w(\tau)\cdot \mathcal{T}\left\{e^{-i\int_0^\tau ds H_r(s)}\right\} 
    }\right\}.
\end{align}
\end{widetext}
Note that the evolution generated by $H_r(\tau)$ necessarily preserves the bosonic number
\begin{equation}
    \left[\mathcal{T}\exp\left(-i\int_0^{\Delta t}d\tau\ H_r(\tau)\right),\Pi_\lambda\right]=0,
\end{equation}
which follows by applying the condition $[H_r(\tau),\Pi_\lambda]=0$ in the Dyson series. Thus, we have
\begin{widetext}
\begin{align}
    &\left\|\overline{\Pi}_{[0,\Lambda']}\mathcal{T}\exp\left(-i\int_{0}^{\Delta t}d\tau\ H(\tau)\right)\Pi_{[0,\Lambda]}\right\| \notag \\
    &~~\leq\sum_{k=\Lambda'-\Lambda}^\infty \int_{0}^{\Delta t}\mathrm{d}\tau_k\cdots\int_{0}^{\tau_3}\mathrm{d}\tau_2\int_{0}^{\tau_2}\mathrm{d}\tau_1
    \|H_w(\tau_k)\Pi_{[0,\Lambda+k-1]}\|\cdots\|H_w(\tau_2)\Pi_{[0,{\Lambda+1}]}\|\|H_w(\tau_1)\Pi_{[0,\Lambda]}\| \notag \\
    &~~\leq\sum_{k=\Lambda'-\Lambda}^\infty\frac{(\int_0^{\Delta t}d\tau\ \chi(\tau))^k}{k!}\sqrt{\Lambda+k}\cdots\sqrt{\Lambda+2}\sqrt{\Lambda+1}.
\end{align}
\end{widetext}
We now proceed as in the time-independent case to obtain the following leakage bound
\begin{equation}
    \left\|\overline{\Pi}_{[0,\Lambda']}\mathcal{T}\left\{e^{-i\int_{0}^{\Delta t}d\tau\ H(\tau)}\right\}\Pi_{[0,\Lambda]}\right\|
    \leq\left(\frac{\sqrt{2}e}{\sqrt{\Lambda'-\Lambda}}\right)^{\Lambda'-\Lambda}
\end{equation}
for any $0\leq\int_0^{\Delta t}d\tau\ \chi(\tau)\leq1/\sqrt{\Lambda}$ and integers $0\leq\Lambda<\Lambda'$ such that $\Lambda'-\Lambda\geq8e^2$.

To extend this analysis to a long time evolution, we divide the evolution into short steps with durations $\Delta t_1+\cdots+\Delta t_s= t$ and corresponding bosonic numbers $\Lambda_0<\Lambda_1<\cdots<\Lambda_s$. We let the bosonic number increase linearly, i.e.,
\begin{equation}
    \Delta\Lambda_j\equiv\Delta\Lambda~~\Rightarrow~~\Lambda_j=\Lambda_0+j\Delta\Lambda
\end{equation}
and choose the duration upper bounds
\begin{equation}
    \int_{\tau_{j-1}}^{\tau_j}d\tau\ \chi(\tau)=\frac{1}{\sqrt{\Lambda_{j-1}}},
\end{equation}
which implies
\begin{align}
    \int_{0}^{\tau_s}d\tau\ \chi(\tau)&=\sum_{j=1}^s\frac{1}{\sqrt{\Lambda_0+(j-1)\Delta\Lambda}} \notag \\
    &\geq\frac{2}{\Delta\Lambda}\left(\sqrt{\Lambda_0+s\Delta\Lambda}-\sqrt{\Lambda_0}\right).
\end{align}
Note that $\lim_{s\rightarrow\infty}\int_{0}^{\tau_s}d\tau\ \chi(\tau)=\infty$, so we can choose the first integer $s$ such that $\int_{0}^{\tau_s}d\tau\ \chi(\tau)\geq \int_{0}^{t}d\tau\ \chi(\tau)$. Explicitly,
\begin{equation}
    s=\left\lceil\frac{1}{\Delta\Lambda}\left(\left(\sqrt{\Lambda_0}+\frac{\int_{0}^{t}d\tau\ \chi(\tau)\Delta\Lambda}{2}\right)^2-\Lambda_0\right)\right\rceil.
\end{equation}
With these time upper bounds, we simply let
\begin{align}
    \Delta t_1&=\Delta\tau_1,\ \notag \\
    \Delta t_2&=\Delta\tau_2,\ \ldots,\ \notag \\
    \Delta t_s&=t-\left(\Delta t_1+\cdots+\Delta t_{s-1}\right)\leq \Delta\tau_s.
\end{align}
The remaining analysis proceeds in a similar way as in the time-independent case. This yields a cutoff value
%\begin{equation}
%    \sqrt{\widetilde{\Lambda}}
%    =\sqrt{\Lambda_0}+\mathcal{O}\left(\int_{0}^{t}{\rm d}\tau\ \chi(\tau)\log\left(\frac{\Lambda_0\int_{0}^{t}d\tau\ \chi(\tau)}{\epsilon}\right)\right).
%\end{equation}
%for a single bosonic mode and
\begin{equation}
    \sqrt{\widetilde{\Lambda}}
    =\sqrt{\Lambda_0}+\mathcal{O}\left(\int_{0}^{t}{\rm d}\tau\ \chi(\tau)\log\left(\frac{N\Lambda_0\int_{0}^{t}d\tau\ \chi(\tau)}{\epsilon}\right)\right).
\end{equation}
for $N\geq 1$ bosonic modes.

%%%%%%%%%%%%%%%%%%%%%%%%%%%
\section{Conclusion}

In this tutorial review, we have elucidated and analyzed the methodologies and techniques employed for the quantum simulation of various types of boson-related model Hamiltonians, from qubit mapping to state preparation and evolution. We have discussed the construction of the effective Hamiltonian of bosonic quantum systems within the coupled cluster context through unitary or non-unitary flows and the design of corresponding hybrid quantum algorithms potentially suitable for future larger-scale applications.

A critical aspect of practical quantum simulation involves truncating the bosonic mode. Hence, we have devoted special attention to the error analysis of such truncations. We detail the mathematical derivations, their results, and the implications these have on the techniques used in quantum simulations. Notably, our bound applies to time-dependent Hamiltonians and is more stringent than a recent bound from Ref.~\citenum{Tong2022provablyaccurate}, aiding in the development and evaluation of novel algorithms and approximate approaches.

Looking ahead, the challenge and importance of simulating open quantum systems accurately and efficiently will only increase, especially for systems where fermions interact with bosonic modes. Customized quantum simulation and effective Hamiltonian techniques for boson-related quantum systems, guided by associated error analysis, can serve as valuable resources in the field. They inspire further studies to advance our understanding and capabilities in quantum simulations.

\section{Acknowledgments}
B.P., D.C. and K.K. are supported by  the ``Embedding QC into Many-body Frameworks for Strongly Correlated Molecular and Materials Systems'' project, which is funded by the U.S. Department of Energy (DOE), Office of Science, Office of Basic Energy Sciences, the Division of Chemical Sciences, Geosciences, and Biosciences. This work was supported by the Quantum Science Center, a National Quantum Information Science Research Center of the DOE.
KK also acknowledges the support from the Center for MAny-Body Methods, Spectroscopies, and Dynamics for Molecular POLaritonic Systems (MAPOL)
under FWP 79715, which is funded as part of the Computational Chemical Sciences (CCS) program by the U.S. Department of Energy, Office of Science, Office of Basic Energy Sciences, Division of Chemical Sciences, Geosciences and Biosciences at Pacific Northwest National Laboratory (PNNL).
%the support from the "Many-Body Methods, Spectroscopies, and Dynamics for Molecular Polaritonic Systems" project  funded by the U.S. Department of Energy (DOE), Office of Science, Office of Basic Energy Sciences, the Division of Chemical Sciences, Geosciences, and Biosciences.
The Pacific Northwest National Laboratory is operated by Battelle for the DOE under Contract DE-AC05-76RL01830.

\appendix
\section{Quantum Simulation of Spin-Boson Model}\label{app_a}

We compute the time evolution of a spin-boson model according to a Lindblad master equation that accounts for experimental imperfections causing heating and dephasing of the motional mode~\cite{An_2015}
\begin{align}
    \frac{\partial \hat{\rho}(t)}{\partial t} 
    = &-i[\hat{H},\hat{\rho}(t)] 
    + \Gamma\bigg[ 2 \hat{n} \hat{\rho}(t)\hat{n} - \{\hat{n}\hat{n},\hat{\rho}(t)\} \bigg] \notag \\
    &+ \gamma\bigg[ 2 b \hat{\rho}(t) b^\dagger - \{b b^\dagger ,\hat{\rho}(t)\} \notag \\
    &~~~~~~~~ + 2 b^\dagger \hat{\rho}(t) b - \{ \hat{n},\hat{\rho}(t)\}  \bigg] \label{eq:Lindblad}
\end{align}
where $\Gamma$ and $\gamma$ are dephasing parameter and heating rate, respectively, and $\{\cdot,\cdot\}$ is the anticommutator, i.e.,
\begin{align}
    \{A,B\} = AB + BA.
\end{align}
The classical simulation can be done using, for example, a fourth order Runge-Kutta method with a relatively small time step (e.g. $10^{-3}$) to achieve the numerical stability of the simulation in the studied time regime.
To perform the quantum simulation of Eq. (\ref{eq:Lindblad}), we first transform (\ref{eq:Lindblad}) to its vectorized form. This can be done through Choi-Jamio{\l}kowski isomorphism (also called channel-state duality)~\cite{CHOI1975285,JAMIOLKOWSKI1972275,PhysRevA.87.022310},i.e.,
\begin{align}
    |i\rangle\langle j| = |j \rangle\otimes|i\rangle,
\end{align}
which from a matrix point of view is same as transforming the matrix to a vector through stacking its columns, e.g.
\begin{align}
    M = \left[\begin{array}{cc}
        m_{11} & m_{12} \\
        m_{21} & d_{22}
    \end{array}\right] \Rightarrow 
    \vec{M} = \left[ \begin{array}{c}
      m_{11} \\ m_{21} \\ m_{12} \\ m_{22}
    \end{array}\right].
\end{align}
There are two useful properties of this vectorization
\begin{align}
    \left\{ \begin{array}{l}
     \Tr{A^\dagger B} = \vec{A}^\dagger \vec{B} \\
     M = ABC \Rightarrow \vec{M} = (C^T\otimes A) \vec{B}
    \end{array}
    \right. ~~,
\end{align}
from which one can rewrite the operator products in (\ref{eq:Lindblad}), e.g.
\begin{align}
    \hat{H}\hat{\rho} = \hat{H}\hat{\rho} I 
    &\Rightarrow 
    (I\otimes \hat{H})\vec{\rho} \notag \\
    \hat{\rho}\hat{H} = I\hat{\rho}\hat{H} 
    &\Rightarrow 
    (\hat{H}^T \otimes I)\vec{\rho} \notag \\
    b^\dagger \hat{\rho} b 
    &\Rightarrow 
    (b^T \otimes b^\dagger) \vec{\rho} \notag \\
    b \hat{\rho} b^\dagger 
    &\Rightarrow 
    (b^\ast \otimes b) \vec{\rho} \notag \\
    \vdots& \notag 
\end{align}
Here, $\ast$, $T$, and $\dagger$ denote the complex conjugate, transpose, and adjoint operators, respectively. Then (\ref{eq:Lindblad}) can be re-expressed as 
\begin{align}
    \frac{\partial \vec{\rho}(t)}{\partial t} = \hat{\mathcal{L}} \vec{\rho}(t) ~~\Rightarrow~~  \vec{\rho}(t) =  e^{\tilde{\mathcal{L}}t} \vec{\rho}(0) 
\end{align}
with the Lindbladian $\hat{\mathcal{L}}$ defined as
\begin{align}
    \hat{\mathcal{L}} &= -i I\otimes \hat{H} + i\hat{H}^T\otimes I \notag \\
    &~~~~ + \Gamma\bigg[ 2 \hat{n}^T\otimes \hat{n} - I\otimes (\hat{n}\hat{n}) - (\hat{n}\hat{n})^T \otimes I \bigg] \notag \\
    &~~~~ + \gamma\bigg[ 2 b^\ast \otimes b - I\otimes (bb^\dagger) - (b^\ast b^T) \otimes I \notag \\
    &~~~~~~~~~~~~ + 2 b^T \otimes b^\dagger - I\otimes \hat{n} - \hat{n}^T \otimes I \bigg].
\end{align}
The Lindblian above is not time dependent and generally nonunitary. Simple Lindbladian can be directly diagonalized through a unitary transform $U$
\begin{align}
    e^{\tilde{\mathcal{L}}t} = e^{U \tilde{\mathcal{L}}_D U^\dagger t} = U e^{\tilde{\mathcal{L}}_D t} U^\dagger.
\end{align}
For more general Hamiltonians (including Dirac Hamiltonian for relativistic quantum simulations) and optical potentials, the nonunitary propagator $e^{\tilde{\mathcal{L}}t}$ can be expressed as a linear combination of unitary (LCU) operators $e^{\tilde{\mathcal{L}}t} = \sum_i c_i \hat{U}_i$,
and encode the linear combination of unitary propagates on quantum computers. One way to find the unitary basis is to express $e^{\tilde{\mathcal{L}}t}$ as the sum of a Hermitian $\mathcal{A}$ and an anti-Hermitian $\mathcal{B}$ operators, and approximate each of them using first-order Taylor expansion~\cite{Davids2021prl}
\begin{align}
    e^{\tilde{\mathcal{L}}t} = \mathcal{A} + \mathcal{B} 
\end{align}
with
\begin{align}
\begin{array}{l}
    \mathcal{A} = \frac{1}{2} ( e^{\tilde{\mathcal{L}}t} + e^{\tilde{\mathcal{L}^\dagger}t}) = (ie^{-i\epsilon\mathcal{A}}-ie^{i\epsilon\mathcal{A}})/2\epsilon + \mathcal{O}(\epsilon^2) \\
    \mathcal{B} = \frac{1}{2} ( e^{\tilde{\mathcal{L}}t} - e^{\tilde{\mathcal{L}^\dagger}t}) = (e^{\epsilon\mathcal{B}} -e^{-\epsilon\mathcal{B}})/2\epsilon + \mathcal{O}(\epsilon^2)
    \end{array} .
\end{align}
The four unitaries can be implemented separately or together on a dilated space. However, implementing each unitary and its associated Trotterized approximation may lead to a deep circuit structure, which warrants exploration of the optimal circuit structure through analytical and numerical means.

\section{Unitary Flow for Many-Body-Localization}\label{App:U_path}

We can define a unitary path $\hat{U}(s)$ depending on a parameter $s$ such that
for the Hamiltonian $\hat{H}$ of a given quantum system the following unitary transformation
\begin{align}
    \hat{H}'(s) = \hat{U}(s) \hat{H} \hat{U}^\dagger(s) \label{eq: upath}
\end{align}
can gradually reduce the band of $\hat{H}$, and repeating unitary transformations $n$ times
\begin{align}
    \hat{H}_D = \hat{U}_n \cdots \hat{U}_1 \hat{H} \hat{U}^\dagger_1 \cdots \hat{U}^\dagger_n \label{eq: uflow}
\end{align}
would give the diagonal form of the Hamiltonian $\hat{H}_D$ (in analogous to the Jacobi eigenvalue 
algorithm where a series of Given rotations are performed). The goal of seeking for an optimal $s$
to gradually reduce the band of $\hat{H}$ can be viewed as a optimization problem for which 
one would usually need to interrogate the corresponding gradient of $\hat{H}'$ with respect to $s$,
\begin{align}
    \frac{\partial \hat{H}'}{\partial s} = [\mathcal{G}(s), \hat{H}'(s)], \label{eq: grad}
\end{align}
with the anti-Hermitian generator $\mathcal{G}(s)$ defined as
\begin{align}
    \mathcal{G}(s)= \frac{\partial \hat{U}}{ \partial s}\hat{U}^\dagger. \label{eq: generator}
\end{align}
Eq. (\ref{eq: grad}) is also recognized as the most general form of a unitary flow on the Hamiltonian~\cite{bartlett2003flow}.

Among several unitary paths or their generators in this section, Wegner Generator~\cite{Deift_1983,Wegner1994} is one of the simplest, and is defined as
\begin{align}
    \mathcal{G}(s) = [\hat{H}'_D(s), \hat{H}'(s)], \label{eq: wegner}
\end{align}
with its matrix elements
\begin{align}
    \mathcal{G}_{ij} = \hat{H}'_{ij}(d_i - d_j) = \left\{ \begin{array}{cc}
      0   &  \text{if $i=j$}\\ h_{ij}(d_i-d_j)   & \text{if $i\neq j$}
    \end{array} \right. .
\end{align}
Here $d_i$'s and $h_{ij}$'s are diagonal and off-diagonal entries of $\hat{H}'$, respectively.
Note that both $d_i$'s and $h_{ij}$'s are $s$-dependent, their $s$-derivatives can be derived 
by plugging the Wegner generator (\ref{eq: wegner}) to the the unitary flow (\ref{eq: uflow}),
\begin{align}
    \left( \frac{\partial \hat{H}'}{\partial s} \right)_{ij} 
    &= \sum_{k} \mathcal{G}_{ik} \hat{H}'_{kj} -  \hat{H}'_{ik} \mathcal{G}_{kj} \notag \\
    &= \sum_{k} h_{ik} \hat{H}'_{kj} (d_i - d_k) - \hat{H}'_{ik} h_{kj}(d_k - d_j) \\
    &= -h_{ij} (d_i-d_j)^2 + \sum_{k\neq i,j} h_{ik} h_{kj} (d_i + d_j - 2d_k) \notag 
\end{align}
in which the diagonal elements reduce to 
\begin{align}
    \left( \frac{\partial \hat{H}'}{\partial s} \right)_{ii} &= \frac{\partial d_i}{\partial s} = 2\sum_{k} |h_{ik}|^2 (d_i-d_k)     
\end{align}
To demonstrate how the amplitudes of the $d_i$'s or $h_{ij}$'s evolve along the Wegner unitary flow, 
one can look at the derivatives of $d_i^2$'s or those of $|h_{ij}|^2$. The reason why we can only look
at one of these is because their sum is constantly zero due to the fact that 
\begin{align}
    &\frac{\partial}{\partial s}{\rm Tr}(\hat{H}'^2) \notag \\
    &= \frac{\partial}{\partial s}\sum_i \hat{H}_{ii}^{\prime2} = \frac{\partial}{\partial s}\sum_{i,k} \hat{H}'_{ik}\hat{H}'_{ki} 
     = \frac{\partial}{\partial s}\left( \sum_i d_i^2 + \sum_{i,j\neq i} |h_{ij}|^2 \right) \notag \\
    &= \frac{\partial}{\partial s}{\rm Tr}(\hat{U}_n\cdots\hat{U}_1\hat{H}^2\hat{U}^\dagger_1\cdots\hat{U}^\dagger_n) \notag \\
    &= \frac{\partial}{\partial s}{\rm Tr}(\hat{H}^2\hat{U}^\dagger_1\cdots\hat{U}^\dagger_n\hat{U}_n\cdots\hat{U}_1) \notag \\
    &= \frac{\partial}{\partial s}{\rm Tr}(\hat{H}^2) = 0,
\end{align}
where we utilize ${\rm Tr}(\mathbf{AB}) = {\rm Tr}(\mathbf{BA})$ and $\hat{U}^\dagger_i\hat{U}_i = \mathbf{1}$.
Since
\begin{align}
    \frac{\partial}{\partial s}\sum_i d_i^2 
    &= 2 \sum_i d_i \frac{\partial d_i}{\partial s} = 2 \sum_{i,k} |h_{ik}|^2 (2d_i^2 - 2d_id_k) \notag \\
    &= 2 \underbrace{\sum_{i,k} |h_{ik}|^2 (d_i^2 + d_k^2 - 2d_id_k)}_{\text{\scriptsize indices $i$ and $k$ are interchangable here}} \notag \\
    &= 2 \sum_{i,k} |h_{ik}|^2 (d_i - d_k)^2 \ge 0 \notag \\
\Rightarrow  \frac{\partial}{\partial s} \sum_{i,j\neq i} |h_{ij}|^2 &\le 0,
\end{align}
the amplitudes of the diagonal (off-diagonal) elements of $\hat{H}$ will monotonically increase (decrease) as the flow evolves.

It's worth mentioning that for certain types of Hamiltonians, e.g. number-conserving quadratic Hamiltonian 
\begin{align}
    \hat{H} = \sum_{ij} H_{ij} a_i^\dagger a_j
\end{align}
where the matrix $H$ must be Hermitian for the operator $\hat{H}$ to also be Hermitian, and $a^\dagger_i$ and $a_i$ can be either fermionic or bosonic creation and annihilation operators, the Hermitian matrix $H$ can be diagonalized by a unitary transformation $U$ with eigenvalues $E=\{\epsilon_1,\cdots,\epsilon_n\}$. This transformation also applies to the creation and annihilation operators for the diagonalization of the operator $\hat{H}$, i.e.
\begin{align}
    &U^\dagger H U = E,~~a_i^\dagger = \sum_m \alpha_m^\dagger (U^\dagger)_{mj}, ~~ a_i = \sum_l U_{im} \alpha_m, \notag \\
    \Rightarrow ~~&
    \hat{H} = \sum_{ij} \sum_{mn} \alpha_m^\dagger U^\dagger_{mi} H_{ij} U_{jn} \alpha_n = \sum_n \epsilon_n \alpha_n^\dagger \alpha_n.
\end{align}
One method to identify such a unitary transformation $U$ is the well-known Fourier-Bogoliiubov transformation method. To see how it works, we can first take the fermionic Hamiltonian of XY model (with an external field in $z$ direction) as an example. The Hamiltonian reads
\begin{align}
    \hat{H} = -J \sum_j \left[
    \frac{1+\gamma}{2}X_j X_{j+1} + \frac{1-\gamma}{2}Y_j Y_{j+1} + \lambda Z_j
    \right]
\end{align}
whose fermionic form is 
\begin{align}
    \hat{H} &= - J \sum_j \left[
    f_{j+1}^\dagger f_j + f_j^\dagger f_{j+1} + \gamma(f_{j+1}f_j + f_j^\dagger f_{j+1}^\dagger) \right. \\
    &~~~~\left.- 2\lambda f_j^\dagger f_j + \lambda \right] \notag
\end{align}
with $J>0$, $\gamma$, and $\lambda$ being scalars. Note that one can also use hard-core bosonic operator to get the same form of the Hamiltonian. 
Employing the \textbf{Fourier transform} of the (boson or fermion) creation/annihilation operators
\begin{align}
    f_{j}^\dagger &= \frac{1}{\sqrt{N}} \sum_{k} e^{-ikja} f_k^\dagger, ~~~~ 
    f_{j} = \frac{1}{\sqrt{N}} \sum_{k} e^{+ikja} f_k 
\end{align}
with $a=\frac{2\pi}{N}$ and
\begin{align}
    k=\left\{\begin{array}{ll}
    -\frac{N-1}{2},\cdots,-1,0,1,\cdots,\frac{N-1}{2}, &~~\text{if $N$ is odd} \\
    -\frac{N}{2},\cdots,-1,0,1,\cdots,\frac{N}{2}-1, &~~\text{if $N$ is even}
    \end{array} \right. , \notag 
\end{align}
and utilizing the relation 
\begin{align}
    N^{-1} \sum_{j=1}^N e^{i(k+k')ja} = \delta_{k,-k'},
\end{align}
we can equivalently rewrite the Hamiltonian in the $k$-space, i.e.
\begin{align}
    \hat{H} = - J \sum_k &\left[
    2(\cos(ka)-\lambda) f_k^\dagger f_k + \gamma (f_{k} f_{-k} + f_k^\dagger f_{-k}^\dagger) e^{ika} \right. \notag \\
    & ~~+ \lambda  \Big] \notag \\
    = - J \sum_k &\left[
    (\cos(ka)-\lambda) (f_k^\dagger f_k + f_{-k}^\dagger f_{-k} ) \right. \notag \\
    & ~~+ i\gamma \sin(ka) (f_{k} f_{-k} + f_k^\dagger f_{-k}^\dagger) + \lambda
    \Big]. \label{eq: FT_Ham}
\end{align}
where the second line is due to the fact that $k$ is distributed symmetrically around zero, e.g., the second term can be rewritten as
\begin{align}
    &\sum_k (f_{k} f_{-k} + f_k^\dagger f_{-k}^\dagger) e^{ika} \notag \\
    &~~= \sum_k (f_{k} f_{-k} + f_{-k}^\dagger f_{k}^\dagger) e^{-ika} = -\sum_k (f_{k} f_{-k} + f_k^\dagger f_{-k}^\dagger) e^{-ika} \notag \\
    &~~= \frac{1}{2} \sum_k (f_{k} f_{-k} + f_k^\dagger f_{-k}^\dagger) (e^{ika} - e^{-ika}) \notag \\
    &~~= i\sum_k \sin(ka) (f_{k} f_{-k} + f_k^\dagger f_{-k}^\dagger)
\end{align}
Note that the above coupling between $k$ and $-k$ can be removed through the \textbf{Bogoliubov transformation}, another common tool often used to diagonalize fermonic and/or bosonic Hamiltonians. To see that, we can define
\begin{align}
    g^\dagger_k &= u_k f^\dagger_k + iv_k f_{-k}, ~~~~
    g^\dagger_{-k} = u_k f^\dagger_{-k} - iv_k f_{k}, \notag \\
    g_k &= u_k f_k -i v_k f^\dagger_{-k}, ~~~~
    g_{-k} = u_k f_{-k} +i v_k f^\dagger_{k},
\end{align}
with $u_k = \cos \theta_k$, $v_k = \sin \theta_k (\theta_k \in \mathbf{R})$, and $u_k^2 + v_k^2 = 1$.
It's then easy to show that
\begin{align}
    f_k^\dagger &= u_k g_k^\dagger - iv_k g_{-k}, ~~~~
    f_{-k}^\dagger = u_k g_{-k}^\dagger +iv_k g_{k},  \notag \\
    f_k &= u_k g_k +i v_k g_{-k}^\dagger, ~~~~
    f_{-k} = u_k g_{-k} -i v_k g_{k}^\dagger. \label{eq: bogoliubov}
\end{align}
Plugging (\ref{eq: bogoliubov}) to (\ref{eq: FT_Ham}) and denoting
\begin{align}
    \Delta_k = \gamma\sin(ka),~~ \epsilon_k = \lambda-\cos(ka),~~E_k = \sqrt{\Delta_k^2+ \epsilon_k^2}
\end{align}
we then have 
\begin{align}
    \hat{H} = J \sum_k& \left\{
    \left[ \epsilon_k (u_k^2 - v_k^2) + 2 \Delta_k u_k v_k \right] (g_k^\dagger g_k + g_{-k}^\dagger g_{-k}) \right. \notag \\
    &~~ + 2v_k^2\epsilon_k - 2u_kv_k\Delta_k \notag \\
    &~~ +\left[ 2i \epsilon_k u_k v_k - i \Delta_k (u_k^2 - v_k^2) \right] (g_k^\dagger g_{-k}^\dagger + g_{k}g_{-k}) \notag \\
    &~~ - \lambda \Big\}.
\end{align}
The full diagonalization then requires the $(g_k^\dagger g_{-k}^\dagger + g_{-k}g_{k})$ term in the above equation to disappear which means $u_k$ and $v_k$ need to be chosen such that
\begin{align}
    &2i \epsilon_k u_k v_k - i \Delta_k (u_k^2 - v_k^2) = 0, \notag \\
    \Rightarrow ~~&
    \tan 2\theta_k = \frac{2u_k v_k}{u_k^2 - v_k^2} = \frac{\Delta_k}{\epsilon_k},\notag \\ 
     & u_k^2 - v_k^2 = \pm \epsilon_k/E_k, ~~
     u_kv_k = \pm \Delta_k/2E_k. \notag
\end{align}
By choosing the signs of $u_k^2 - v_k^2$ and $u_k v_k$, the Hamiltonian is simplified to
\begin{align}
    \hat{H} &= J \sum_k \left( 2E_k g_k^\dagger g_k + 2v_k^2\epsilon_k - 2u_kv_k\Delta_k - \lambda  \right) \notag \\
    &= J \sum_k \left( 2E_k g_k^\dagger g_k + \epsilon_k - E_k - \lambda  \right) \notag \\
    &= J \sum_k \left( 2E_k g_k^\dagger g_k + \cos (ka) - E_k \right) \notag \\
    &= J \sum_k 2E_k  \left( g_k^\dagger g_k  - \frac{1}{2} \right).
\end{align}
where we also utilize the fact $\sum_k \cos(ka) = 0$.

\noindent Bogoliubov transformation can be also used for bosonic system, though the unitary matrix becomes different. Take the diagonalization of the following two-site Hamiltonian as an example 
\begin{align}
    \hat{H} &= \epsilon (c_1^\dagger c_1 + c_2^\dagger c_2) + \lambda (c_1^\dagger c_2^\dagger + c_2 c_1) \notag \\
    &= \frac{1}{2} \left( \begin{array}{cccc}
    c_1^\dagger & c_2 & c_2^\dagger & c_1
    \end{array} \right) H \left( \begin{array}{cc}
    c_1 \\ c_2^\dagger \\ c_2 \\ c_1^\dagger
    \end{array} \right) + \epsilon 
\end{align}
where $\epsilon,\lambda \in \mathbf{R}$ and $H$ is a block-diagonal matrix
\begin{align}
    H = \left( \begin{array}{cccc}
    \epsilon & \lambda   & 0 & 0 \\
    \lambda  & -\epsilon & 0 & 0 \\
    0 & 0 & \epsilon & -\lambda \\ 
    0 & 0 & -\lambda & -\epsilon 
    \end{array}\right).
\end{align}
If $c^\dagger$ and $c$ are fermionic creation and annihilation operators, the fermionic Bogoliubov transformation is 
\begin{align}
    \left( \begin{array}{c}
    c_1 \\ c_2^\dagger \\ c_2 \\ c_1^\dagger
    \end{array} \right) = 
    U_f \left( \begin{array}{c}
    d_1 \\ d_2^\dagger \\ d_2 \\ d_1^\dagger
    \end{array} \right),~~
    U_f = \left( \begin{array}{cccc}
     u & v & 0 &  0 \\ 
    -v & u & 0 &  0 \\
     0 & 0 & u & -v \\
     0 & 0 & v &  u \\
    \end{array} \right)
\end{align}
with the typical choices of $u$  and $v$ being
\begin{align}
     u = \cos \theta,~~ v = \sin \theta,~~\text{s.t.}~~ u^2+v^2 = 1.
\end{align}
When $\tan 2\theta = -\lambda/\epsilon$, $U_f$ diagonalizes $H$ with eigenvalue $\tilde{\epsilon} = \pm \sqrt{\epsilon^2+\lambda^2}$, and the Hamiltonian has the form
\begin{align}
    \hat{H} = \tilde{\epsilon}(d_1^\dagger d_1 + d_2^\dagger d_2) - \epsilon + \tilde{\epsilon}.
\end{align}
If $c^\dagger$ and $c$ are bosonic creation and annihilation operators, then the bonic Bogoliubov transformation is 
\begin{align}
    \left( \begin{array}{c}
    c_1 \\ c_2^\dagger \\ c_2 \\ c_1^\dagger
    \end{array} \right) = 
    U_b \left( \begin{array}{c}
    d_1 \\ d_2^\dagger \\ d_2 \\ d_1^\dagger
    \end{array} \right),~~
    U_b = \left( \begin{array}{cccc}
     u & v & 0 &  0 \\ 
     v & u & 0 &  0 \\
     0 & 0 & u &  v \\
     0 & 0 & v &  u \\
    \end{array} \right)
\end{align}
with the typical choices of $u$  and $v$ being
\begin{align}
     u = \cosh \theta,~~ v = \sinh \theta,~~\text{s.t.}~~ u^2-v^2 = 1,
\end{align}
When $\tanh 2\theta = -\lambda/\epsilon$, $U_b$ diagonalizes $H$ with eigenvalue $\tilde{\epsilon} = \sqrt{\epsilon^2-\lambda^2}$ implying that $|\epsilon| > |\lambda|$ (otherwise the Hamiltonian would be representing a system at an unstable equilibrium point).

%\textcolor{black}{
%1. FFFT and circuit \\
%
%2. Givens rotation
%}

\section{Trotterized D-UCCSD Ans\"{a}tz for Three-Level Two-Boson Model System}\label{DUCC_boson}

For a three-level two-boson model for which the excitation manifold is partitioned as depicted in Figure \ref{fig:boson_cc}a, the direct D-UCCSD ans\"{a}tz can be written as
\begin{align}
    |\phi\rangle = e^{\sigma_{\#2}}e^{\sigma_{\#1}}|\phi_0\rangle \label{App2:ducc}
\end{align}
with the reference $|\phi_0\rangle$ being denoted by $|200\rangle$ indicating the boson occupation at each level and
\begin{align}
    \sigma_{\#1} &= r_1 \left[ (b_1^\dagger)^2 (b_0)^2 - (b_0^\dagger)^2 (b_1)^2 \right] + r_2 \left( b_1^\dagger b_0 - b_0^\dagger b_1 \right) \notag \\
    \sigma_{\#2} &= s_1 \left[ (b_2^\dagger)^2 (b_0)^2 - (b_0^\dagger)^2 (b_2)^2 \right] + s_2  \left[ b_2^\dagger b_1^\dagger (b_0)^2 \right. \notag \\
    &~~~~ \left. - (b_0^\dagger)^2 b_1 b_2  \right] + s_3 \left( b_2^\dagger b_0 - b_0^\dagger b_2 \right).
\end{align}

Nevertheless, it is challenging to direct implement the unitaries $e^{\sigma_{\#1}}$ and $e^{\sigma_{\#2}}$. In practice, to ease the implementation, we typically employ the Trotterized forms of (\ref{App2:ducc}) with certain ordering to ensure the exactness. For example, consider the following ans\"{a}tz
\begin{align}
    |\phi\rangle = \prod_{i=3}^1 e^{\sigma_{\#2}^i}\prod_{j=2}^1 e^{\sigma_{\#1}^j}|\phi_0\rangle \label{App2:ducc}
\end{align}
where
\begin{align}
    \sigma_{\#1}^1 &= r_1 \left[ (b_1^\dagger)^2 (b_0)^2 - (b_0^\dagger)^2 (b_1)^2 \right], \notag \\
    \sigma_{\#1}^2 &= r_2 \left( b_1^\dagger b_0 - b_0^\dagger b_1 \right), \notag \\
    \sigma_{\#2}^1 &= s_1 \left[ (b_2^\dagger)^2 (b_0)^2 - (b_0^\dagger)^2 (b_2)^2 \right], \notag \\
    \sigma_{\#2}^2 &= s_2  \left[ b_2^\dagger b_1^\dagger (b_0)^2 - (b_0^\dagger)^2 b_1 b_2  \right], \notag \\ 
    \sigma_{\#2}^3 &= s_3 \left( b_2^\dagger b_0 - b_0^\dagger b_2 \right).
\end{align}
Now the question is whether such a representation can represent an arbitrary state $|\Psi\rangle$ for the three-level two-boson model. To see that, we need to check (\textbf{a}) if the above ans\"{a}tz can generate all the possible configurations of the model, and (\textbf{b}) if $\{r_1,r_2,s_1,s_2,s_3\}$ can be determined for $|\Psi\rangle$. 

Regarding (\textbf{a}), realizing that the each exponential operator in (\ref{App2:ducc}) is a Givens rotation, and utilizing the following algebraic relations
\begin{align}
    e^{\sigma_{\#1}^1}|200\rangle &= c_{2r_1} |200\rangle + s_{2r_1}|020\rangle, \label{App2:algebra1} \\
    e^{\sigma_{\#1}^2}|200\rangle &= c_{r_2}^2 |200\rangle + \sqrt{2}s_{r_2}c_{r_2}|110\rangle + s_{r_1}^2 |020\rangle, \\
    e^{\sigma_{\#1}^2}|020\rangle &= s_{r_2}^2 |200\rangle - \sqrt{2}s_{r_2}c_{r_2}|110\rangle + c_{r_2}^2 |020\rangle, \\
    e^{\sigma_{\#2}^1}|200\rangle &= c_{2s_1} |200\rangle + s_{2s_1}|002\rangle, \\
    e^{\sigma_{\#2}^1}|110\rangle &= |110\rangle, \\
    e^{\sigma_{\#2}^1}|020\rangle &= |020\rangle, \\
    e^{\sigma_{\#2}^2}|200\rangle &= c_{\sqrt{2}s_2} |200\rangle + s_{\sqrt{2}s_2}|011\rangle, \\
    e^{\sigma_{\#2}^2}|002\rangle &= |002\rangle, \\
    e^{\sigma_{\#2}^2}|110\rangle &= |110\rangle, \\
    e^{\sigma_{\#2}^2}|020\rangle &= |020\rangle, \\
    e^{\sigma_{\#2}^3}|200\rangle &= c_{s_3}^2 |200\rangle + \sqrt{2}s_{s_3}c_{s_3}|101\rangle + s_{s_3}^2 |002\rangle, \\
    e^{\sigma_{\#2}^3}|002\rangle &= s_{s_3}^2 |200\rangle - \sqrt{2}s_{s_3}c_{s_3}|101\rangle + c_{s_3}^2 |002\rangle, \\
    e^{\sigma_{\#2}^3}|110\rangle &= c_{s_3} |110\rangle + s_{s_3}|011\rangle, \\
    e^{\sigma_{\#2}^3}|011\rangle &= c_{s_3} |011\rangle - s_{s_3}|110\rangle, \label{App2:algebra2}
\end{align}
where $c_p$ and $s_p$ denote $\cos(p)$ and $\sin(p)$, respectively. It is straightforward to see the all the possible configurations of the model can be mapped out through the consecutive Givens rotations (see Figure \ref{fig:conf}).

\begin{figure}
    \centering
    \includegraphics[width=\linewidth]{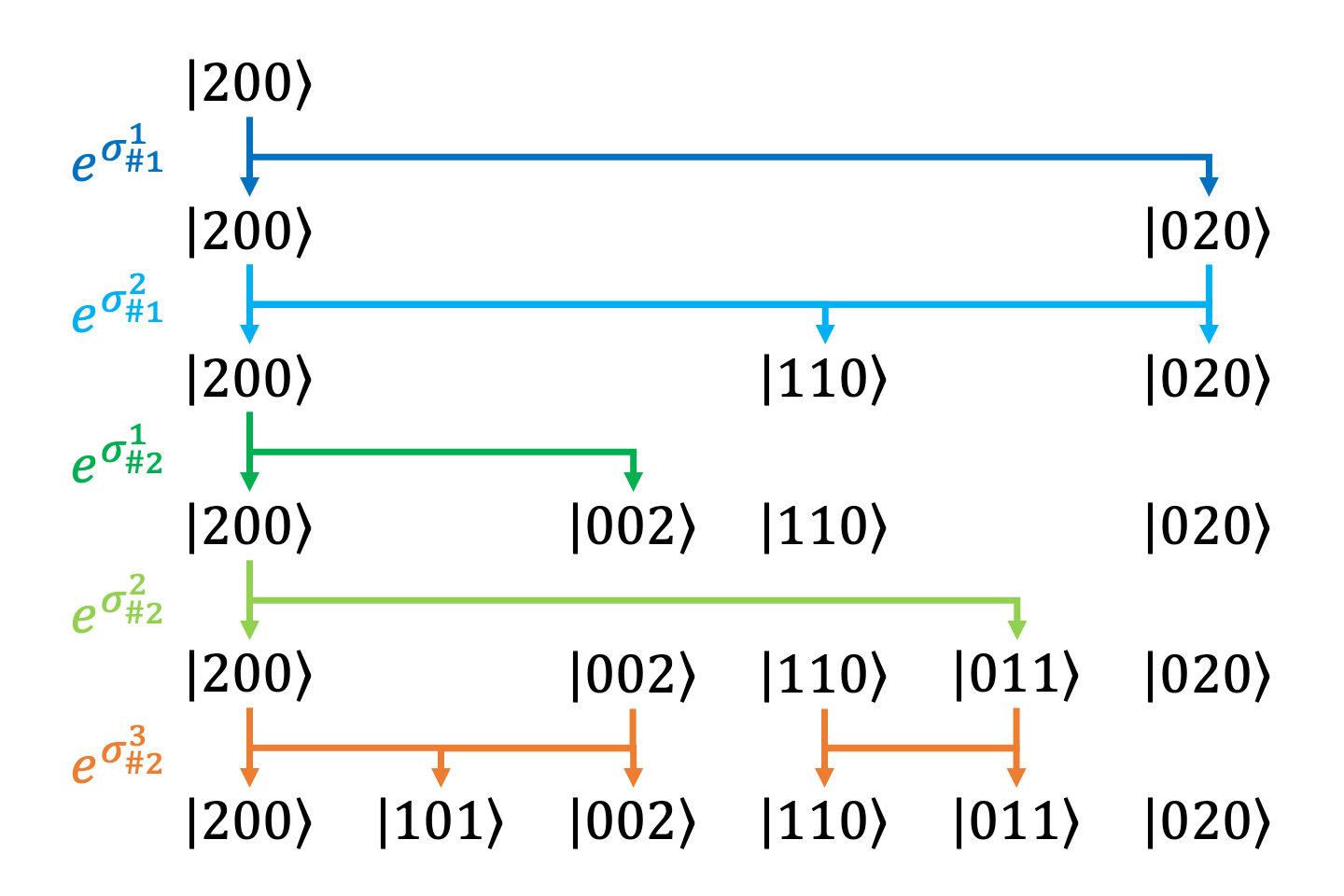}
    \caption{The expansion of the configuration space when sequentially applying the Givens rotation in ans\"{a}tz (\ref{App2:ducc}) to the reference state $|200\rangle$.}
    \label{fig:conf}
\end{figure}

Regarding (\textbf{b}), an arbitrary state $|\Psi\rangle$ can be defined as a linear combination of all six configurations
\begin{align}
    |\Psi\rangle &= d_1 |200\rangle + d_2 |101\rangle + d_3 |002\rangle \notag \\
    &~~~~ + d_4 |110\rangle + d_5 | 011 \rangle + d_6 | 020 \rangle.
\end{align}
with arbitrary $\{d_i\}$ ($i=1,\cdots,6$) satisfying $\sum_{i=1}^6 d_i^2 = 1$. To determine $\{r_1,r_2,s_1,s_2,s_3\}$, we can utilize the reverse flow of Figure \ref{fig:conf} (i.e. from $|\Psi\rangle$ to $|\phi\rangle$), in particular the fact that only one configuration disappears after each reverse Givens rotation, to solve the rotations one by one. We can take the first two step as examples. The first step applies $e^{-\sigma_{\#2}^3}$ on $|\Psi\rangle$ to remove $|101\rangle$. Since the only relevant equations are 
\begin{align}
    d_1 \cdot e^{-\sigma_{\#2}^3}|200\rangle &= d_1 \cdot ( c_{s_3}^2 |200\rangle - \sqrt{2}s_{s_3}c_{s_3}|101\rangle \notag \\
    &~~~~ + s_{s_3}^2 |002\rangle ), \notag \\
    d_2 \cdot e^{-\sigma_{\#2}^3}|101\rangle &= d_2 \cdot ( \sqrt{2}s_{s_3}c_{s_3} |200\rangle + c_{2s_3}|101\rangle \notag \\
    &~~~~ - \sqrt{2}s_{s_3}c_{s_3} |002\rangle ), \notag \\
    d_3 \cdot e^{-\sigma_{\#2}^3}|002\rangle &= d_3 \cdot ( s_{s_3}^2 |200\rangle + \sqrt{2}s_{s_3}c_{s_3}|101\rangle \notag \\
    &~~~~ + c_{s_3}^2 |002\rangle ), 
\end{align}
we can then write
\begin{align}
    &c_{2s_3}d_2 - \sqrt{2}s_{s_3}c_{s_3} (d_1 - d_3) = 0,
\end{align}
from which if $d_2\neq 0$, we can choose one non-singular solution to be
\begin{align}
    s_3 = \tan^{-1} \left( \frac{-\sqrt{2}(d_1-d_3)\pm\sqrt{\Delta}}{2d_2} \right) 
\end{align}
with $\Delta = 2(d_1-d_3)^2 + 4d_2^2 \ge 0$. Then $d_1,d_3,d_4,d_5$ will be updated to $d'_1,d'_3,d'_4,d'_5$ ($d_6$ stays the same) before proceeding to the second step. In the second step, the operation $e^{-\sigma_{\#2}^2}$ removes the configuration $|011\rangle$. Since the only relevant equations are 
\begin{align}
    d'_1 \cdot e^{-\sigma_{\#2}^2}|200\rangle &= d'_1 \cdot ( c_{\sqrt{2}s_2} |200\rangle - s_{\sqrt{2}s_2}|011\rangle ), \notag \\
    d'_5 \cdot e^{-\sigma_{\#2}^2}|011\rangle &= d'_5 \cdot ( s_{\sqrt{2}s_2} |200\rangle + c_{\sqrt{2}s_2}|011\rangle ), 
\end{align} 
from which the parameter $s_2$ can be chosen to be
\begin{align}
    s_2 = \frac{1}{\sqrt{2}}\tan^{-1} \left( \frac{d'_5}{d'_1} \right).
\end{align}
Following the same procedure we can also find the solutions for the remaining parameters. Therefore, both (\textbf{a}) and (\textbf{b}) are satisfied meaning the ans\"{a}tz (\ref{App2:ducc}) is able to represent arbitrary state of the model.

Note that the procedure described above for addressing (\textbf{b}) is similar to the one that is utilized to prove the exactness of the disentangled UCC ans\"{a}tze for fermionic systems (see Ref. \citenum{ExactUCC_19}). Although the explicit Givens rotations, e.g. (\ref{App2:algebra1})-(\ref{App2:algebra2}), are different between fermionic and bosonic systems, there are some common effects after applying a series single and/or double Givens rotations on a fermionic or bosonic state. For example, with certain ordering of these rotations, one can successively generate/eliminate configurations. Regarding the higher order rotations dealing with more bosons, given the fact that the high order excitation can be rewritten as a series of nested commutators of one-body and two-body excitations,\cite{ExactUCC_19} high order Givens rotations can be approximated as the products of one-body and two-body Givens rotations to an arbitrary accuracy. Therefore, one can  generalize the above procedure to generate a Trotterized D-UCCSD ans\"{a}tz to represent any state for a general bosonic system to an arbitrary accuracy.

\bibliography{refs}
\end{document}